\documentclass[a4paper]{aart}
\usepackage[centertags]{amsmath}
\usepackage{amsmath}
\usepackage[dutch,british]{babel}

\usepackage{amsfonts}
\usepackage{amssymb}
\usepackage{amsthm}
\usepackage{import}

\usepackage{newlfont}
\usepackage{color}

 \usepackage{bbm}
\include{amsthm_sc}

\usepackage{stmaryrd}
\usepackage{mathrsfs}
\usepackage{pstricks,pst-node}

\usepackage{fancyhdr}
\usepackage{graphicx}
\usepackage{fancybox}

\usepackage{geometry}
\usepackage{setspace}

\theoremstyle{plain}
  \newtheorem{theorem}{Theorem}[section]

  \newtheorem{lemma}[theorem]{Lemma}
    \newtheorem{assumption}{Assumption}
  
\theoremstyle{definition}

\theoremstyle{remark}

\numberwithin{equation}{section}
\numberwithin{figure}{section}

\newcommand{\deq}{\mathrel{\mathop:}=}
\newcommand{\e}[1]{\mathrm{e}^{#1}}
\newcommand{\R} {\mathbb{R}}
\newcommand{\C} {\mathbb{C}}
\newcommand{\N} {\mathbb{N}}
\newcommand{\Z} {\mathbb{Z}}

\newcommand{\adj}{^*} 
\newcommand{\ad}{\mathrm{ad}}
\newcommand{\bra}[1]{\langle #1 |}

\newcommand{\ket}[1]{| #1 \rangle}

\DeclareMathOperator{\Tr}{Tr}

\DeclareMathOperator{\re}{\mathrm{Re}}
\DeclareMathOperator{\im}{\mathrm{Im}}

\newcommand{\caA}{{\mathcal A}}

\newcommand{\caC}{{\mathcal C}}
\newcommand{\caD}{{\mathcal D}}
\newcommand{\caE}{{\mathcal E}}

\newcommand{\caG}{{\mathcal G}}

\newcommand{\caJ}{{\mathcal J}}
\newcommand{\caK}{{\mathcal K}}
\newcommand{\caL}{{\mathcal L}}
\newcommand{\caM}{{\mathcal M}}

\newcommand{\caO}{{\mathcal O}}
\newcommand{\caP}{{\mathcal P}}
\newcommand{\caQ}{{\mathcal Q}}
\newcommand{\caR}{{\mathcal R}}
\newcommand{\caS}{{\mathcal S}}
\newcommand{\caT}{{\mathcal T}}
\newcommand{\caU}{{\mathcal U}}
\newcommand{\caV}{{\mathcal V}}
\newcommand{\caW}{{\mathcal W}}

\newcommand{\caY}{{\mathcal Y}}
\newcommand{\caZ}{{\mathcal Z}}

\newcommand{\bbC}{{\mathbb C}}

\newcommand{\bbH}{{\mathbb H}}

\newcommand{\bbN}{{\mathbb N}}

\newcommand{\bbR}{{\mathbb R}}

\newcommand{\bbT}{{\mathbb T}}

\newcommand{\bbV}{{\mathbb V}}

\newcommand{\bbZ}{{\mathbb Z}}

\newcommand{\opunit}{\text{1}\kern-0.22em\text{l}}

\newcommand{\scrB}{{\mathscr B}}

\newcommand{\scrE}{{\mathscr E}}

\newcommand{\scrH}{{\mathscr H}}

\renewcommand{\d}{{\mathrm d}}

\newcommand{\Dom}{\mathrm{Dom}}

\newcommand{\beq}{ \begin{equation} }
\newcommand{\eeq}{ \end{equation} }

\newcommand{\baq}{ \begin{eqnarray} }
\newcommand{\eaq}{ \end{eqnarray} }
\newcommand{\bal}{  \begin{align} }
\newcommand{\eal}{\end{align}    }
\newcommand{\bet}{ \begin{theorem} }
\newcommand{\eet}{ \end{theorem} }

\newcommand{\bosondispersion}{\om}

\newcommand{\str}{ |}

\newcommand{\links}{\mathrm{l}}
\newcommand{\rechts}{\mathrm{r}}

\newcommand{\uv}{\underline{v}}
\newcommand{\uu}{\underline{u}}
\newcommand{\uw}{\underline{w}}
\newcommand{\ut}{\underline{t}}
\newcommand{\ux}{\underline{x}}
\newcommand{\uvs}{\underline{\vs}}

\newcommand{\adjoint}{\ad}
\newcommand{\norm}{ \|}

\newcommand{\lone}{\mathbbm{1}} 
\newcommand{\irr}{_{\mathrm{irr.}}}
 
 \let\be=\beta  

\let\ve=\varepsilon  \let\ga=\gamma 
\let\ka=\kappa \let\la=\lambda \let\om=\omega 
\let\si=\sigma \let\vs=\varsigma

 \let\Ga=\Gamma \let\La=\Lambda

\newcommand{\tdl}{\mathop{\lim}\hspace{0.15cm}\hspace{-0.22cm}_{\Lambda}}

\newcommand{\ben}{\begin{arabicenumerate}}
\newcommand{\een}{\end{arabicenumerate}}

\newcommand{\sys}{\mathrm{S}}
\newcommand{\res}{\mathrm{R}}

\newcommand{\lat}{ \mathbb{Z}^d }
\newcommand{\tor}{ {\mathbb{T}^d}  }

\newcommand{\dd}{\mathrm{d}}
\newcommand{\ii}{\mathrm{i}}

\newcommand{\refer}{{\beta}}
\newcommand{\referres}{{\res,\beta}}

\newcommand{\matz}{{\mathcal{Z}_{[0,t]}^{\Lambda}}}

\newcommand{\Matz}{\mathcal{Z}_{[0,t]}}

\newcommand{\kin}{_{M}}
\newcommand{\kini}{_{\widetilde M}}

\newcommand{\eig}{u}

\newcommand{\field}{\chi}
\newcommand{\infi}{{t_{-}(I)}}
\newcommand{\supi}{{t_{+}(I)}}

 \long\def\symbolfootnote[#1]#2{\begingroup%
 \def\thefootnote{\fnsymbol{footnote}}\footnote[#1]{#2}\endgroup}

\begin{document}

\begin{center}
\large{ \bf{Quantum Diffusion with Drift and the Einstein Relation II}} \\
\vspace{20pt} \normalsize

\vspace{10pt} 
{\bf   W.\ De Roeck }\\
\vspace{10pt} 
{\bf   Institute for Theoretical Physics \\
Universit\"at Heidelberg\\ 
 69120 Heidelberg, Germany}

\vspace{20pt}

{\bf   J. Fr\"ohlich\symbolfootnote[2]{
Present address: School of Mathematics, The Institute for Advanced Study, Princeton, NJ 08540, United States.  The stay of J.~F.\ at the IAS is supported by
`The Fund for Math' and `The Monell Foundation'.
}   }\\
\vspace{10pt} 
{\bf   Institute for Theoretical Physics \\
ETH Z\"urich \\
8093 Z\"urich, Switzerland}
\vspace{20pt} 

{\bf K.\ Schnelli  }

\vspace{10pt} 
{\bf   Department of Mathematics \\
Harvard University \\
Cambridge MA 02138, United States}

\vspace{15pt} \normalsize

\end{center}

\vspace{20pt} \footnotesize{ \noindent {\bf Abstract}: }
This paper is a companion to~\cite{paper1}.  Its purpose is to describe and prove
a certain number of technical results used in~\cite{paper1}, but not proven
there.   Both papers concern long-time properties (diffusion, drift) of the motion of a driven quantum particle coupled to an array of thermal reservoirs.  
The main technical results derived in the present paper are $(1)$ an asymptotic perturbation theory applicable for small driving, and, $(2)$ the construction of time-dependent correlation functions of particle observables.

\vspace{5pt} \footnotesize \noindent {\bf KEY WORDS:}  diffusion, kinetic limit,
quantum brownian motion  \vspace{20pt}
\normalsize
\section{Introduction}\label{sec: intro}
 In~\cite{paper1} and in the present paper, we study a model of a quantum tracer particle
hopping on the lattice $\mathbb{Z}^d$ ($d \geq 1$) and
interacting with an array of thermal reservoirs placed at the sites of
$\mathbb{Z}^{d}$, which are described quantum-mechanically: When the particle visits a
site $x \in \mathbb{Z}^d$ it can emit or absorb quanta of the thermal reservoir
at site $x$, thus changing its own momentum. Those quanta correspond to
non-interacting, massless modes -- phonons or photons -- in a state of thermal
equilibrium at a positive temperature $T=\beta^{-1}$, the same for all the
reservoirs. Reservoirs at different sites are independent of one another. As a
consequence, memory effects only arise when the particle returns
to sites it has visited previously. At positive temperature and assuming a certain analyticity property of the particle-reservoir coupling, such memory effects turn out to decay exponentially fast in time.
There is a constant external force acting on the tracer particle. The particular
model described here has first been studied in~\cite{deroeckfrohlichpizzo}, for
a {\it vanishing} external force, and it has been proven there that the particle
exhibits ``quantum Brownian motion''; (see 
also~\cite{pilletdebievre,KangSchenker} for related results). The purpose of the
analysis presented in~\cite{paper1} and in this paper is to determine asymptotic
properties of the motion of the particle in the presence of a
\emph{non-vanishing} external force, as time $t$ tends to $\infty$. As announced
in~\cite{paper1}, one expects that, for a sufficiently weak external force and
weak particle-reservoir interactions, the state of the quantum particle (that is, a state on the algebra generated by functions of its momentum)
approaches a ``non-equilibrium steady state'' (NESS). This state describes a
uniform motion of the tracer particle with a mean drift velocity $v$ that
depends on the force pushing it, the strength of interaction of the particle
with the reservoirs and the temperature of the reservoirs. Furthermore, the
particle exhibits diffusion -- ``quantum Brownian motion'' -- around its mean
motion. The diffusion constant at zero force is related to the derivative, at
zero external force, of the 
drift velocity with respect to the force by the \emph{Einstein relation}.

One major difficulty encountered in our analysis is that, for a \emph{non-zero}
external force, the state of the particle (a density matrix on the particle
Hilbert space obtained by tracing out the degrees of freedom of all the
reservoirs) appears to show phase coherence over fairly long distance scales. In
contrast, when the external force vanishes, the state of the particle exhibits
exponential decoherence in particle position space. Because of the lack of
decoherence in the presence of an external force our result concerning the
approach of the particle state to a NESS only holds in an ergodic mean, and our
formula for the diffusion constant involves an abelian limit.

 While the material
in~\cite{paper1} is quite elegant and of some general conceptual interest, the
present paper is primarily technical. Most technicalities to be coped with in 
the following are related to the thermodynamic limit and to the need to control
large-time asymptotics, in particular cluster properties of connected
correlation functions of particle observables, as time differences tend to
$\infty$. The need to introduce finite-volume approximations of the system
studied in our papers comes from the structure of our proof of the Einstein
relation. Technical difficulties arise because the components of the particle
position are unbounded operators and the potential of a constant force grows
linearly in the particle position.

Background from physics and an account of the difficulties encountered in the
analysis of this and of more realistic models of quantum transport have been
presented in~\cite{paper1} and will not be repeated here. Suffice it to say that
we expect that one would face major problems if one attempted to extend the
results of~\cite{paper1} and of the present paper to more realistic models of
particle transport -- in particular, continuum models -- or if one tried to
prove stronger convergence results for the model studied here.

The organization of our paper is as follows:\\
In Section~2, we recall some notation and the definition of the model studied
in~\cite{paper1} and in this paper. We then define the effective quantum
dynamics of the particle, after tracing out the degrees of freedom of the
thermal reservoirs, and we introduce correlation (Green) functions of operators
representing properties of the particle (particle observables) in various states
of the system. We define the time-reversal operator and recall the KMS condition
characterizing a thermal equilibrium state. Equilibrium states at vanishing
external force appear in our proof of the Einstein relations.

In Section~3, the main assumptions concerning our model are summarized, and the
main results proven in~\cite{paper1} and in this paper are stated. All results
are only proven for a weak external force and at weak coupling of the tracer
particle to the reservoirs, (where ``weak'' depends on the temperature of the
reservoirs and the kinetic energy operator of the particle). In Lemmas~3.1 and
3.2 we state our results on the existence and some properties of the
thermodynamic limit. Theorem~3.3 describes ``return to equilibrium'' at
vanishing external force. Theorem~3.4 concerns the approach, in an ergodic mean,
of the state of the tracer particle to a ``non-equilibrium stationary (or
steady) state'' (NESS). In Theorem~3.5, the convergence of the mean velocity of
the particle to a non-zero drift velocity, as time $t$ tends to $\infty$, is
asserted, and a formula for the diffusion constant involving an abelian limit is
presented. The Einstein relation forms the content of Theorem~3.6.

In Section~4, the Dyson expansion of the propagator describing the
time-evolution of general mixed states of the system in powers of the operator
describing the interactions between the particle and the reservoirs is derived.
This expansion also yields an expansion of the effective particle dynamics  and
of correlation functions of particle observables when the degrees of freedom of
the reservoirs are traced out; (Subsections~4.2 through~4.4). In Subsection~4.5,
the Dyson expansions are cast in the form of ``polymer expansions'' for dilute
gases of extended particles with hard-core interactions contained in a
one-dimensional ``space'', with ``space'' corresponding to the time axis of the
original system. 

In Section~5, the expansions derived in Section~4 are further studied and used
to prove existence and properties of the thermodynamic limit of various
quantities; see Subsection~5.1. (These results can be used to prove Lemmas~3.1
and~3.2.) In Subsection~5.2, bounds on the effective dynamics of the tracer
particle, after tracing out the reservoir degrees of freedom, are proven. In
Subsection~5.3, the Laplace transform in the time variable of the effective
particle dynamics is introduced, yielding an object resembling a resolvent of an
effective Hamiltonian that depends on a spectral parameter with the
interpretation of an energy. This ``pseudo-resolvent'' plays a fundamental role
in the proofs of the main results stated in Section~3. In Subsection~5.4, a
direct-integral  (fiber-momentum) decomposition of translation-invariant
operators acting on the particle state and in particular of the effective
particle dynamics is recalled. In Subsection~5.5, the main contributions, in
fibers of fixed momentum, to the pseudo-resolvents 
introduced in Subsection~5.3 are identified, and some key spectral properties of
these operators are described; (but see also~\cite{paper1}). 

In Section~6, the behavior of the ``pseudo-resolvents'' near the origin of the
complex spectral parameter plane (``zero energy'') is analyzed, which yields
information on large-time asymptotics of the effective dynamics.

In Section~7, the effective dynamics and the correlation functions are studied
for a particle in a vanishing external force. In this situation, the state of
the system (restricted to continuous functions of the particle momentum
operator) is shown to approach an equilibrium state at the temperature of the
reservoirs. This yields a proof of Theorem 3.3. 

In Section~8, the Einstein relation (Theorem~3.6) is proven.

All other results, in particular Theorems~3.4 and~3.5, have already been proven
in the companion paper~\cite{paper1}, using the results established in Sects. 5
and 6 of this paper. In order to render the present paper comprehensible and
more or less self-contained, it is unavoidable to repeat a certain amount of
material concerning definitions and notation from~\cite{paper1}. The reader is
recommended to consult the companion paper~\cite{paper1} for background from
physics, motivation and discussion of the main results.\\

\textit{Acknowledgements.}
 W.D.R.\ is grateful to the DFG for financial support.

\section{Definition of the model} \label{sec: model}

In this section, we recall the definition of the model studied in \cite{paper1}
and in this paper in a finite volume approximation. Among other things, we also
repeat some important notations and conventions, the definition of the reservoir
equilibrium states and the definition of the effective particle dynamics.

\subsection{Notations and conventions}\label{section2.1.1}
\subsubsection{Banach spaces}
Given a Hilbert space $\scrE$, we use the standard notation 
\begin{equation*} \scrB_p(\scrE)\deq \left\{  A \in \scrB(\scrE)\,:\,
\Tr\left[(A^*A)^{p/2}\right] < \infty  \right\}  ,\qquad   1 \leq p \leq
\infty\,,  \end{equation*}
with $\scrB_\infty(\scrE)\equiv \scrB(\scrE)$ the bounded operators on $\scrE$,
and
\begin{equation*}
\norm A \norm_p \deq \left(\Tr\left[(A^*A)^{p/2}\right]\right)^{1/p}\,, \qquad 
\norm A \norm\deq\norm A \norm_{\infty}\,.
\end{equation*}
 For operators acting on $\scrB_p(\scrE)$, e.g., elements of
$\scrB(\scrB_p(\scrE))$, we often use the calligraphic font: $\caV,\caW$ etc..
An operator $A \in \scrB(\scrE)$ determines  bounded operators
\mbox{$\mathrm{Ad}(A)\,,\adjoint(A)\,,A_{\links}\,,A_{\rechts} $} on $\scrB_p(\scrE)$ by
\begin{align*}
\mathrm{Ad}(A) B\deq ABA\adj\,,\qquad
\adjoint(A) B\deq[A,B]= AB-BA
\end{align*}
and
\begin{align}\label{eq:2.5}
A_{\links}B\deq AB\,,\qquad A_{\rechts}B\deq BA\adj\,,\qquad   B \in
\scrB_p(\scrE)\,.
\end{align}
Note that $(A_1)_{\links}(A_2)_{\rechts}=(A_2)_{\rechts} (A_1)_{\links}$, as
operators on $\scrB_p(\scrE)$, $A_1,A_2\in\scrB(\scrE)$, i.e., the left- and
right multiplications commute. The norm of  operators in $\scrB(\scrB_p(\scrE))$
is defined by \begin{equation*}\label{def: norm on operators}
\norm \caW \norm \deq \sup_{A \in \scrB_p(\scrE)}   \frac{\norm \caW(A)
\norm_p}{\norm A\norm_p}\,.
\end{equation*}
In the following, we usually set $p=1$ or $2$.

\subsubsection{Scalar products}
For vectors $\ka \in \bbC^d$, we let $\mathrm{Re}\, \ka$ denote the vector
$(\mathrm{Re}\, \ka^1, \ldots, \mathrm{Re}\, \ka^d)$, where $\mathrm{Re}$
denotes the real part. Similar notation is used for the imaginary part,
$\mathrm{Im}$. The scalar product on $\bbC^d$ is written as
$(\kappa_1,\kappa_2)$ or  
$\overline{\kappa_1}\cdot\kappa_2$ and the norm as $\str \ka \str \deq
\sqrt{(\ka,\ka)}$. 
The scalar product on  an infinite-dimensional Hilbert space $\scrE$  is written
as $\langle \cdot\,, \cdot \rangle$, or, occasionally, as $\langle \cdot\,,
\cdot \rangle_{\scrE}$. All scalar products are defined to be  linear in the
second argument and anti-linear in the first one.

\subsubsection{Kernels}
 For  $\scrE=\ell^2(\Z^d)$,  we  can  represent
$A\in\scrB_{2}(\scrE)$ by its kernel $A(x,y)$, i.e., $(Af)(x)=\sum_{y} A(x,y)f(y)$, $f\in\scrE$. Similarly, an operator, $\mathcal{A}$, acting on $\scrB_2(\scrE)$ can be represented by
 its kernel $\mathcal{A}(x,y,x',y')$ satisfying
{$(\mathcal{A}\rho)(x,y)=\sum_{x',y'}\mathcal{A}(x,y,x',y')\rho(x',y')$},
$\rho\in\scrB_2(\scrE)$. 
Occasionally, we use the notation $\ket x$ for $\delta_x  \in \scrE$, defined by $\delta_x(x')=\delta_{x,x'}$, and $\bra x$ for $\langle\,\delta_x\,,\,\cdot\,\rangle$.
 In this
notation $\ket x\bra y$ stands for the rank-one operator $\delta_x\,\langle
\delta_y\,,\,\cdot\,\rangle$. Similarly, for the choice
$\scrE=\mathrm{L}^2(\tor)$, we often use the notation $\ket f$ for
$f\in\mathrm{L}^2(\tor)$ and $\bra g$ for $\langle g,\,\cdot\,\rangle$,
$g\in\mathrm{L}^2(\tor)$. In this `Dirac notation', $\ket f\bra g$ stands for
the rank-one operator $f\langle g,\,\cdot\,\rangle$ on $\mathrm{L}^2(\tor)$.

\subsection{The particle}\label{particle}
Consider the hypercube $ \Lambda=\Lambda_{L}=\bbZ^d \cap [-L/2,L/2]^d$, for some $L\in 2\N$. The particle Hilbert space is chosen as
$\mathscr{H}_{\sys}=\ell^2(\Lambda)$ where the subscript $S$ refers to `system'.
  
To describe the hopping term (kinetic energy), we choose a real function $\ve:
\bbT^d \to \bbR$ and we consider the self-adjoint operator $T\equiv T^{\La} $ on
$\ell^2(\La)$ with symmetric kernel\footnote{{ Later, we will consider $T^{\Lambda}$ as an operator on $\ell^2(\Z^d)$ by the natural embedding of $\ell^2(\La)$ into $\ell^2(\Z^d)$. As such, it has the kernel 
\begin{align}
T^{\Lambda}(x,x')=\begin{cases}
                   \hat\epsilon(x-x')\,,\quad &\textrm{if}\quad x,x'\in \Lambda\\
0 &\textrm{else }
                  \end{cases}\,,
\end{align}
 i.e., we impose Dirichlet boundary conditions.}}
\begin{equation*}
T(x,x') = \hat \ve(x-x')\,,
\end{equation*}
with $\hat \ve$ the Fourier transform of $\ve$. Since we will assume $\ve$ to
be analytic, the hopping is short range.

A natural choice for the dispersion law is $\varepsilon(k) = \sum_j 2 (1-\cos
k^j) $, corresponding to $T=-\Delta$, with $\Delta$ the lattice
Laplacian on $\ell^2(\Lambda)$ with Dirichlet boundary conditions. This choice satisfies all our assumptions, to be stated in 
Section~\ref{assume}.

We define the particle Hamiltonian as
\begin{equation*}
H_{\sys}\deq T-F \cdot X\,,
\end{equation*}
where $F \in \R^d$ is an external force field, e.g., an electric field, and $X\equiv X^{\Lambda}$ 
denotes the position operator on $\scrH_{\mathrm{S}}$, defined by $Xf(x)=x f(x)$. In what
follows we will write $F= \la^2 \field$, with $\field$ a rescaled field, (a
notation to be motivated later).\\

\subsection{The reservoir}\label{thereservoir}
\subsubsection{Dynamics}
For each $x\in\Z^d$, we define a reservoir Hilbert space at site $x$ by
\begin{align*}
\mathscr{H}_{\res_x}\deq\Gamma_s( \mathrm{L}^2(\bar\Lambda))\,,
\end{align*}
where $\bar\Lambda=\bar\Lambda_{L}=\R^d\cap[-L/ 2,L/2]^d$ and $\Gamma_s(\caE)$ is the
symmetric (bosonic) Fock space over the Hilbert space~$\caE$. We assume that the
reader is familiar with basic concepts of second quantization, such as Fock
space and creation/annihilation operators; (we refer to,
e.g.,~\cite{derezinski1} for definitions and background).
The total reservoir Hilbert space is defined by
\begin{align*}
\mathscr{H}_{\res}\deq\bigotimes_{x\in\Lambda}\mathscr{H}_{\res_x}\,.
\end{align*}
Note that for all $x$, the spaces $\mathscr{H}_{\res_x}$ are isomorphic to each other. We remark that there is no compelling reason to restrict the one-site reservoirs to
the same region, $[-L/2,L/2]^d$, as the particle system, but this simplifies our
notation. The reservoir Hamiltonian is defined as
\begin{align}\label{def: res ham}
H_{\res}\deq\sum_{x\in\Lambda} \, \, \sum_{q\in\bar\Lambda\adj} \bosondispersion(q)
 a_{x,q}\adj a_{x,q}\,,
\end{align}
where $\bar\Lambda\adj=\frac{2\pi}{L}\Z^d$ is the set of quasi-momenta for the
reservoir at site $x$, and the operators $a^{\#}_{x,q}$ are the canonical
creation/annihilation operators satisfying the commutation relations
\begin{align*}
{ [a_{x,q},a\adj_{x',q'}]=\delta_{x,x'}\delta_{q,q'}\,, \quad\quad[a_{x,q},a_{x',q'}]=[a\adj_{x,q},a\adj_{x',q'}]=0\,,}
\end{align*}
and we choose the dispersion law $\bosondispersion(q)=  \str q
\str + \delta_{q,0}$. Note that this dispersion law corresponds to photons or phonons, except for $q=0$,  where we have modified this dispersion
law at $q=0$ by adding an infrared
regularization that does not affect any of our results; e.g., if we replace
$\delta_{0,q}$ by $K \delta_{0,q}$, with $K >0$, then all  infinite-volume
objects studied in this paper are independent  of $K$.

\subsubsection{Equilibrium state} \label{sec: eq state}
Next, we introduce the {\it Gibbs state} of the reservoir at inverse temperature
$\beta$, $0<\be<\infty$. It is given by the density matrix 
\begin{equation}\label{eq:2.14}
\rho_{\res, \refer} \deq\frac{1}{Z_{\referres}}\e{-\beta
H_{\res}}\,,\quad\quad\textrm{where}\quad{Z_{\referres}}=\Tr_{\res}[\e{-\beta
H_{\res}}]\,,
\end{equation}
where $\Tr_{\res}$ denotes the trace over $\mathscr{H}_{\res}$.

An alternative way to describe this density matrix is to specify the expectation
values of arbitrary observables, which we denote by $ \langle O
\rangle_{\rho_\referres}\deq \Tr \left[O\rho_{\res,\refer}\right]$. 
 For $\varphi\in\ell^2(\bar\Lambda^*)$, we write $a_{x}(\varphi)= \sum_{q
\in\bar\Lambda^* }\varphi(q)a_{x,q}$, and we choose observables, $O$, to be
polynomials in the operators $a_{x}(\varphi)$.  One then finds that, for any
$x,x'$ and $\varphi,\varphi'\in\ell^2(\bar\Lambda^*)$:
\begin{itemize}
\item[$i.$] Gauge-invariance:
\begin{equation}\label{eq:gaugeinvariance}
\langle a_x\adj(\varphi)  \rangle_{\rho_\referres} =\langle a_x(\varphi) 
\rangle_{\rho_\referres}=0\,;
\end{equation}
\item[$ii.$]  Two-point correlations: Let
$\varrho_{\beta}\deq(\e{\beta\bosondispersion}-1)^{-1}$, with the one-particle dispersion law
$\bosondispersion(q)=\str q \str+\delta_{q,0}$, be the Bose-Einstein density (operator). Then
\end{itemize}
\begin{align*}
\left( \begin{array}{cc} \langle a^*_{x}(\varphi)
a_{x'}(\varphi')  \rangle_{\rho_\referres}
& \langle a^*_{x}(\varphi) a^*_{x'}(\varphi')  \rangle_{\rho_\referres}  
\\[1mm]
\langle a_{x}(\varphi)a_{x'}(\varphi') \rangle_{\rho_\referres}&\langle
a_{x}(\varphi) a^*_{x'}(\varphi')
 \rangle_{\rho_\referres}
\end{array} \right)  =   \delta_{x,x'} \left(\begin{array}{cc}
 \langle \varphi' ,\varrho_{\beta} \varphi \rangle & 0    \\[1mm]
0 & \langle \varphi , (1+ \varrho_{\beta}) \varphi') \rangle
\end{array}\right)\,;
\end{align*}
\begin{itemize}
\item[$iii.$] Wick's theorem:
 \begin{align}
\langle a^{\#}_{x_{2n}}(\varphi_{2n})  \ldots a^{\#}_{x_{1}}(\varphi_{1}) 
\rangle_{\rho_\referres}  & =  \sum_{\pi\in\mathrm{Pair}(n)} \prod_{(r,s) \in
\pi} \langle a^{\#}_{x_s}(\varphi_s) a^{\#}_{x_r}(\varphi_r) 
\rangle_{\rho_\referres}   \label{eq: gaussian property1}\,,\\[2mm]
\langle a^{\#}_{x_{2n+1}}(\varphi_{2n+1})  \ldots a^{\#}_{1}(\varphi_{1}) 
\rangle_{\rho_\referres}  &=  0\,,\label{eq: gaussian property2}
\end{align}
where $\mathrm{Pair}(n)$ denotes the set of partitions of $\{1,\ldots,2n \}$ into $n$ pairs and
the product is over these pairs $(r,s)$, with the convention that $r<s$. Here,
$\#$ stands either for $\adj$ or nothing.
\end{itemize}

\subsection{The interaction}\label{the interaction}
We define the Hilbert space of state vectors of the coupled system (particle and
reservoirs) by
\begin{equation*}
\mathscr{H}\deq \mathscr{H}_{\mathrm{S}}\otimes\mathscr{H}_{\mathrm{R}}\,.
\end{equation*}
We pick  a smooth `structure factor' $\phi\in\mathrm{L}^2(\bbR^d)$ and we define its
finite volume version $\phi^{\La} \in \ell^2(\bar\La^*)$ by
\mbox{$\phi^{\La} (q)=  (2\pi/L)^{d/2}\phi(q)$}, with the normalization chosen such that $\norm
\phi\norm_{\mathrm{L}^2(\bbR^d)}= \mathop{\lim}\limits_{L \to \infty}
\norm \phi^\La\norm_{\ell^2({\bar\La^*})} $. We will drop the superscript
$\La$. The interaction between the particle and the reservoir at site $x$ is given by
\begin{equation*}
\lone_x  \otimes  \Psi_x(\phi),  \quad \textrm{where} \quad    \Psi_x(\phi)=  
a_x(\phi)+ {a}\adj_x(\phi)\
\end{equation*}
is the field operator, and $\lone_x=\ket x\bra x$ denotes the projection onto the
lattice site $x$. The interaction Hamiltonian is taken to be
\begin{equation*}
H_{\mathrm{SR}}\deq \sum_{x \in \Lambda}   \lone_x  \otimes \Psi_x(\phi) \quad  
\textrm{on} \quad 
\mathscr{H}_{\sys} \otimes  \mathscr{H}_{\res}\,.
\end{equation*}

The total Hamiltonian of the interacting system on $\scrH$ is then given by
\begin{align}\label{eq:1.13}
H\deq T\otimes\lone-\lambda^2\field\cdot X\otimes\lone+\lone\otimes H_{\res}+\lambda H_{\mathrm{SR}}\,,
\end{align}
where $\lambda\in\R$ is a coupling constant. The interaction term $H_{\sys\res}$
is relatively  bounded w.r.t.\ $H_{\sys}+H_{\res}$ with arbitrarily small
relative bound.  It follows that $H$ is essentially selfadjoint on the domain 
$\scrH_\sys \otimes \Dom(H_\res)$, (where $\Dom(H_\res)$ denotes the domain of
$H_{\res}$).

\subsection{Effective dynamics and correlation functions} \label{effective dynamics}
The time-evolution in the Schr\"odinger picture is given by
\begin{equation*}
\rho_t=\e{-\ii t H}\rho\,\e{\ii t H}\,,\quad\quad\rho\in\mathscr{B}_1(\mathscr{H})\,.
\end{equation*}
We will first choose an initial state $\rho$ of the form $\rho= \rho_\sys
\otimes \rho_{\res,\refer}$, with $\rho_{\res,\refer}$ as defined above.
Of course,~$\rho_t$, with $t>0$, will in general not be a simple tensor product, but we
can always take the partial trace, $\Tr_{\res}[\,\cdot\,]$, over $\scrH_\res$ to
obtain the `reduced density matrix' $\rho_{\sys,t} $ of the system;
\begin{equation*}
\rho_{\sys,t} =  \Tr_\res \left[ \e{-\ii t H}(\rho_\sys \otimes
\rho_{\res,\refer}) \e{\ii t H} \right]= : \caZ_{[0,t]}\rho_\sys\,,
\end{equation*}
and we call $\caZ_{[0,t]}:
\scrB_1(\mathscr{H}_{\sys})\rightarrow\scrB_1(\mathscr{H}_{\sys}):  \rho_{\sys}
\mapsto \rho_{\sys,t} $ the {\it reduced} or {\it effective} dynamics. 
It  is a trace-preserving and completely positive map. 

In the present paper, we will mainly consider observables of the form $O \otimes
\lone$ with $O \in \scrB(\scrH_\sys)$, in which case we can also write
\begin{align}\label{expectationvalue}
\langle O(t)\rangle_{\rho_{\sys}\otimes \rho_{\res,\refer}}\deq\Tr[O(t)\rho_{\sys}\otimes\rho_{\res,\refer}]=\Tr_{\sys}[O
\rho_{\sys,t}]\,,
\end{align}
where the trace $\Tr[\,\cdot\,]$ is over the Hilbert space $\scrH$, the trace  $\Tr_{\sys}[\,\cdot\,]$ is over the particle Hilbert space
$\mathscr{H}_{\sys}$ and~$O(t)$ is the Heisenberg picture time evolution of the observable~$O\otimes\lone$, i.e.,
\begin{align}\label{Heisenberg}
 O(t)\deq\e{\ii tH}(O\otimes\lone)\,\e{-\ii t H}\,.
\end{align}
Note that $O(t)$ is, in general, \textit{not}
of the product form $O' \otimes\lone$, for some $O'$. 

Next, consider several observables
$O_1,\ldots,O_m\in\mathscr{B}(\mathscr{H})$ and a set of times $t_1,\ldots
,t_m\in\R$. For $\rho_\sys\in\scrB_1(\scrH_{\sys})$ we define correlation
functions by the formula
\begin{equation}\label{definition of correlation function}
\langle O_m(t_m)\cdots O_1(t_1)\rangle_{\rho_{\sys}\otimes \rho_{\res,\refer}}
\deq\Tr[O_m(t_m)\cdots O_1(t_1)  (\rho_\sys \otimes \rho_{\res,\refer} )]\,,
\end{equation}
the trace being over the Hilbert space $\scrH$.

\subsubsection{Equilibrium states}
Apart from an initial state (density matrix) of the product form $\rho_\sys
\otimes \rho_{\res, \refer}$, we also consider the Gibbs state of the coupled
system when the external force field vanishes, $\field=0$. In finite volume, it is defined by
\begin{equation*}
\rho_\be \deq \frac{1}{Z_{\be}} \e{-\be H^{\field=0}}, \qquad Z_{\be}=  \Tr
\e{-\be H^{\field=0}}, \qquad H^{\field=0} =T\otimes\lone+\lone\otimes H_\res + \la H_{\sys\res}\,,
\end{equation*}
and one easily checks that $\rho_\be \in \scrB_1(\scrH)$. 
The correlation functions determined by $\rho_\beta$ are written as
\begin{equation*}
\langle O_m(t_m)\cdots O_1(t_1)\rangle_{\rho_\refer} \deq\Tr[ O_m(t_m)\cdots
O_1(t_1)  \rho_{\refer}  ]\,,
\end{equation*}
with $O_1,\ldots,O_m$ and $t_1,\ldots,t_m$ as in~\eqref{definition of correlation function}.
\subsubsection{Time-reversal} 
We define an anti-linear time-reversal operator $\Theta= \Theta_\sys \otimes
\Theta_\res$, where $\Theta_\sys$ is given by 
\begin{equation*}
\Theta_\sys f(x)  = \overline{f(x)}\,, \qquad  f \in \ell^2(\La)\,,
\end{equation*}
and $\Theta_\res$ by $\Theta_\res\deq \Ga_s(\theta_\res)$, with the one-particle
operator $\theta_\res$ given by
\begin{equation*}
\theta_\res \varphi_{x}(q)  = \overline{\varphi_{x}(-q)}\,, \qquad  \varphi_{x}
\in  \ell^2(\bar \La^*)\,, \qquad x \in \Lambda\,.
\end{equation*}
If the dispersion law $\varepsilon$ of the particle and the form factor $\phi$
are invariant under time-reversal, i.e., $\varepsilon(k)=\varepsilon(-k)$,
$\phi(q)=\overline{\phi(-q)}$ (as will be assumed) then we have that
\begin{equation*}
 \Theta H^{\field=0} \Theta=H^{\field=0}, \qquad   \Theta \rho_\refer \Theta =
\rho_\refer\,,
\end{equation*}
expressing time-reversal invariance of the model. 
\subsubsection{KMS condition}
The KMS condition characterizing the Gibbs state $\rho_{\beta}$ can be expressed
as follows: Denote by $\bbH_{\beta}$ the strip
\begin{align}\label{beta strip}
 \bbH_{\be}\deq \{z \in \bbC\,:\, 0 \leq \im z \leq \be \}\,.
\end{align}
Then, for $O_1,O_2\in\scrB(\scrH)$,  the correlation function  $z\mapsto\langle O^{\field=0}_2(z)
O_1\rangle_{{\rho_\beta}} $ is
analytic in the interior of the strip~$\bbH_{\be}$, bounded and continuous on~$\bbH_{\be}$, and satisfies the \emph{KMS (boundary) condition}
\begin{align}
\langle O^{\field=0}_2(t) O_1\rangle_{{\rho_\referres}}=\langle O_1
O^{\field=0}_2(t+\ii\beta)\rangle_{{\rho_\referres}}\,,\qquad t\in\R\,. \label{eq: finite volume kms}
\end{align}
This follows from the cyclicity of the trace. Note that we write $O^{\field=0}(t)$ to indicate that, here, the time evolution
is generated by $H^{\field=0}$.

\section{Assumptions and Results}\label{sec: results}

\subsection{Assumptions}\label{assume}
The model introduced in the last section is parametrized by two functions: the dispersion law $\varepsilon: \bbT^d\to\bbR$, and the form factor
$\phi:\bbR^d\to\bbC$. In this subsection, we formulate our assumptions  on these two
functions. The  (multi-) strip $\mathbb{V}_{\delta}$ is defined by
\begin{equation}
\bbV_{\delta}\deq\{ z \in (\bbT+\ii \bbT)^d\,:\, |\im z \str \leq \delta \}\,.
\end{equation}

 \renewcommand{\theassumption}{\Alph{assumption}}
 
\begin{assumption} \label{ass: analytic dispersion}\emph{ [Particle
dispersion]} 
The function $\varepsilon$ extends to an analytic function in a region
containing a strip $\bbV_{\delta}, \delta >0$. In particular, the norm
\begin{equation*}
\norm \ve \norm_{\infty,\delta} \deq\sup_{p \in \bbV_{\delta}} \str \ve(p) \str 
\end{equation*}
is finite, for some $\delta>0$. 
Furthermore, there  does not exist any  $v \in\R^d$ such that the function
\[\tor\ni k \mapsto
 ( v, \nabla\varepsilon(k))\]  vanishes identically.
\end{assumption}
This assumption allows us to estimate the free particle propagator
$\e{-\ii t H_{\sys}}$ on the particle Hilbert space \mbox{$\scrH_{\sys}=\ell^2(\La_L)$}
as
follows: 
\begin{align}\label{eq: propagation bound}
\big|\big(\e{-\ii t H_{\sys}}\big)(x,x')\big|\le C\e{-\nu|x-x'|}\e{ t  \norm \im\ve
\norm_{\infty,\nu} }\,.
\end{align}
For $L=\infty$, the bound~\eqref{eq: propagation bound} is
the {\it Combes-Thomas} bound; for finite $L$, it can be established in an
analogous way. If we replace $\Z^d$ by $\R^d$,  any physically acceptable dispersion law
$\varepsilon$ is unbounded, and there is no exponential decay in $|x-x'|$.  This is the main reason why
the system studied in this paper is defined on a lattice.

The next assumption deals with the `time-dependent' correlation function defined (in finite-volume) as
 \begin{equation}\label{eq: clear presentation finite volume correlationfunction one}
\hat \psi^{\La}(t)\deq   \sum_{q \in \bar\La^*}  
\str \phi^{\La}(q)\str^2 \left(  \frac{\e{-\ii t \bosondispersion(q)  }}{\e{\beta\bosondispersion(q)}-1}+\frac{\e{\ii t \bosondispersion(q)}  }{1-\e{-\beta\bosondispersion(q)}}  \right),
\end{equation}
and in the thermodynamic limit as
\begin{equation} \label{eq: clear presentation of correlation function}
\hat\psi(t)\deq
 \int \d q \, \str \phi(q)\str^2  \left(  \frac{\e{-\ii t \str q
\str}}{\e{\beta\str q \str}-1}+\frac{\e{\ii t \str q \str}  }{1-\e{-\beta\str q
\str}}  \right).\
\end{equation}
Since the correlation function $\hat \psi$ is determined by the form factor $\phi$, the following assumption is in fact a constraint on the choice of $\phi$.
\begin{assumption} \label{ass: exponential decay}\emph{[Decay of reservoir correlation
function]}
The form factor $\phi$ is a spherically symmetric function, i.e., $\phi(q) =:
\phi(\str q \str)$. The correlation functions
$\hat\psi^{\La}(z)$, $\hat \psi(z)$ are uniformly bounded in $ \La$ and $z
\in \bbH_{\be}$, (see~\eqref{beta strip}), and 
\begin{align*}
\tdl\hat\psi^{\La}(z) = \hat \psi(z) 
\end{align*}
holds uniformly on compacts in $
\bbH_\be$, where $\lim_{\La}$ stands for $\lim_{L \to \infty}$ (recall that $\La\equiv\La_L$).  Furthermore, the number
\beq \label{eq: gs shift}
\sum_{q \in \bar\La^*}  \bosondispersion(q)^{-1} \str \phi^{\La}(q)\str^2
\eeq
is  bounded uniformly in~$\La$. Most importantly,  $\hat\psi(z)$ is continuous on $\bbH_\be$ and 
\begin{align*}
\str\hat{\psi}(z)\str \leq C \, \e{-g_\res\str z\str}\,, \qquad z \in \bbH_\be\,.
\end{align*}
\end{assumption}
This assumption entails that the
reservoirs exhibit exponential loss of memory. This is a key ingredient for our
analysis.

 Often, one also considers the `spectral density'
\begin{equation}\label{definition of psi ohne hut}
\psi (\bosondispersion) = \frac{1}{2 \pi} \int_{-\infty}^{\infty} \d t  \, \hat \psi(t) \,\e{\ii t \bosondispersion}\,.
\end{equation}
It satisfies the so-called `detailed balance' property $\e{\be \bosondispersion}\psi(\bosondispersion)= \psi(-\bosondispersion)$, which
expresses, physically, that the reservoir is in thermal equilibrium at inverse temperature $\beta$.

Assumptions~\ref{ass: analytic dispersion} and~\ref{ass: exponential decay} are
henceforth required and will not be repeated.

\subsection{Thermodynamic limit}\label{sec: thermo}
Up to this point, we have considered a system in a finite volume (cube), $\La$
or $\bar \La$, characterized by its linear size $L$. However, if we wish to
study dissipative effects, we must, of course, pass to the thermodynamic limit,
in order to eliminate finite-volume effects such as Poincar\'e recurrence. This
amounts to taking $\La= \bbZ^d, \bar \La=\bbR^d$ and is accomplished below. 

In this section, we will explicitly put a label $\La$ on all quantities
referring to a system in a finite volume.  As an example, $\scrH_\sys$ now
stands for $\ell^2(\bbZ^d)$, and we write $\scrH^\La_\sys$ for 
$\ell^2(\La)$.   The shorthand $\lim_{\La}$ stands for  the thermodynamic limit,
$\lim_{L \to \infty}$.

\subsubsection{Observables of the system}\label{sec: observables of the system}
We begin by defining some classes of infinite-volume system observables, (i.e., certain
types of bounded operators on $\scrH_{S}$). We say that an operator $O \in
\scrB(\scrH_\sys)$ is {\it exponentially localized} whenever
\begin{equation*}
\str O(x,x')\str \leq C \e{-\nu (\str x \str+\str x' \str) }, \qquad \textrm{for
some}\, \nu >0\,.
\end{equation*}
An important r\^{o}le is played by the so-called \textit{quasi-diagonal} operators. These are operators $O\in\scrB(\scrH_\mathrm{S})$ with
the property that

\begin{equation*}
\str O(x,x')\str \leq C \e{-\nu (\str x-x' \str) }, \qquad \textrm{for some}\,
\nu >0\,.
\end{equation*}
We denote by $\mathop{\mathfrak{A}}\limits^{\circ} $ the class of quasi-diagonal operators and by $\mathfrak{A} $ its norm-closure.

An observable $O \in \scrB(\scrH_\sys)$ is said to be \textit{translation-invariant} whenever
$\caT_yO=O, $ for arbitrary $y \in \bbZ^d$, where $\caT_y O(x,x') \deq O(x+y,x'+y)$.  
Translation-invariant operators on $\scrH_\sys$ form a \textit{commutative}
$C^*$-algebra denoted by $\mathfrak{C}_{\mathrm{ti}}$. We also introduce
the algebras
\begin{equation*}
{\mathop{\mathfrak{A}}\limits^{\circ}}_{\mathrm{ti}}
\deq\mathfrak{C}_{\mathrm{ti}} \cap \mathop{\mathfrak{A}}\limits^{\circ}\,,
\qquad \mathfrak{A}_{\mathrm{ti}} \deq \mathfrak{C}_{\mathrm{ti}} \cap
\mathfrak{A}\,.
\end{equation*}

An operator $O \in  \mathfrak{C}_{\mathrm{ti}} /
{\mathop{\mathfrak{A}}\limits^{\circ}}_{\mathrm{ti}}/ \mathfrak{A}_{\mathrm{ti}}$
can be identified with a multiplication operator, $M_{f}$, on the Hilbert space~$\mathrm{L}^2(\tor)$, i.e., $M_fg=fg$, $g\in\mathrm{L}^2(\tor)$,
with $f:\tor\mapsto \C$ a bounded and measurable/real-analytic/continuous function. Physically, the variable in $\bbT^d$ is the momentum of the particle.

These classes of operators are introduced because certain expansions used in our analysis will apply to quasi-diagonal operators or translation-invariant quasi-diagonal operators, and they can be extended to the closures of these algebras by density.

In analyzing diffusion and in the proof of the Einstein relation we also need to consider certain observables that are unbounded operators: We introduce the
~$^*$-algebra $\mathfrak{X}$ that consists of
polynomials in the components,~$X^{i}$, $i=1,\ldots,d$, of the particle-position operator $X$.

Given an infinite-volume observable $O\in\scrB(\mathscr{H}_{\sys})$,
$\mathscr{H}_{\sys}=\ell^2(\Z^d)$, or $O \in \mathfrak{X}$, we
associate an observable $O^{\Lambda}= \lone_{\La} O \lone_{\La}$ on $\scrH_{\sys}^{\Lambda}=\ell^2(\Lambda)$ with it, where
$ \lone_{\La}$ is the orthogonal projection  $\ell^2(\bbZ^d) \to
\ell^2(\Lambda)$.  

\subsubsection{Dynamics}
We choose not to construct directly the time-evolution of infinite-volume observables and
infinite-volume states, although this could be done by using the Araki-Woods representation of the system in the thermodynamic limit. Instead, we will analyze
the infinite-volume dynamics of \textit{`reduced'} states, i.e., of states restricted to
particle observables and correlation (Green) functions of particle-observables by 
constructing these objects as thermodynamic limits of finite-volume expressions.

An infinite-volume density matrix of the particle system
$\rho_{\mathrm{S}}\in\mathscr{B}_1(\scrH_{\sys})$ is called {\it exponentially localized} if
\begin{align}\label{exponentially localized density matrix}
 |\rho_{\mathrm{S}}(x,x')|\le C\e{-\nu(|x|+|x'|)}\,,\qquad \textrm{for some}\, \nu >0\,.
\end{align}
Given such an infinite-volume density matrix $\rho_{\mathrm{S}}$, we
associate finite-volume density matrices 
\begin{align}\label{finite volume density matrix}
 \rho_{\sys}^{\Lambda}\deq
\frac{1}{Z_{\rho_{\sys}}^{\Lambda}}\lone_{\La} \rho_{\sys}
\lone_{\La}\in\mathscr{B}_1(\scrH_{\sys}^{\Lambda})\,,\qquad
Z_{\rho_{\sys}}^{\Lambda}\deq\Tr_{\sys}[\lone_{\La} \rho_{\sys}
\lone_{\La}]\,,                                      
\end{align}
 with it. Note that, due to the normalization by
$Z_{\rho_{\sys}}^{\Lambda}$, $\rho_{\sys}^{\Lambda}$ is a density matrix on
$\scrH_{\sys}^{\Lambda}$. 

Recall the definition of the reduced dynamics, $\caZ^{\La}_{[0,t]}$, introduced
in Section~\ref{effective dynamics}. We set
\begin{equation}\label{eqabove}
\caZ_{[0,t]} \rho_{\sys} \deq\tdl  \caZ_{[0,t]}^{\La} \rho_\sys^{\La}  \,.
\end{equation}
The next lemma asserts that the thermodynamic limit (as $\Lambda$ and $\bar\Lambda$ increase to $\mathbb{Z}^{d}$, $\R^d$, respectively) in~\eqref{eqabove} exists, 
and that the resulting reduced dynamics 
$\caZ_{[0,t]}$ is translation-invariant.
\begin{lemma} \label{lem: thermodynamic dyn}
The limit on the right side of Equation~\eqref{eqabove} exists in
$\scrB_1(\scrH_\sys)$, and this defines the map $\caZ_{[0,t]}: 
\scrB_1(\scrH_\sys)\to \scrB_1(\scrH_\sys)$.  The map $\caZ_{[0,t]}  $ preserves
the trace, i.e.,  $\Tr_{\sys}[ \caZ_{[0,t]} \rho_\sys]= \Tr_{\sys}[ \rho_\sys]$,
positivity and exponential localization of the state of the particle, i.e., if
$\rho_\sys$ has any of these properties, then so does  $\caZ_{[0,t]}
\rho_\sys$. 
Moreover, $\caZ_{[0,t]}$ is translation-invariant; $\caT_{-y}\caZ_{[0,t]}
\caT_y=\caZ_{[0,t]} $ for $y \in \bbZ^d$  with $\caT_y$ as in Subsection~\ref{sec:
observables of the system}. As a consequence of the above, 
for $O$ in $\mathfrak{A}$ or $\mathfrak{X}$, and for an exponentially localized
state $\rho_\sys$, we can define
\begin{equation*}
 \langle O(t) \rangle_{\rho_\sys \otimes \rho_\referres}\deq \Tr_{\sys} [O \caZ_{[0,t]} \rho_{\sys}]\,. 
\end{equation*}

\end{lemma}

\subsubsection{Correlation functions}
Next, we define the infinite-volume analogues of the finite-volume correlation
functions 
\begin{equation*}
\langle O^{\La}_m(t_m)\cdots O^{\La}_1(t_1)\rangle_{\rho^\La_\refer} \qquad
\text{and}  \qquad \langle O^{\La}_m(t_m)\cdots
O^{\La}_1(t_1)\rangle_{\rho^\La_\sys \otimes \rho^\La_{\res,\refer}}
\end{equation*}
that have been introduced in Section~\ref{effective dynamics}.

Consider observables $O_1, \ldots, O_m$ in $\mathfrak{A}$ or $\mathfrak{X}$,
times $0 \le t_1 < \ldots < t_m$, and an exponentially localized density matrix
$\rho_\sys$. We then define
 \begin{align}
\langle O_m(t_m)\cdots O_1(t_1)\rangle^{}_{\rho_{\sys}\otimes
\rho_\referres}\deq\tdl\langle O^{\La}_m(t_m)\cdots
O^{\La}_1(t_1)\rangle^{}_{\rho^{\La}_{\sys}\otimes \rho_{\referres}^{\La}} \,.
\label{def: thermo limit loc}
\end{align}

Similarly, for observables $O_1,
\ldots, O_m \in \mathfrak{A}_{\mathrm{ti}}$ and times $0 \le t_1 < \ldots <
t_m$, we define
 \begin{align}
\langle O_m(t_m)\cdots O_1(t_1)\rangle_{\rho_\be}\deq\tdl\langle
O^{\La}_m(t_m)\cdots O^{\La}_1(t_1)\rangle_{\rho^{\La}_\be}\,.   \label{def:
thermo limit beta}
\end{align}
Note that we construct the thermodynamic limit of equilibrium correlation functions only for translation-invariant observables, since, pictorially,
the particle is uniformly distributed in space and hence the expectation values of localized observables vanish. Also note that, in Equation~\eqref{def: thermo limit beta}, we do \textit{not}
constrain the time-evolution to be the one generated by the Hamiltonian with
$\field=0$. But, of course, we  have to do so if we want the correlation
functions to be stationary in time, as in the following lemma.
\begin{lemma}\label{lem: thermodynamic obs}
The limits on the right hand sides of Equations~\eqref{def: thermo limit loc} and ~\eqref{def: thermo limit beta}
exist.   For $m=1,2$, they exist for $t_1, t_2 \in \bbR$ (i.e., not
necessarily positive or ordered).  
For the dynamics with $\field=0$, and for arbitrary observables $O_1, O_2 \in
\mathfrak{A}_{\mathrm{ti}}$, writing $O(t)$, instead of $O^{\field=0}(t)$, and
$O_\Theta$, instead of $\Theta_\sys O \Theta_\sys $, the following properties hold:
\begin{itemize}
\item[$i.$]Stationarity:  $\langle O_2(t_2)O_1(t_1) \rangle_{\rho_\be}=\langle
O_2(t_2+t)O_1(t_1+t) \rangle_{\rho_\be}$, for any $t \in \bbR$;
\item[$ii.$] Time-reversal invariance:  $\langle
O_{2,\Theta}(-t_2)O_{1,\Theta}(-t_1) \rangle_{\rho_\be}=\overline{\langle
O_2(t_2)O_1(t_1) \rangle}_{\rho_\be}$;
\item[$iii.$] KMS condition: There exists a function $ z \mapsto f_{ O_1, O_2}(z)$,
analytic in the interior of the strip $\bbH_{\be}$, bounded and continuous on $\bbH_{\be}$, that satisfies the \emph{KMS (boundary) condition}
\begin{equation*}
f_{ O_1, O_2}(t) \deq \langle O_2 O_1(t)\rangle_{\rho_\be}, \qquad  f_{ O_1,
O_2}(t+\ii\beta)=\langle O_1(t) O_2\rangle_{\rho_\be}\,,\quad\quad t\in\R\,.
\end{equation*}
\end{itemize}
\end{lemma}
We remark that there is no particular reason to limit our construction of
general correlation functions to one- and two-point functions, $m=1,2$. However,
focusing on these special cases will enable us to keep our notation manageable
in the technical sections. Lemmas~\ref{lem: thermodynamic dyn} and \ref{lem: thermodynamic obs} are proven in Section~\ref{sec: dyson
analyis}, using rather straightforward estimates. These are the only ones among our
results that do not require exponential decay of the reservoir correlation
function $\hat \psi$ (cf.\ Assumption~\ref{ass: exponential decay}), nor small coupling, $\lambda$, or weak external field, $\field$.

\subsection{Results}\label{sec: small results}
Next, we summarize our \textit{main results}. Throughout this section, it is understood that we consider the infinite-volume system; i.e., $\Lambda=\Z^d$, $\bar\Lambda=\R^d$.  

Theorem~\ref{thm:correlationfunctions} concerns the system in equilibrium, i.e., in a vanishing external force field, and asserts that, in this case, the system has the property of `return to equilibrium'. Theorem~\ref{thm: stationary} states the corresponding result for an off-equilibrium system: It claims that the state of the system approaches, for small external force fields~$\field$, a `{\it Non-Equilibrium Stationary State}' (NESS) in the limit of large times. Theorem~\ref{thm: diffusion} asserts that the motion of the particle is \textit{diffusive} at large times. For $\field\not=0$, this diffusive motion is around an average uniform motion (i.e., a drift at a constant velocity). Our last result, Theorem~\ref{thm: einstein}, confirms the fluctuation-dissipation formula of linear response theory: The equilibrium ($\field=0$) diffusion matrix is related to the response of the particle's motion to the field $\field$ through the `Einstein relation'.

 Theorems~\ref{thm: stationary} and~\ref{thm: diffusion} have already appeared in~\cite{paper1}, (see Theorems~3.2 and~3.3). We restate them here for completeness. The  main purpose of the present paper is to prove Theorems~\ref{thm:correlationfunctions} and~\ref{thm: einstein}, besides developing analytical techniques and establishing technical results that have been used to prove various results in \cite{paper1}.

Our first result (partially contained in \cite{deroeckfrohlichpizzo}) concerns
the model without external force field, i.e., $\field=0$.   

\begin{theorem}\emph{[Return to equilibrium]}\label{thm:correlationfunctions}
Let $\field = 0$. Then there exist a constant $k_\la>0$ and a decay rate $g>0$ such
that, for $0 < \str \la \str < k_\la$, the following holds.
For an arbitrary exponentially localized density matrix $\rho_{\sys}$, arbitrary
observables $O_1, \ldots, O_m \in \mathfrak{A}_{\mathrm{ti}}$ and times $0 \le
t_1 <\ldots < t_m$, 
\begin{align}\label{eq:2.12}
\langle O_m(t_m)\cdots O_1(t_1)\rangle_{\rho_{\sys}\otimes\rho_\referres}=
\langle O_m(t_m)\cdots O_1(t_1)\rangle_{\rho_\be}+ \caO(\e{- \la^2 g t_1})\,,
\qquad as \text{       } t_1 \to \infty\,,
\end{align}
and  the  correlation functions exhibit the following `exponential cluster
property': 
\begin{align}\label{eq:2.13}
\langle O_2(t_2) O_1(t_1)\rangle_{\rho_\be} =  \langle
O_1\rangle_{\rho_\beta}\langle O_2\rangle_{\rho_\beta}+\caO(\e{- \la^2 g \str
t_2-t_1 \str}) \,,\qquad as \text{    }t_2-t_1 \to \infty\,.
\end{align}
(In these equations, $O(t)$ stands for $O^{\field = 0}(t)$.)
\end{theorem}
In~\eqref{eq:2.13}, we consider only two
observables; but an analogous statement holds for $n>2$ observables. 

As already remarked, $\mathfrak{A}_{\mathrm{ti}}$ is commutative. Hence the positive and normalized functional $O \mapsto \langle
O \rangle_{\be}$ on $\mathfrak{A}_{\mathrm{ti}}$ can be expressed in terms of a probability measure. Recall that $O \in \mathfrak{A}_{\mathrm{ti}}$ is of the form $M_f$ for $f \in C(\tor)$. Anticipating the notation of Theorem \ref{thm: stationary}, we can write
\begin{align*}
\langle M_f \rangle_{\rho_\be}=   \langle  f, \zeta^{0,\la}  \rangle_{\mathrm{L}^2(\tor)}\,.
\end{align*}

Next, we state results that hold off equilibrium: In the theorems below, we use the notation~$O(t)$ for~$O^{\field}(t)$, even if $\field \neq 0$.

Our next result describes the approach of
the state of the system to a `Non-Equilibrium Stationary State' (NESS), in the
limit of large times. However, it is slightly weaker, because we are forced to consider
ergodic averages, since the external force field attenuates dissipative effects of the reservoirs; for a more extended discussion we refer to~\cite{paper1}.   In fact,  in statement \eqref{eq:diffconst} of Theorem \ref{thm: diffusion}, we cannot even control the ergodic average, but only the abelian average.

\begin{theorem}\label{thm: stationary}\emph{[Approach to NESS]}
There are constants
$k_\lambda, k_\field$, such that, for $0 <|\lambda| < k_\lambda$, $|\field|<
k_\field$, there exists a real-analytic function $\zeta \equiv
\zeta^{\field,\la}$ on $\bbT^d$, satisfying $\zeta \geq 0$ and $ \int_{\bbT^d} \dd k\, \zeta(k)=1$, i.e., $\zeta$ is a probability density,  such that the following statements hold for any 
exponentially localized density matrix, $\rho_\sys$,  and continuous function $f: \tor \to \bbR$:
\begin{itemize}
\item[$i.$] For $\field\not=0$, 
\begin{equation*}
\frac{1}{T}\int_{0}^{T}\,\dd t\,  \langle M_f(t)\rangle_{ \rho_{\sys} \otimes
\rho_\referres}  =     \langle  f, \zeta^{\field,\la}  \rangle_{\mathrm{L}^2(\tor)} +
\caO(1/T)\, , \qquad\textrm{ as } T \to \infty\,. 
\end{equation*}
\item[$ii.$] For $\field = 0$, 
\begin{align*}
\langle M_f(t)
\rangle_{ \rho_{\sys} \otimes
\rho_\referres}= \langle f,\zeta^{0,\lambda}\rangle_{\mathrm{L}^2(\tor)}+\caO(\e{-\lambda^2 g t})\,, \qquad\textrm{ as } t \to \infty\,,
\end{align*}
 where $g>0$ is the decay constant appearing in~\eqref{eq:2.12} and~\eqref{eq:2.13}. Moreover $\zeta^{0,\lambda}$ satisfies `time
reversal invariance'; $\zeta^{0,\la}(k) =    \zeta^{0,\la}(-k) $.
\end{itemize}
\end{theorem}
Our next result asserts that the motion of the particle is \textit{diffusive} around an average uniform motion.
 \begin{theorem}\label{thm: diffusion}\emph{[Diffusion]}
 Under  the same assumptions as in Theorem~\ref{thm: stationary}, 
\begin{equation*}
\lim_{t\to \infty}\frac{1}{t}    \langle X(t) \rangle_{\rho_{\sys} \otimes
\rho_\referres}  =v(\field)\,,
\end{equation*}
where $v(\field)$ is the `asymptotic velocity' of the particle and is given by
$v(\field)=\langle\nabla\varepsilon,\zeta^{\field,\la}\rangle$. For
$\field\not=0$, we have $v(\field)\not=0$. The dynamics of the particle is
diffusive,
in the sense that the limits
\begin{equation}\label{eq:diffconst}
D^{ij}(\field) \deq 
\lim_{T\to\infty}\frac{1}{T^2}\int_{0}^{\infty}\,\mathrm{d}t\,\mathrm{e}^{
-\frac{t}{T}}\,  \langle (X^{i}(t)-v^i(\field)t)(X^{j}(t)-v^j(\field)t)
\rangle_{\rho_{\sys} \otimes \rho_\referres}
\end{equation}
exist, where the `diffusion tensor' $D(\field)$ is positive-definite, with
$D(\field)=\caO(\lambda^{-2})$, as $\lambda\to 0$.
\end{theorem}

Note that the claim about the asymptotic velocity follows formally from
Theorem~\ref{thm: stationary} by defining the velocity operator  as
\begin{equation} \label{def: finite volume velocity}
V^j \deq\ii [ H,X^j]   = \ii [T,X^j]=M_{\nabla^j\epsilon}\,,
\end{equation}
and writing $X(t)=X(0)+\int_0^t \d s V(s)$. Although it is quite easy to make this
reasoning precise, we warn the reader that, at this point, it is
\textit{formal}, because the Heisenberg-picture observables $X(t)$ and $V(t)$ have not been constructed as
operators in the thermodynamic limit. They are formal objects appearing in correlation
functions that are constructed as thermodynamic limits of finite-volume correlation functions.

Our next result states that the equilibrium diffusion matrix $D(\field=0)$
(which is in fact a multiple of the identity matrix) is related to the response
of the particle's motion to the field $\field$. The corresponding identity is
known as the `Einstein relation':
\begin{theorem}\emph{[Einstein relation]}\label{thm: einstein}
 Under  the same assumptions as in Theorem~\ref{thm: stationary}, 
\begin{equation}\label{eq: einstein relation}
\frac{\partial}{\partial
\field^i}\bigg|_{\field=0}v^j(\field)={\lambda^2\beta}D^{ij}(\field=0)\,,
\end{equation}
where $D(\field=0)$ is defined in Equation~\eqref{eq:diffconst} and it equals
\begin{align*}
D^{ij}(\field=0)=\frac{1}{2}\int_{\R}\dd t\,\langle
V^i(t)V^j\rangle_{\rho_\beta}\,.
\end{align*}
\end{theorem}
Note that, by the positivity and isotropy of the diffusion matrix, this theorem
also shows that, for small but non-zero $\field$, $v(\field)$ does not vanish.
The origin of the unfamiliar factor $\la^2$ on the right side of \eqref{eq: einstein
relation} is found in the fact that the driving force field in the Hamiltonian
is $\la^2 \field$, rather than $\field$.


\section{Dyson expansion: The formalism}\label{section:expansions}\label{sec:
definition of expansions}

In this  section and the next one, we expand the effective dynamic $\matz$ (see
Section~\ref{effective dynamics}) and correlation functions in absolutely
convergent series. 
This task is carried out in two steps: First, we derive the expansions
without worrying about careful estimates. To avoid ambiguities concerning the
definition of operators, we do this in finite volume. This is done in the
present section, which is therefore essentially an algebraic exercise, because
the convergence of the expansions is a trivial matter. 
 In a second step, which is postponed to Section~\ref{sec: dyson analyis},
  we derive bounds on these expansions and prove their convergence uniformly in~$\Lambda$. 

In the present section, we set
$\mathscr{H}^{\La}_{\sys}=\ell^2(\Lambda), \scrH_{\res_x}^{\La} =
\Ga_s(\mathrm{L}^2(\bar \La))$, etc., as in Section~\ref{thereservoir}. To keep notations simple, we temporarily drop the superscript $\La$ everywhere. 
We start by defining `Green functions'.

\subsection{Green functions}\label{section:6.1}
The {\it (interacting) Green functions} are defined as follows:  Recall the equilibrium density matrix of
the reservoirs $\rho_{\referres}=Z_{\referres}^{-1}\e{-\beta H_{\res}}$,
$Z_{\referres}=\Tr{\e{-\beta H_{\res}}}$, and 
let $\rho$ be a density matrix on
$\mathscr{H}_{\sys}$.  We define a map
$\mathcal{Q}: \scrB_1(\scrH_\sys) \to \scrB_1(\scrH):  \rho \mapsto
\rho\otimes\rho_{\referres}$.
 Let
$I\subset\R_+$ be a finite time interval. Subsequently, we abbreviate $\inf I$
and $\sup I$ by $\infi$ and $\supi$, respectively.  The length of an interval
$I$ is
denoted by $|I|\deq \supi-\infi$. Let $\caS_1,\ldots,\caS_m$ be operators acting
on
$\scrB_2(\mathscr{H}_{\sys})$. The most relevant choice will be $\caS_i$ equal
to $(O_i)_{\links}$ or $(O_i)_{\rechts}$
for some `observables' $O_i\in\scrB(\mathscr{H}_{\sys})$, $i=1,\ldots,m$, (we
use here the left- and right-multiplications that were defined
in~\eqref{eq:2.5}). The
(interacting) Green function on $I$ is defined as the map
\begin{align*}
\mathcal{G}_I\big(\caS^{s_1}_1,\ldots, \caS^{s_m}_m\big):\scrB
_2(\mathscr{H}_{\sys})&\longrightarrow\scrB_2(\mathscr{H}_{\sys})
\end{align*}
given by\small
\begin{align}\label{eq:6.5}
\mathcal{G}_I\big(\caS^{s_1}_1,\ldots, \caS^{s_m}_m\big)(\,\cdot\,
)\deq\Tr_{\res}\bigg[\e{-\ii (\supi-s_m)\caL }\,\caS_m\e{-\ii
(s_m-s_{m-1})\caL}\caS_{m-1}\cdots \caS_1\e{-\ii (s_1-\infi)
\caL}\,\mathcal{Q}(\,\cdot\,)\bigg]\,,
\end{align}\normalsize
where $\infi\le s_1<s_2<\ldots<s_m\le\supi$, and the trace is over the
reservoir Hilbert space $\mathscr{H}_{\res}$. Here,
$\caL\deq\mathrm{ad}(H)=[H,\,\cdot\,]$ denotes the Liouvillian associated to
$H$, an essentially selfadjoint operator on $\scrB_2(\scrH)$.  The notation
$\caS_j^{s_j}$,
$j=1,\ldots,m$, merely indicates where the operator $\caS_j$ should be placed on the right side of \eqref{eq:6.5}.
In
particular, $\caS_j^{s_j}$ is {\it not} the operator $\caS_j$ time-evolved to
$s_j$.  Since the
operators carry a time label $s_j$, their order in the bracket on the left side of \eqref{eq:6.5} is irrelevant. We remark that we have defined the Green functions as operators on $\scrB_2(\scrH_\mathrm{S})$, i.e., we view density matrices as Hilbert-Schmidt operators through the embedding $\scrB_1(\scrH_\mathrm{S})\subset\scrB_2(\scrH_\mathrm{S}) $ since it is more convenient to work with the Hilbert space $\scrB_2(\scrH_\mathrm{S})$.

A special case of interest is $m=0$, i.e., when no $\caS_i$'s are present in the
Green function. We set
\begin{align*}
\mathcal{Z}_I(\,\cdot\,)\deq\mathcal{G}_I
(\emptyset)(\,\cdot\,)=\Tr_{\res}\left[\e{-\ii |I| \caL}\mathcal{Q}(\,\cdot\,)\right]\,.
\end{align*}
For $I=[0,t]$ this notation agrees with the notation for the effective dynamics
in Section~\ref{effective dynamics}.
\newline

For later purposes, we define a special class of operators on $\scrB_{2}(\scrH_{\mathrm{S}})$: We say $\caS\in\scrB(\scrB_2(\scrH_{\mathrm{S}}))$ is {\it quasi-diagonal}, whenever its kernel satisfies
\begin{align}\label{definition locally acting superop}
 |\caS(x_\links,x_\rechts; x'_\links,x'_\rechts)|\le C\e{-\nu
|x_\links-x'_\links|-\nu|x_\rechts-x'_\rechts|}\,,\qquad\textrm{for
some }\nu>0\,,\qquad (x_\links,x_\rechts, x'_\links, x'_\rechts\in \Z^d)\,.
\end{align}
Note the analogy with quasi-diagonal observables $O \in
\scrB(\scrH_\sys)$ defined in Section~\ref{sec: thermo}; in particular $\caS=
(O)_{\vs \in \{ \links,\rechts\}}$ is quasi-diagonal if and only if $O$ is quasi-diagonal.
\subsection{An expansion for $\mathcal{Z}_I$}
In this subsection, we derive an expansion for $\mathcal{Z}_{I}$,  with
$I\subset\R_+$ a finite interval. We define the free particle dynamics,
$\mathcal{U}_I$, on
$\scrB_2(\mathscr{H}_{\sys})$ by
\begin{align}\label{eq:6.2.7}
\mathcal{U}_{[t_1,t_2]}\deq\,\e{-\ii (t_2-t_1)\ad(H_{\sys})}\,,\quad
t_1,t_2\in\R_+\,,\quad H_{\sys}=T-\lambda^2\field\cdot X\,,
\end{align}
and the particle-reservoir interaction, $H_{\mathrm{SR}}(t)$, in the interaction picture, which
we decompose in spatially localized terms 
\begin{align*}
H_{\mathrm{SR}}(t) \deq \e{\ii t
H_{\res}} H_{\mathrm{SR}}  \e{-\ii t
H_{\res}} = \sum_x \lone_x \otimes \Psi_x(t)\,, \qquad   \Psi_x(t) \deq \e{\ii t
H_\res}\left(  a^*_x(\phi)+ a_x(\phi) \right) \e{-\ii t H_\res}\,.
\end{align*} 
Iterating Duhamel's formula
\begin{align*}
\e{\ii t\ad(H_{\res})}\e{-\ii t
\ad(H)}&=\mathcal{U}_{[0,t]}-\ii\lambda\int_0^t\dd
s\,\mathcal{U}_{[s,t]}\,\ad(H_{\mathrm{SR}}(s))\,\e{\ii s \ad(H_{\res})}\,\e{-\ii
s\ad(H)}\,,
\end{align*}
we find the Lie-Schwinger- or Dyson series for $\mathcal{Z}_I$:\small
\begin{align*}
\mathcal{Z}_I(\,\cdot\,)=\sum_{n\ge0}(-\ii
\la)^n\mathop\int\limits_{{\infi<t_1<\ldots<t_n<\supi}}\dd t_1\cdots\dd
t_n\Tr_{\res}\big[\mathcal{U}_{[t_n,\supi]}\ad(H_{\mathrm{SR}}(t_n))
\cdots\ad(H_{\mathrm{SR}}(t_1))\mathcal{U}_{[\infi,t_1]} \caQ(\cdot)\big]\,.
\end{align*}\normalsize
where the $n=0$ term on the right side is understood as $\mathcal{U}_I$.
We refrain from giving a proof of the (norm)-convergence of this series, since we
establish similar, but more involved, bounds in Section~\ref{sec: dyson
analyis}. We will use the shorthand notations
\begin{equation*}
\lone_{x,\vs}\deq (\lone_x)_{\vs}\,, \qquad \Psi_{x,\vs}(t) \deq(-\ii \Psi_{x}(t)
)_{\vs}\,,
\end{equation*}
for  $x\in\La$, $\vs\in\{ \links,\rechts \}$ (the left- and right
multiplications $(\,\cdot\,)_{\vs}$ were introduced in~\eqref{eq:2.5}).
In this notation, the formal Lie-Schwinger series for $\mathcal{Z}_I$ can be
rewritten as:
\begin{align}\label{eq:6.2.14}
\mathcal{Z}_I(\,\cdot\,)=&\sum_{n\ge0}(-\lambda)^{n}\sum_{
\underline{x}\in\Lambda^n}\sum_{\uvs \in\{ \links,\rechts \}^n}
\int_{\Delta_I^n}\dd\underline{t}\,
\Tr_{\res}\big[\Psi_{x_n,\vs_n}(t_n)\cdots\Psi_{x_1,\vs_1}(t_1)\rho_{\referres}
\big]
\nonumber\\[2mm]
&\qquad \qquad
\times\mathcal{U}_{[t_n,\supi]}\lone_{x_n,\vs_n}\mathcal{U}_{[t_{n-1}
,t_n]}\cdots\lone_{x_1,\vs_1}\mathcal{U}_{[\infi,t_1]}(\,\cdot\,)
\,,
\end{align}
where we use the shorthand
\begin{align} \label{eq:def integration simplex}
\int_{\Delta_I^n}\dd\underline
t\deq\delta_{n,0}+\int_{{\infi<t_1<t_2<\ldots<t_n<\supi}}\,\dd t_1\,\cdots\dd
t_n\,.
\end{align}
In a next step, we evaluate the trace over the reservoir Hilbert space in
Equation~\eqref{eq:6.2.14} using the quasi-free property of
$\rho_{\referres}$. To do so, we introduce some more notation also
applicable to the slightly more complicated expansions of general Green
functions. 
\subsection{Free Green functions, reservoir correlations and the path expansion
of $\mathcal{Z}_I$}
 Let $\caS_1,\ldots, \caS_m$ be operators acting on $\scrB_2(\scrH_\sys)$.  Let
$I\subset\R_+$ be an interval, and choose a set of times \mbox{$\infi\le s_1 < s_2 \ldots
< s_m\le
\supi$}. We define {\it free Green functions} by
\begin{align}\label{eq:6.9}
\mathcal{G}_I^{0}\big(\caS^{s_1}_1,\ldots, \caS^{s_m}_m\big)\deq 
\mathcal{U}_{[s_m,\supi]} \caS_m      \ldots    \mathcal{U}_{[s_1,s_2]} \caS_1 
\mathcal{U}_{[\infi,s_1]}\,.
\end{align}
For $m=0$, we set $\mathcal{G}_I^{0}(\emptyset)=\mathcal{U}_{I}$. As in \eqref{eq:6.5}, the time labels
$s_i$ just indicate where the operators should be placed. 

Since the operators $\lone_{x,\vs}$ often show up in combination with the
free time evolution $\mathcal{U}_t$ on $\scrB_2(\scrH_{\mathrm{S}})$, we introduce the
following shorthand notation: A {\it path}, $\varpi$, over a (closed) interval
$I\subset\R_+$ is a finite collection of triples
\begin{align*}
(x_i,\vs_i,t_i)\,,\quad i=1,2,\ldots\,,
\end{align*}
where $x_i\in\Lambda$, $\vs_i\in\{ \links,\rechts \}$ and $t_i\in I$. The number
of triples in a path $\varpi$ is denoted by $|\varpi|$. The set of all paths
over an interval $I$, referred to as  `path space', is denoted by
$\mathcal{P}_{I}$.  The free Green function associated to a path,~$\caG^0_{I}(\varpi)$, is defined as
\begin{align*}
\caG^0_{I}(\varpi)\deq\mathcal{U}_{[t_n,\supi]}\lone_{x_n,\vs_n}\mathcal{U}_{[t_{
n-1},
t_n
]}\cdots\lone_{x_1,\vs_1}\mathcal{U}_{[\infi,t_1]}\,,
\end{align*}
where $\varpi=((x_1,\vs_1,t_1),\ldots,(x_n,\vs_n,t_n))$, with $\infi< t_1<
t_2<\ldots< t_n\le \supi$, $I\subset\R_+$ an interval. 

We can evaluate the term containing the partial trace in~\eqref{eq:6.2.14} using the quasi-free property, or Wick rule, of
the reservoir
state: If $n$ is odd, this terms vanishes. Hence, we replace $n$ by $2n$
subsequently. Denote by $\mathrm{Pair}(n)$ be the set of partitions $\pi$ of the
integers
$1,\ldots,2n$ into $n$ pairs. We write $(r,s)\in\pi$ if $(r,s)$ is one of these
pairs, with the convention that $r<s$. Wick's rule states that
\begin{align}\label{eq:6.29}
\lambda^{2n}\Tr_{\res}\left[\Psi_{x_{2n},\vs_{2n}}(t_{2n})\cdots\Psi_{x_{1},
\vs_1}(t_1)
\rho_{\referres}\right]={\sum_{\pi\in
\mathrm{Pair}(\varpi)}}\zeta(\varpi,\pi)\,,
\end{align}
where $\varpi=((x_1,\vs_1,t_1),\ldots,(x_{2n},\vs_{2n},t_{2n}))$,
$\mathrm{Pair}(\varpi)\equiv \mathrm{Pair}({n})$, $|\varpi|=2n$, and
\begin{align}\label{eq:6.200}
\zeta(\varpi,\pi)\deq\prod_{(r,s)\in\pi}\lambda^2   h(t_r, t_s,\vs_s, \vs_r)
\delta_{x_r,x_s}\,,
\end{align}
where we have set, for $u,v\in\R_+$,
\begin{align}\label{eq:6.200bis}
 h(u,v ,\vs, \vs') \deq  \begin{cases} -\hat\psi(u-v)\,,&\textrm{if }
\vs=\links\,,\phantom{\rechts}\vs'=\links\,, \\
-{\hat\psi(v-u)}\,,&\textrm{if }\vs=\rechts\,,\phantom{\links}\vs'=\rechts\,,\\
\hat\psi(v-u)\,,&\textrm{if }\vs=\rechts\,,\phantom{\links}\vs'=\links\,,\\
{\hat\psi(u-v)}\,,&\textrm{if }\vs=\links\,,\phantom{\rechts} \vs'=\rechts\,,
\end{cases}
\end{align}
with $\hat \psi(t)\equiv\hat \psi^{\La}(t) $ as defined in~\eqref{eq: clear presentation finite volume correlationfunction one}.

If $n=0$, i.e., $\varpi=\emptyset$, we set the right side
of Equation~\eqref{eq:6.29} equal to one.
Integration over path space~$\mathcal{P}_I, I\subset\R_+$ an interval, is
denoted by the shorthand
\begin{align}\label{path space measure}
\int_{\mathcal{P}_I}\dd\varpi F(\varpi) \,\deq 
F_0(\emptyset)+\sum_{n\ge1}\,\sum_{\underline
x\,,\uvs}\,\int_{\Delta_I^{2n}}       \dd\underline t\,  F_n(\varpi)\,,
\end{align}
for $F=(F_n)_{n \in \bbN}$, where now $\underline x\in\Lambda^{2n}$, $\uvs\in\{ \links,\rechts \}^{2n}$. We will treat $\d \varpi$ merely as a shorthand notation, though it
is straightforward to check that $\d \varpi$  indeed defines a measure on an
appropriate measure space. In this notation, the expansion for $\mathcal{Z}_{I}$ in
Equation~\eqref{eq:6.2.14} takes the compact form
\begin{align}\label{eq:6.201}
\mathcal{Z}_{I}=\int_{\mathcal{P}_{I}}\dd\varpi\,\sum_{\pi\in
\mathrm{Pair}(\varpi)}\,\zeta(\varpi,\pi)\,\mathcal{G}_I^0(\varpi)\,,
\end{align}
which we call the {\it `path expansion'} of $\mathcal{Z}_I$.

\subsection{Path expansion of the correlation functions}
\label{section6.2.3}
With the formalism introduced in the previous subsections, it is straightforward
to derive the path expansion for the interacting Green function
$\mathcal{G}_I\big(\caS_1,\ldots, \caS_m\big)$,
{$I\subset\R_+$ an interval}, where $\caS_1,\ldots, \caS_m$ are operators acting on
$\scrB_2(\mathscr{H}_{\sys})$: We expand each propagator
$\e{-\ii (s_j-s_{j-1}) \caL}$ in Equation~\eqref{eq:6.5}
in its Lie-Schwinger series and proceed as previously. As a result 
we obtain the path expansion for the Green function:
\begin{align*}  
\mathcal{G}_I\big(\caS^{s_1}_1,\ldots, \caS^{s_m}_m\big)=\int_{\mathcal{
P}_I}\dd \varpi\,\sum_{\pi\in
\mathrm{Pair}(\varpi)}\,\zeta(\varpi,\pi)\,\mathcal{G}_I^0\big(\caS^{s_1}_1,
\ldots, \caS^{s_m}_m\,|\,\varpi\big)\,,
\end{align*}
where we use the shorthand notation
\begin{align*}
 \mathcal{G}_I^0\big(\caS^{s_1}_1,\ldots,
\caS^{s_m}_m\,|\,\varpi\big)\deq\mathcal{G}_I^0(\caS_1^{s_1},\ldots,\caS_m^{s_m}
,\lone_{x_1,\vs_1}^{t_1},\ldots,\lone_{x_{2n},\vs_{2n}}^{t_{2n}})\,,
\end{align*}
with $\varpi=((x_1,\vs_1,t_1),\ldots,(x_{2n},\vs_{2n},t_{2n}))$.
\subsection{Polymer expansions} 
In this subsection, we rearrange our
expansions in a `polymer' form. This will enable us to explore the
exponential decay of the reservoir correlation function $\hat\psi$ in the next
sections.
Given that  paths  $\varpi$ are collections of triples, we can define the `union' path $\varpi= \varpi_1 \cup \varpi_2$. 
Let us write $\min(t(\varpi))$, $\max(t(\varpi))$ for the smallest and largest time in the
path $\varpi$, respectively.  If $\max(t(\varpi_1)) < \min(t(\varpi_2))$, then we write $\varpi_1 < \varpi_2$. 

Given two pairings $\pi_1\in \mathrm{Pair}(\varpi_1)$,
$\pi_2\in
\mathrm{Pair}(\varpi_2)$,  with $\varpi_1 < \varpi_2$, 
we define $\pi = \pi_1 \cup \pi_2 \in 
\mathrm{Pair}(\varpi_1 \cup \varpi_2)$ in the obvious way, namely $(r,s) \in \pi$ if, either $(r,s) \in \pi_1$, or $(\str \varpi_1 \str+ r, \str \varpi_1 \str+ s) \in \pi_2$.  In an analogous way, we also define unions of a finite ordered set  of paths, i.e.,\ $\varpi_1 < \varpi_2 < \ldots < \varpi_l$  and pairings over them. 
Note the factorization property of the weights
$\zeta$ defined in
Equation~\eqref{eq:6.200}:
\begin{align} \label{eq:factorization of correlations}
\zeta(\varpi_1\cup\varpi_2 \cup \ldots \cup \varpi_l,\pi_1\cup\pi_2 \cup \ldots \cup \pi_l)=\zeta(\varpi_1,\pi_1)\zeta(\varpi_2,
\pi_2) \ldots \zeta(\varpi_l,\pi_l) \,.
\end{align}
We call a pairing $\pi \in \mathrm{Pair}(\varpi)$ irreducible if it can be written as a union $\pi = \pi_1 \cup \pi_2 $ with $\pi_j \in \varpi_j$ with $\varpi= \varpi_1 \cup \varpi_2$ and $\varpi_1 < \varpi_2 $. 

We call $D\deq(\varpi,\pi)$, $\varpi\in\mathcal{P}_I$, $\pi\in
\mathrm{Pair}(\varpi)$, a
{\it diagram}.  A diagram $D=(\varpi,\pi)$ is irreducible over an interval $I$ if $\pi$ is irreducible and $\infi= \min (t(\varpi)), \supi= \max(t(\varpi))$.
We define the domain of a diagram $D$ as
\begin{align*}
\mathrm{Dom}(D)\deq   \bigcup_{(r,s) \in \pi}  [t_r,t_s]\,.
\end{align*}
Note that a diagram $D=(\varpi,\pi)$ is irreducible over $I$ if and only if $\mathrm{Dom}(D)=I$.
 If a diagram is not
irreducible it can be decomposed uniquely into irreducible diagrams with
pairwise disjoint domains:
\begin{align}\label{eq:6.23}
\mathrm{Dom}(D)=\bigcup_{j=1}^l\, \mathrm{Dom}(D_j)\,,  \quad D_j\textrm{
irreducible over }I_j\subseteq I\,,
\end{align}
for some $l\ge 1$, with $D_j= (\pi_j,\varpi_j)$,  such that $\pi= \cup_j \pi_j, \varpi=\cup_j \varpi_j$ and $\varpi_1 < \varpi_2 < \ldots < \varpi_l$.   It follows that  the
diagrams
$D_j$ are ordered in the sense that $\sup \mathrm{Dom}(D_j)<\inf
\mathrm{Dom}(D_{j+1})$, $j=1,\ldots, l-1$ and that $I_j$ are intervals. 

Let $I\subset\R_+$ be an interval. We define
$\mathcal{V}_I\in\scrB(\scrB_2(\scrH_{\sys}))$ as the
sum over the `amplitudes' of all irreducible diagrams over $I$: Given a path
$\varpi\in\mathcal{P}_I$, we use the `$\delta$-function'
\begin{equation}\label{eq:6.26}
 \delta(\partial I -\partial \varpi) \deq   \delta(\min (t(\varpi))-\infi) 
\,\delta(\max (t(\varpi))-\supi)\,,
\end{equation}
 to restrict the integration over path space
$\mathcal{P}_I$ to paths having an element (i.e., a triple) at the initial time
$\infi$ and an element at the final time $\supi$ of the interval $I$. Hence,
$\mathcal{V}_I$ is
given by
\begin{align}\label{eq:6.27}
\mathcal{V}_I\deq\int_{\mathcal{
P}_{I}}\dd \varpi \, \delta(\partial I -\partial \varpi) \sum_{\substack{\pi \in
\mathrm{Pair} (\varpi)\\
\pi\,\mathrm{irreducible}}} \zeta(\varpi,\pi) \,    \mathcal{G}^0_I(\varpi)\,,   
\end{align}
where we sum only over irreducible $\pi$ (alternatively, such $\pi$ that render the diagram
$D=(\varpi,\pi)$ irreducible over~$I$). As in~\eqref{path space measure}, the $\d
\varpi$-integral and the `$\delta$-function' in~\eqref{eq:6.27} are a shorthand
notation for certain sums and integrals; the $\delta$-function in fact indicates
that the integral over the time-coordinates is over  $\Delta_I^{2n-2}$ instead
of  $\Delta_I^{2n}$.
Using the factorization property of $\zeta$ and the definition of $\mathcal{V}$
we rewrite the path expansion of $\mathcal{Z}_I$ in Equation~\eqref{eq:6.201} in
terms of
irreducible diagrams:
\begin{align}\label{eq:6.28}
\mathcal{Z}_{I}=\sum_{l\ge0}\int_{\Delta_{I}^{2l}}\dd
\underline
t\,\mathcal U_{[t_{2l},\supi]}\mathcal{V}_{[t_{2l-1},t_{2l}]}\mathcal
U_{[t_{2l-2},t_{2l-1}]}\cdots\mathcal{V}_{[t_1,t_2]}\mathcal U_{[\infi,t_1]}\,.
\end{align}
This expansion can be viewed as a one-dimensional `polymer' expansion. The
polymers correspond to connected subsets of the interval $I$ with weights
given by
$\mathcal{V}$. Two polymers `interact' via hard core exclusion taken into
account in the integration domain $\Delta_{I}^{2l}$. In
Equation~\eqref{eq:6.27},
we may consider diagrams with $|\varpi|>2$ as `excitations'. Diagrams without
`excitations', i.e., diagrams whose irreducible decomposition contains only paths $\varpi$ with ${|\varpi|=2}$, are called {\it ladder
diagrams}. The origin of the nomenclature becomes clear when one only retains
diagrams with $|\varpi|=2$ in the expansion~\eqref{eq:6.28}. In the following
sections, we will argue that the leading contributions to
$\mathcal{Z}_{I}$ arise from ladder diagrams.

The formalism developed above can also be applied to Green functions: We extend
the 
definition~\eqref{eq:6.27} by setting
\begin{align} \label{eq:def v of s}   
\mathcal{V}_I\big(\caS^{s_1}_1,\ldots,
\caS^{s_m}_m\big)\deq\int_{\mathcal{P}_I}\dd \varpi \, \delta(\partial I -\partial \varpi) \,\sum_{\substack{\pi\in
\mathrm{Pair}(\varpi)\\ \pi\,\mathrm{irreducible}}}\zeta(\varpi,\pi)\,\mathcal{G}
_I^0\big(\caS^{s_1}_1,\ldots, \caS^{s_m}_m\,|\,\varpi\big)\,,
\end{align}
for $\caS_1,\ldots, \caS_m\in\scrB(\scrB_2(\scrH_\caS))$, $\infi\le s_1<\ldots
<s_m\le \supi$. We refer to $\mathcal{V}_I$ as the `{\it dressing operator}' in
the following. For simplicity, we restrict our discussion to Green functions
involving two `observables'
$\caS_1,\,\caS_2\in\scrB(\scrB_2(\scrH_\caS))$. Given two times $s_1,s_2\in\R_+$
we
use the shorthand notation $(1)$,~$(2)$, for $\caS_1^{s_1}$, $\caS_2^{s_2}$,
respectively. Using the factorization property of $\zeta$ and the definition of
the dressing operator $\caV$ we find
\begin{align}
\mathcal{G}_I\big((1),(2)\big)=&\mathcal{Z}_{[s_2,\supi]}\caS_2\mathcal{Z}_{[s_1
,s_2]}\caS_1\mathcal{Z}_{[\infi,s_1]}+\int_{s_1}^{s_2}\dd
t_2\,\int_{\infi}^{s_1}\dd
t_1\,\mathcal{Z}_{[s_2,\supi]}\caS_2\mathcal{Z}_{[t_2,s_2]}\mathcal{V}_{[t_1,t_2
]}(1)\mathcal{Z}_{[\infi,t_1]}\nonumber\\
&+\int_{s_2}^{\supi}\dd t_2\,\int_{s_1}^{s_2}\dd
t_1\,\mathcal{Z}_{[t_2,\supi]}\mathcal{V}_{[t_1,t_2]}(2)\mathcal{Z}_{[t_1,s_1]}
\caS_1\mathcal{Z}_{[\infi,s_1]}\nonumber\\
&+\int_{s_2}^{\supi}\dd t_2\,\int_{\infi}^{s_1}\dd
t_1\,\mathcal{Z}_{[t_2,\supi]}\mathcal{V}_{[t_1,t_2]}((1),(2))\mathcal{Z}_{[
\infi,t_1]}\,,  \label{eq: cag of two}
\end{align}
and
\begin{align}
\mathcal{G}_I\big((1)\big)=\mathcal{Z}_{[s_1,\supi]}\caS_1\mathcal{Z}_{[\infi,
s_1]}+\int_{s_1}^{\supi}\dd t_2\,\int_{\infi}^{s_1}\dd
t_1\,\mathcal{Z}_{[t_2,t]}\mathcal{V}_{[t_1,t_2]}(1)\mathcal{Z}_{[\infi,t_1]}\,.  
\label{eq: cag of one}
\end{align}

\subsection{Generalized expansions}\label{negative1}
The expansion presented in the previous subsections were appropriate for initial
states of the form
 $\rho_{\sys}\otimes\rho_{\referres}$. In the present section, we replace the
factorized initial states by the coupled `Gibbs' state of the interacting system
at vanishing external force $\field$:
\begin{align*}
\rho_{\beta}=\frac{1}{Z_{\beta}}\e{-\beta H^{\field=0}}\,,\qquad
Z_{\beta}=\Tr[\e{-\beta
H^{\field=0}}]\,,\qquad H^{\field=0}=T+H_{\res}+\lambda H_{\mathrm{SR}}\,.
\end{align*}
To extend our formalism to this particular initial state, we define operators
\begin{align*}
D\deq \e{-\beta H/2}\e{\beta H_{\res}/2}\,,\qquad \caD(O)
\deq\mathrm{Ad}(D)O=DOD\adj\,,
\end{align*}
for observables $O$ in some subspace of $\scrB_2(\scrH)$.  Note that $D$ and
$\caD$ are unbounded operators, even in  finite volume $\La$.  Their use lies in
a non-commutative `Radon-Nykodim' identity:
\begin{align*}
\rho_{\beta}=   \frac{Z_{\be,\res}}{Z_{\be}} \caD (\lone \otimes \rho_\res)\,.
\end{align*}
Since we are in finite volume and we do not have periodic boundary conditions,
the operators~$D$ and~$\caD$ are not translation-invariant. If they {\it were}
translation-invariant and $O$ were a translation-invariant observable, then we
could write 
\begin{equation} \label{eq: desired translation invariance}
\Tr \left[O \rho_\beta \right] =  \frac{Z_{\be,\res} }{Z_{\be}}\Tr \left[O  \caD (\lone
\otimes \rho_\res)\right]  =    \frac{Z_{\be,\res} \str \La \str }{Z_{\be}}  \Tr
[O \caD (\lone_{{x=0}} \otimes \rho_{\referres} )]  =  \Tr
[O \caD (\eta_\be \otimes \rho_{\referres} )] \,,
\end{equation}
with  $\eta_\be\deq\frac{Z_{\be,\res} \str \La \str }{Z_{\be}} \lone_{0} $ and the assumed translation-invariance is used in  the second equality.  The rank-one operator~$\eta_\be$ is positive but not normalized, i.e., $\Tr_{\sys}[ \eta_\be ]\neq 1$, and it has a well-defined thermodynamic limit. Since translation invariance is broken only at the boundary of $\La$, Equality~\eqref{eq: desired translation invariance} is correct up to an error that vanishes in the
thermodynamic limit; (recall the definition of quasi-diagonal operators in~\eqref{definition locally acting
superop}).
\begin{lemma}\label{lemma: use of pinned} 
Assume that the operators $\caS_1,\ldots,\caS_m$ on $ \scrB_2(\scrH_\sys)$  are quasi-diagonal and translation-invariant in the sense that 
\begin{equation}
\caS_i(x_\links, x_\rechts,x'_\links, x'_\rechts )= \caS_i(x_\links+y, x_\rechts+y,x'_\links+y, x'_\rechts+y )\,,
\end{equation} 
whenever all variables are in $\La$. Then, for any $0\le s_1<s_2<\ldots<s_m$,
\begin{align}
\Tr\left[ \caS_m   \ldots \caS_2 \e{\ii (s_2-s_1)\caL} \caS_1  \e{\ii s_1\caL}  \rho_\be\right]  =\Tr\left[  \caS_m \ldots
\caS_2 \e{\ii (s_2-s_1) \caL} \caS_1  \e{\ii s_1 \caL}  \caD  ( \eta_{\beta}
\otimes \rho_{\referres} )\right] +\caO\left(\frac{\str\partial \La \str}{\str \La \str}\right)\,,  
\label{eq: guinea pig}
\end{align}
where the error term $\caO(\str\partial \La \str/\str \La \str)$ is understood as $\caO(L^{-1})$, as $L \to \infty$. Furthermore, the limit $\lim_\La \frac{Z_{\be,\res} \str \La \str }{Z_{\be}} $ exists and is finite (note that the number $\frac{Z_{\be,\res} \str \La \str }{Z_{\be}}$ appears in the operator~$\eta_\be$ in \eqref{eq: guinea pig}).
\end{lemma}
Note that \eqref{eq: desired translation invariance} is \eqref{eq: guinea pig} with $m=1$, $\caS_1=(O)_\links$ and $s_1=0$. The proof of Lemma~\ref{lemma: use of pinned} is postponed to Section~\ref{section thermodynamic limit}.

The  representation of equilibrium expectations (and correlations) on the right side of \eqref{eq: guinea pig}
 is useful for us, because it allows us to treat equilibrium correlations en
par with correlations in a state where the particle is initially localized on
the lattice. Indeed, for, e.g., $m=2$, the expression on the right side of~\eqref{eq: guinea pig} differs from the previously considered
expressions $\Tr_\sys[ \caG_{I}(\caS_2^{s_2},\caS_1^{s_1} ) (\rho_\sys)]$ only through
the particular choice $\rho_\sys= \eta_\be$ (keeping in mind that $\eta_\be$ is not normalized) and the presence of the operator~$\caD$. Our strategy will be to expand  $\caD$ in its Lie-Schwinger series 
(treating again $\la H_{\sys\res}$ as a perturbation) and to merge this
expansion with the one developed in the previous sections. Things are set up so
that the only change to be made in the framework developed above is that the
interval~$I$ is now a subset of $\bbR_\be\deq [-\beta/2, \infty)$, and the objects
$\caU_I$ and $h(s,s',\vs,\vs')$ need to be redefined whenever $I$ or
$\left\{s,s' \right\}$ have a non-zero intersection with $[-\beta/2,0)$. The
necessary generalizations are:
\begin{itemize}
\item[$i.$]{\it Free particle propagation $ \mathcal{U}_{I}$:} For any interval
$I \subset \bbR_{\be}$, we set $\caU_I\deq \caU_{I_2}\caU_{I_1}$, with $I_1\cup
I_2=I$ and \mbox{$I_1 \subset[-\beta/2,0], I_2 \subset \bbR_+$}, and we define
\begin{align}\label{eq:6.48}
\mathcal{U}_{I}(\,\cdot\,)\deq\begin{cases}
\e{-\ii  \str I \str  {H_{\sys}}}(\,\cdot\,)\,\e{\ii \str I \str 
{H_{\sys}}}\,,&\mathrm{if}\quad I \subset \bbR_+\,,   \\
\e{-  \str I \str   T }(\,\cdot\,)\,\e{-  \str I \str 
T }\,,&\mathrm{if}\quad I \subset   [-\beta/2, 0]\,. 
\end{cases}
\end{align}
\item[$ii.$]{\it Correlation function:} It is convenient to introduce the maps $m_{\pm}: \bbR_\be \to \bbC$ defined by $m_{\pm}(s)\deq s$, for $s \geq 0$, and  $m_{\pm}(s)\deq \pm \ii s$, for $s \in [-\be/2,0]$. Then we set 
\begin{align}\label{eq:extension correlation to negative}
 h(s,s' ,\vs, \vs') \deq    \si(s,\vs) \si(s',\vs')   \begin{cases}
\hat\psi(m_-(s)-m_-( s'))\,,&\textrm{if } \vs=\links\,,\phantom{\rechts}\vs'=\links\,,
\\[1mm]
{\hat\psi( m_+(s')- m_+(s))}\,,&\textrm{if
}\vs=\rechts\,,\phantom{\links}\vs'=\rechts\,, \\[1mm]
   \hat\psi(m_-( s')- m_+(s))\,,&\textrm{if
}\vs=\rechts\,,\phantom{\links}\vs'=\links\,, \\[1mm]
 {\hat\psi(m_-(s)- m_+(s') )} \,,&\textrm{if }\vs=\links\,,\phantom{\rechts}
\vs'=\rechts\,,
\end{cases}
\end{align}
where $\si(s\leq 0,\vs)= 1$  and $\si(s > 0 ,\vs)= -\ii, \ii$, for $\vs=
\links,\rechts$, respectively.

\end{itemize}
With these modifications, we can extend the definitions of $\caV_I, \caZ_I,
\caG_I$ and all relations between them; in particular~\eqref{eq:def v of s},
\eqref{eq: cag of two} and~\eqref{eq: cag of one}, remain valid. For example, we have from~\eqref{eq: guinea pig}
\begin{align}
 \Tr[O(s)\rho_\be]   &=   \Tr_\sys [(O)_{\vs} \caZ_{[-\be/2, s]} \eta_\be
]+\caO\left(\frac{\str\partial \La \str}{\str \La \str}\right)\,,
\end{align}
where $\vs= \links$ or $\vs = \rechts$; and $\caZ_{[-\beta/2,t]}$ may be decomposed as
\begin{equation}  \label{eq: relation caz cazbeta}
 \caZ_{[-\beta/2,t]} 
=   \caZ_{[0,t]}  \caZ_{[-\beta/2, 0]} +  \int_{-\beta/2}^0 \d u \, \int_{0}^t
\d v \,   \caZ_{[v, t]} \caV_{[u, v]}   \caZ_{[-\beta/2, u]}\,.
 \end{equation}
This concludes our discussion on finite-volume expansions. In the next section, we move on to discussing the thermodynamic limit and convergence of the series expansions introduced above.

\section{Dyson series: Analysis and bounds} \label{sec: dyson analyis}
{In this section, we analyze the Dyson series.  The first part, Section \ref{section thermodynamic limit}, though lengthy, contains only soft estimates that require neither the full power of Assumption \ref{ass: exponential decay}, nor the crucial fact that we consider separate reservoirs at each lattice point.  This part  could have been avoided by defining the model in the thermodynamic limit from the start.   In contrast, Section \ref{sec: bounds on effective dynamics} contains the crucial estimates that are specific to our model.}

\subsection{Thermodynamic limit}\label{section thermodynamic limit}
In the previous section, we have derived the expansion
 \begin{align}\label{eq: repeat expansion}
\mathcal{Z}_{I}=\int_{\mathcal{P}_{I}}\dd \varpi\,\sum_{\pi\in
\mathrm{Pair}(\varpi)}\,\zeta(\varpi,\pi)\,\mathcal{G}^{0}_I(\varpi)\,,
\end{align}
where $I \subset [-\be/2,\infty)$, with all objects in finite volume
$\La=\La_{L}$. 
Next, we propose to pass to the thermodynamic limit.   First, we
note that the operators $\zeta(\varpi,\pi)\,\mathcal{G}^{0}_I(\varpi)$ on the 
right side of~\eqref{eq: repeat expansion}  are well-defined as  operators
on $\scrB_2(\scrH_\sys)$, for $\La= \bbZ^d$, $\bar\La=\R^d$, respectively. 
Indeed,   the correlation functions $\zeta$ are products of the functions $\hat
\psi$, which were defined for $\La=\R^d$ in~\eqref{eq: clear presentation of
correlation function}, and the operators~$\caU_I  $  are well-defined on
$\scrB_2(\ell^2(\bbZ^d))$ because \mbox{$H_\sys= T- \la^2 \field \cdot X$}
is a self-adjoint operator. 
Finally, the shorthand  $\int \d \varpi$ contains sums over $x_i \in \La$, which
have to be interpreted now as sums over~$\bbZ^d$.
This gives meaning to $\mathcal{Z}_{I}$ as a series of
operators on $\scrB_2(\scrH_\sys)$. 
Likewise, the series for  the correlation functions $\caV_I(\caS_1^{s_1},
\ldots,\caS_m^{s_m})$ and $ \caG_I(\caS_1^{s_1}, \ldots,\caS_m^{s_m} )$ are
well-defined term by term, respectively, for $\caS_j \in
\scrB(\scrB_2(\scrH_\sys))$.    

Below, we prove convergence of these series and we establish that they are
indeed the limits of their natural finite-volume counterparts.   The main ingredient here is the thermodynamic limit of correlation functions in Assumption \ref{ass: exponential decay}. 

We introduce some notation that will also be used in the subsequent analysis.  Let
$\Lambda=\Lambda_L$ or $\Z^d$. We write~$\uw$ to denote {\it walks} in $\Lambda \times
\Lambda$, i.e., sequences $\uw= (w_0,w_1, \ldots, w_n)$, where each $w$ is of the
form $(w_\links, w_\rechts)$, with $w_\links, w_\rechts \in \Lambda$, and we
write $l(\uw)\deq n+1$ for  the length of the walk. (The walks~$\uw$ should not be
confused with the paths of triples $\varpi$.) On $\Lambda\times\Lambda$ we use the distance
$$|w|_{\Lambda\times\Lambda}=|(w_\links,w_\rechts)|_{\Lambda\times\Lambda}\deq
|w_\links|+|w_\rechts|\,,$$where $|\,\cdot\,|$ denotes the Euclidean distance on
$\Lambda$. Moreover, we set $\norm \uw \norm\deq \sum_{j=1}^n \str
w_{j}-w_{j-1}\str_{\Lambda\times\Lambda} $. For simplicity, we abbreviate
$|w|\equiv|w|_{\Lambda\times\Lambda}$ hereafter. Finally, we write $\uw: w \to
w'$ whenever $w=w_0$, \mbox{$w'=w_n$}, i.e., $\uw$ is a walk starting at $w$ and ending at $w'$.

Subsequently, we use the above notation of walks for kernels of operator acting on~$\scrB_{2}(\scrH_{\mathrm{S}})$: Instead of writing
$\caS(w_\links,w_\rechts, w_\links',w_\rechts')$, we write $\caS(w,w')$, with $w=(w_\links,w_\rechts)$, $w'=(w_\links',w_\rechts')$, for $\caS\in\scrB(\scrB_2(\scrH_{\mathrm{S}}))$. For example, the definition of quasi-diagonal operators in~\eqref{definition locally acting superop}, reads $\str\caS(w',w)\str \leq C
\e{-\nu\str w'-w \str} $, for some $\nu>0$, in this notation.
\newline

Before we are able to analyze the thermodynamic limit, we still have to define what is the
`natural finite volume analogue' of $\caV_I(\caS_1^{s_1}, \ldots,\caS_m^{s_m}),
\caG_I(\caS_1^{s_1}, \ldots,\caS_m^{s_m} )$ when starting from infinite volume
$\caS_j$-operators. We set
  $$ \caS^{\La}_j\deq \mathrm{Ad}(\lone_{\La}) \caS_j
\mathrm{Ad}(\lone_{\La})\,,$$
where $\lone_\La$ is the orthogonal projection  $\ell^2(\bbZ^d) \to
\ell^2(\Lambda)$.  

\begin{lemma}\label{lem: thermo on kernels}
Let $\caA$ be one of the following (infinite volume) operators
\begin{equation} \label{eq: things to be tl constructed}
\caZ_I\,,   \qquad \caG_I(\caS_1^{s_1}, \ldots, \caS_m^{s_m})\,, \qquad 
\caV_I(\caS_1^{s_1}, \ldots, \caS_m^{s_m})\,,
\end{equation}
with $I \subset [-\be/2, \infty)$, and $\caS_1,\ldots,\caS_m$ quasi-diagonal.
Denote by $\caA^{\La}$ their finite-volume analogues as discussed above, in
particular including the replacement of $\caS_j$ by $\caS_j^{\La}$. 
Then, for $\La$ finite and for $\La=\bbZ^d$, the series defining these operators
converge absolutely in norm (as operators on $\scrB_2(\scrH_\sys)$).  Moreover,
\begin{equation*}
\str \caA^{\La}(w,w')  \str \leq C_I  \e{-c \str w'-w \str}\,, \qquad \str
\caA(w,w')  \str \leq C_I  \e{-c \str w'-w \str}\,,
\end{equation*}
where the constant $C_I$ can be chosen uniformly in~$\La$ (for $|\La|$ large
enough), including $\La=\bbZ^d$, independent of the times $s_1,\ldots,
s_m$ and uniform on compacts in $\infi,\supi$. The exponent $c>0$ can be chosen to depend only on the
observables $\caS_1,\ldots, \caS_m$.
Finally,  for any $w,w'$,
\begin{equation*}
\tdl  \caA^{\La}(w,w') =   \caA(w,w')\,, 
\end{equation*}
uniformly on compacts in the variables $s_1,\ldots, s_m$ and $\infi, \supi$.
 \end{lemma}
 
 \begin{proof}
The following bounds apply alike in finite and infinite volume. Denoting by $\mathop{\sum}\limits_{\underline
x,\uvs}^{\varpi}$ the sum over the $x$- and $\vs$-coordinates in the path
$\varpi=((x_1,\vs_1,t_1),\ldots,(x_{2n},\vs_{2n},t_{2n}))$, we obtain, for
$\nu>0$ sufficiently small, and uniformly in $\La$, 
\begin{align}
\sum_{\underline{x},\uvs}^{\varpi} \str\mathcal{G}_I^0\big(\,\varpi\big) (w,w') \str  &
\leq   2^{\str \varpi\str}   \sum_{\footnotesize{\left.\begin{array}{c}   \uw: w
\to w' \\  l(\uw)=\str \varpi\str+2 \end{array}\right.}}   \prod_{j=1}^{\str \varpi
\str+1} \str \caU_{[t_{j-1},t_{j}]}(w_{j-1}, w_j) \str \nonumber  \\[2mm]   
&\leq     C^{1+ |\varpi|} \e{ \str I \str \norm \mathrm{Im} \varepsilon \norm_{\infty,2\nu}} 
\sum_{\footnotesize{\left.\begin{array}{c}   \uw: w \to w' \\  l(\uw)=\str
\varpi\str+2 \end{array}\right.}} \prod_{j=1}^{|\varpi|+1} \e{-2\nu \str w_{j} - w_{j-1}
\str}\nonumber  \\[2mm]
&\leq  C  \e{C \str I \str |\nu|}     \e{-\nu \str w - w' \str}     (C \nu^{-2d})^{\str
\varpi \str }\,,   \label{eq: basic estimate cagzero}
\end{align}
where we introduced the dummy variables $t_0=\infi, t_{\str\varpi\str+1}=\supi$ in
the first line. To get the second line, we used the propagation bound \eqref{eq:
propagation bound} and its imaginary time version \[\norm \e{\tau T}(x,x')
\norm \leq C \e{C\str\tau\str (1+ \str \nu \str)} \e{-2\nu\str x-x'\str},\]
for $\tau \leq \beta/2$.
The term between brackets on the last line of~\eqref{eq: basic estimate cagzero}
originates from the sums 
$
\sum_{w}\e{- \nu\str w \str }  
$.

From~\eqref{eq: repeat expansion} and \eqref{eq: basic estimate cagzero} we get
 \begin{align}
 \str \caZ_I(w,w') \str   &\leq    \sum_{n \geq 0} \sum_{ \pi\in
\mathrm{Pair}(n)}  \int \d t_1 \cdots \d t_{2n}   
\sum_{\underline{x},\uvs}^{\varpi}   \str \zeta(\varpi,\pi) \str     
\str\caG_I^0(\varpi)(w,w') \str  \label{eq: from cag to caz}\nonumber  \\
 &\leq    C  \e{-\nu \str w - w' \str}      \e{ C \str I \str \str\nu \str  }   
\sum_{n \geq 0}   \str \mathrm{Pair}(n)  \str     (C \nu^{-2d})^{2n }   
\int \d t_1 \cdots \d t_{2n}  C^n      \nonumber   \\ 
  &\leq    C  \e{-\nu \str w - w' \str}      \e{ C \str I \str \str\nu \str  }  
\sum_{n \geq 0}    {(2n-1)!!}   \frac{ C^n  \str I \str^{2n}}{(2n) !} \nonumber  \\ & \leq   
C  \e{-\nu \str w - w' \str}      \e{C \str I \str \str\nu \str + C \str I \str^2 } \,,
 \end{align} 
where $(2n-1)!! \equiv  (2n-1)(2n-3) \cdots 1$ is the number of pairings of  $2
n$ elements, i.e.,\ $\str \mathrm{Pair}({n})  \str  $, and the $n=0$ term is understood as $\str \caU_I(w,w') \str$.  Note that we have used Assumption~\ref{ass: exponential decay} to get $ \str \zeta(\varpi,\pi) \str
\leq C^n $. 
 A similar estimate holds for $\caG_I(\cdot), \caV_I(\cdot)$, since ultimately we
only used that the operators $\caU_I$ satisfy the propagation estimate $\str
\caU_I(w,w') \str \leq C_I \e{-\nu \str w -w'\str}$; see~\eqref{eq: propagation
bound}.  In the terminology introduced above, this means that $\caU_I$ is quasi-diagonal. Since the $(\caS_j)$ are also quasi-diagonal, possibly with a different exponential decay rate $c>0$, the proof still applies.  
 
 We now turn to convergence of kernels.  By the uniform bounds above, the
convergence of kernels follows once we have proved that $\caU^{\La}_I(w,w')\to
\caU_I(w,w') $ and $\zeta^{\La}(\varpi,\pi) \to \zeta(\varpi,\pi) $ for any $w,w',\pi,
\varpi$, and uniformly on compacts in the time-arguments.  The first claim is
obvious because $H_\sys^{\La} \to H_\sys$ strongly, moreover $H_\sys^{\La}$ is
bounded for finite $\La$. Thus, functions of $H^{\La}_\sys$ converge
strongly to functions of $H_\sys$.  The second claim follows from Assumption \ref{ass: exponential decay} since 
$\zeta(\varpi,\pi)$ is a product of the correlation functions~$\hat \psi(t),  t \in \bbH_\be$.
 \end{proof}
 
 We now move towards the proof of Lemma \ref{lemma: use of pinned}. We recall that both the particle and the reservoirs are restricted to finite volumes~$\La$ and~$\bar \La$  that are related by~$\La= \bar \La \cap \bbZ^d$. However, since each lattice point is connected to a separate reservoir, there is no compelling reason for these volumes to be related and this is exploited in the present proof: We first perform the thermodynamic limit for the particle ($\La \to \bbZ^d$) but not for the reservoirs ($\bar \La$~remains finite).  To that end, we introduce operators~$\widetilde\caA$ on $\scrB_2(\ell^2(\bbZ^d))$, with $\caA=\caZ_I,\caV_I,\caG_I(\caS^{s_1}_1,\ldots, \caS^{s_m}_m)$, for~$\caS_j$ translation-invariant and quasi-diagonal, and  interval $I \subset [-\beta/2, \infty)$.  These `tilde operators' are obtained from the ones without tildes by choosing the correlation functions~$\zeta$ to be~$\zeta^{\La}$, i.e., in finite volume, but choosing the operators $\caU_I$ and $\caS_j$ in infinite volume.  
Therefore, the operators~$\widetilde\caA$ are 
translation-invariant. Note that given a collection of translation-invariant and quasi-diagonal operators~$\caS_j$, we have now three types of objects, namely $\caA^{\La},\widetilde\caA,\caA$. The latter of the three does not play a r\^{o}le in the proof of Lemma \ref{lemma: use of pinned}. 
\begin{lemma}\label{lem: exp conv} Let $\caA^{\Lambda}$ and $\widetilde{A}$ be as described above. Then
\begin{equation} \label{eq: exp speed of convergence}
\left\str \caA^{\La}(w,w') -  \widetilde\caA(w,w') \right\str \leq C_I \e{-c' d(\{w_l,w_r\},\La^c)}  
\e{-c \str w-w'\str}\,,
\end{equation}
where $d(A,B)= \inf_{x \in A,y \in B} \str x-y\str$ for $A,B \subset \La$, and $\La^c=\bbZ^d \setminus \La$. The constants $c,c'$ can be chosen uniformly in $I, s_1,\ldots, s_m$ and $\caS_1,\ldots, \caS_m$.
 \end{lemma}
 \begin{proof}
We decompose
 \begin{equation}
 H_\sys = H^{\La}_\sys  +   T^{\partial \La} +  H^{\La^c}_\sys\,,
 \end{equation}
 where the operator on the left side corresponds to infinite volume, $H^{\La^c}_\sys$ acts on $\ell^2(\La^c)$ and  
 \begin{equation}
 T^{\partial \La}(x,x') = \begin{cases} \hat\epsilon(x'-x) &  x \in \La, x' \in \La^c \quad \textrm{or} \quad x' \in \La, x \in \La^c  \\    0 & \textrm{otherwise}  \end{cases}.
 \end{equation}
Duhamel's principle yields
 \begin{equation}
\lone_{\La}\left( \e{-\ii t H_\sys} -  \e{-\ii t H_\sys^{\La}} \right)  \lone_{\La}=  \int_0^t \d s \,     \lone_{\La}  \e{-\ii (t-s) H_\sys}   T^{\partial \La}  \e{-\ii s H_\sys^{\La}}  \lone_{\La}\,.
 \end{equation}
 Applying the propagation bound for both $\e{-\ii t H_\sys^{\La}}$ and $\e{-\ii t H_\sys}$ and the exponential decay of the function~$\hat \epsilon(\,\cdot\,)$, we conclude that 
 \begin{equation}
 \left\str ( \e{-\ii t H_\sys} -  \e{-\ii t H_\sys^{\La}} )   (x,x') \right\str  \leq  C \e{Ct } \e{-c \str x-x' \str}  \e{-c' \max(d(x, \La^c), d(x', \La^c) )}, \qquad   x,x' \in \La \,,
 \end{equation}
 and hence we also have
 \begin{equation} \label{eq: thermo prop u}
 \str \caU^{\La}_{I}(w, w') - \caU_{I}(w, w') \str  \leq C_I \e{-c'd(\{w_l,w_r 
\},\La^c)}  \e{-c\str w -w'\str} \,.
 \end{equation}
 Obviously, we can equally well choose $\{w_l',w_r'\}$ instead of $\{w_l,w_r\}$ on the right side of~\eqref{eq: thermo prop u}. Furthermore,  we can also replace the difference 
 $\caU^{\La}_{I}- \caU_{I}$ on the left side by  $\caS^{\La}- \caS$  for quasi-diagonal $\caS$, and we can allow $I \subset [-\be/2, \infty)$. 
 
To address $\caA^{\La}-\tilde \caA$ as required, we recall that the correlation functions $\zeta$ in both operators are the same (i.e.\ the finite-volume ones) so that the only differences originate in the difference on the left side of~\eqref{eq: thermo prop u} (or the generalizations just mentioned).   By repeatedly applying \eqref{eq: thermo prop u} and using the same strategy as in Lemma \ref{lem: thermo on kernels} to sum/integrate over $\varpi$, we get the claim of the Lemma. 

\end{proof}

 \begin{proof}[Proof of Lemma \ref{lemma: use of pinned}]

We abbreviate $g(\La)\deq\frac{Z_{\be,\res} \str \La \str }{Z_{\be}}$.  Recall that
 \begin{equation}
\rho_\be= g(\La) \frac{1}{\str \La \str} \caD(\lone \otimes \rho_\res)=g(\La)
\frac{1}{\str \La \str}    \sum_{x \in \La} \caD(\lone_x \otimes\rho_\res)\,.
\end{equation} 
Therefore, the left side and right side of \eqref{eq: guinea pig} may be written as 
\begin{equation} \label{eq: two necessary sides}
g(\La) \frac{1}{\str \La \str} \sum_{x \in \La} \Tr\left[ \caA^{\La}(\lone_x )\right],\qquad 
g(\La)  \Tr\left[ \caA^{\La}(\lone_0)\right]\,,
\end{equation}
respectively, with $\caA^{\La}= \caG^\La_{[-\be/2,s_m]}(\caS_m^{s_m},\ldots, \caS_1^{s_1} )$. Using kernels, we recast
\begin{align} \frac{1}{\str \La \str}\sum_{x \in \La} \Tr [ \caA^{\La}(\lone_x) ] -  \Tr [ \caA^{\La}(\lone_0) ]  
  = &  \frac{1}{\str \La \str}  \sum_{x,y \in \La} \caA^{\La}(y,y,x,x)-  \sum_{y \in \La} \caA^{\La}(y,y,0,0)\,. \label{eq: diff of kernels}  
\end{align}
Since the $(\caS_j)$ were assumed to be translation-invariant on $\La$ and quasi-diagonal, they are finite-volume restrictions of truly translation-invariant and quasi-diagonal operators. Therefore,  Lemma~\ref{lem: exp conv} applies to the operator~$\caA^{\La}$ above and the associated~$\widetilde \caA$ are translation-invariant.  Let us split
 $\caA^\La=\widetilde \caA+ \caK^\La$ in~\eqref{eq: diff of kernels} with $\caK^\La\deq\caA^\La-\widetilde \caA$ .  By translation invariance,  we can drop the terms containing $\widetilde \caA$, so \eqref{eq: diff of kernels} equals
 \begin{equation}
 \frac{1}{\str \La \str}  \sum_{x,y \in \La} \caK^{\La}(y,y,x,x)-  \sum_{y \in \La} \caK^{\La}(y,y,0,0)\,.
 \end{equation}
By the bounds of Lemma~\ref{lem: exp conv}, this difference is bounded by $C
\frac{\str\partial \La \str}{\str \La \str} $ and hence we obtain
 \begin{equation}
 \left\str \frac{1}{\str \La \str}\sum_{x \in \La} \Tr [ \caA^{\La}(\lone_x) ] -  \Tr [ \caA^{\La}(\lone_0) ]   \right\str \leq \frac{1}{\str \La \str}  \sum_{x \in \La} C\e{-c d(x,\La^c)}  \leq  C
\frac{\str\partial \La \str}{\str \La \str} \,.\label{eq: bound boundary}
 \end{equation}
  Recall that we need to prove that the difference between the two expressions in~\eqref{eq: two necessary sides}  is $\caO(\str\partial \La\str/\str \La \str )$. This follows from~\eqref{eq: bound boundary} provided that $g(\La)$ remains bounded as $\La \to \bbZ^d$, as we show now: By an application of the Golden-Thompson inequality and exploiting the fact that $c\lone \leq \e{\pm \be T^{\La}} \leq C \lone$, for $0<c<C<\infty$, uniformly in~$\La$,  it suffices to check that 
  \begin{equation} \label{eq: bound for free energy}
 0<c'\le \frac{\Tr \e{-\be (H_{\res}+\lambda H_{\mathrm{SR}})}  }{\str \La \str \Tr_\res \e{-\be H_{\res}}}   \le C'<\infty\,,
  \end{equation}
  uniformly in $\La$. The operators in the exponent can be explicitly diagonalized (they are polynomials of order 2 in creation and annihilation operators) and~\eqref{eq: bound for free energy} follows then after a straightforward calculation by the bound on~\eqref{eq: gs shift}. Note for further reference that this also shows that $1/g(\La)$ is uniformly bounded as $\La \to \bbZ^d$.

 It remains to show that $\lim_{\La} g(\La)$ exists. 
  Since $\Tr [\rho_\be]=1$, we have
\begin{equation}
\frac{1}{g(\La)} =  \frac{1}{\str \La \str} \sum_{x \in \La}  \Tr \left[\caD(\lone_x \otimes\rho_\res)\right]\,,
\end{equation}
 and by \eqref{eq: bound boundary}, taking now $\caA^{\La}=\caZ^{\La}_{[-\be/2,0]}$, we get
\begin{equation}\label{eq: normalization of eta}
\frac{1}{g(\La)} =    \Tr[\caZ^{\La}_{[-\be/2,0]}(\lone_0)] + \caO\left(\frac{\str\partial \La \str}{\str \La \str}\right)\,.
\end{equation}
By Lemma \ref{lem: thermo on kernels}, the limit $\lim_\La \Tr[\caZ^{\La}_{[-\be/2,0]}(\lone_0)]$ on the right side exists and is finite.  Hence, either $g(\La)$ diverges or $\lim_{\La}g(\La)$ exists and is finite. But the first possibility was excluded above, hence the proof is complete.
\end{proof}

To construct the correlation functions~(\ref{def: thermo limit loc},~\ref{def:
thermo limit beta}) with $O \in \mathfrak{X}$, ($\mathfrak{X}$ the $*$-algebra
generated by the position operator $(X^j)$), we  want to consider   $\caS_j =
f_j(X)_{\vs}$, $\vs=\links,\rechts$, for some polynomials~$f_j$. In this case,
one can no longer expect $\caV_I(\cdot)$ or $\caG_I(\cdot)$ to be a bounded operators; but their kernels are obviously well-defined, and we can still follow the proof of
Lemma \ref{lem: thermo on kernels}, bounding $\str\caS_j(w',w)\str \leq C\delta_{w,w'} \str w
\str^{N_j}$, where~$N_j$ is the degree of the polynomial~$f_j$.  We obtain 
\begin{equation*}
\str \caA^\La(w,w') \str, \str \caA(w,w') \str \leq  C \str w \str^{N} \e{-\nu
\str w -w'\str}\,,
\end{equation*}
uniformly in $\La$, for some $N$ that is determined by the $\caS_1, \ldots,
\caS_m$.
This polynomial growth in $\str w \str$ is compensated by the exponential decay,
so that, for exponentially localized $\rho_\sys$, (see~\eqref{exponentially localized density matrix}),  expressions such as
 $\caG_I(\cdot ) \rho_\sys,  \caG_I(\cdot  ) \rho_\sys$, etc.\ are again 
exponentially localized operators.  For example, the following identity (trivial
in finite volume) holds for exponentially localized $\rho_\sys$ and interval $I \subset \bbR_+$:
 \begin{equation}
\partial_\field  \caG_I ( \caS_{1}^{s_1}, \ldots,   \caS_{m}^{s_m«}) \rho_\sys 
=    \lambda^2\ii\sum_{\vs= \links, \rechts}  (\delta_{\vs, \links}- 
\delta_{\vs, \rechts})   \int_{I} \d s \, 
\caG_I (   (X_l)^{s}  ,\caS_{1}^{s_1}, \ldots ,  \caS_{m}^{s_m«})  \rho_\sys
\,.\label{eq: examplebound1}
\end{equation}
Another useful  identity is obtained by applying the relation
$X(t)-X(0)=\int_0^t \d s\, V(s)$, with~$V$ the velocity operator, in correlation functions,
e.g.,
\begin{align} \label{eq: manipulations allowed}
\langle    (X(t)-X(0))^2 \rangle_{\rho_\sys\otimes \rho_\referres} &=  \int_0^t
\d s_2\, \int_0^t \d s_1 \langle  V(s_2)\, V(s_1)  \rangle_{\rho_\sys\otimes
\rho_\referres} \,,\\
\langle (X(t)-X(0))\rangle_{\rho_\sys\otimes\rho_\referres}&=\int_0^t \d
s_1\langle V(s_1)  \rangle_{\rho_\sys\otimes \rho_\referres} \,,
\end{align}
 for  exponentially localized $\rho_\sys$. Here, the left sides should be
interpreted as linear combinations of correlations functions, e.g., \small
\begin{align*}
 \langle    (X(t)-X(0))^2 \rangle_{\rho_\sys\otimes \rho_\referres} =\langle(
X(t))^2 \rangle_{\rho_\sys\otimes \rho_\referres} - \langle
X(t)X(0)\rangle_{\rho_\sys\otimes \rho_\referres}-\langle X(0)X(t)
\rangle_{\rho_\sys\otimes \rho_\referres} +\langle (X(0))^2
\rangle_{\rho_\sys\otimes \rho_\referres} \,. 
\end{align*}
 \normalsize
 
 \begin{proof}[Proof of Lemmas~\ref{lem: thermodynamic dyn} and~\ref{lem:
thermodynamic obs}]
 The thermodynamic limit of $\caZ_I$  is immediate from Lemma~\ref{lem:
thermo on kernels}, by the convergence of kernels and the exponential bounds. 
 To deal with correlation functions, we note that, for finite $\La$,
 \begin{equation} \label{eq: bais equality corr and cag}
 \Tr_{\sys}[ \caG_{[0,t]}(\caS_1^{s_1}, \ldots,  \caS_m^{s_m})\rho_\sys]   =  
\langle O_m(s_m) \ldots O_1(s_1) \rangle_{\rho_\sys \otimes \rho_{\referres}}\,,
 \end{equation}
 with $\caS_j= (O_j)_{\links}$, $0 \le s_1 < \ldots <s_m \le t$.  Using the existence of the
thermodynamic limit for $\caG_I(\cdot)$ and the exponential bounds, we show the
convergence for exponentially localized $\rho_\sys$ and $O_j \in
\mathop{\mathfrak{A}}\limits^{\circ}$ (i.e., $O_j$ quasi-diagonal) or~$O_j \in
\mathfrak{X}$, thus defining the right side of \eqref{eq: bais equality
corr and cag} for $\La=\bbZ^d$.  The extension to  $\mathfrak{A}$ is by density. (For \mbox{$m=2$}, the ordering of the times can be relaxed on the right side of \eqref{eq: bais equality corr and cag} by setting $\caS_1=
(O_1)_{\rechts}$, which leads to an exchange of $O_1(s_1)$ and $O_2(s_2)$ on the
right side.)

 For equilibrium correlation functions, we first recall  that the rank-one operator $\eta_\be$ is well-defined in the
thermodynamic limit by Lemma \ref{lemma: use of pinned}. Then, the argument is analogous
to the one above, but replacing \eqref{eq: bais equality corr and cag} by
  \begin{equation} \label{eq: bais equality corr and cag eq}
 \Tr_{\sys}[ \caG_{[-\beta/2,t]}(\caS_1^{s_1}, \ldots,  \caS_m^{s_m})\eta_\be ]  =  
\langle O_m(s_m) \ldots O_1(s_1) \rangle_{\rho_\be}\,.
 \end{equation}
We set $\field=0$, i.e., $O_j(s) $ replaced by $O^{\field=0}_j(s)$ and we consider
\eqref{eq: bais equality corr and cag eq} with $m=2$. The time-reversal
invariance and stationarity of~\eqref{eq: bais equality corr and cag eq} follow
from the finite volume system, where they are explicit. The only thing left to prove is the
infinite-volume KMS condition: We note that the construction of~\eqref{eq: bais
equality corr and cag eq} can be carried out when~$s_1-s_2$ is in the strip
$\bbH_\be$ and the conclusions  of Lemmas~\ref{lemma: use of pinned} and~\ref{lem: thermo on kernels} remain
valid.  This is checked straightforwardly by using that the operators $\e{-{\ii}s H_\sys}$ remain quasi-diagonal for $s \in \bbH_\be$ (in fact, for any $s \in \bbC$) and the
correlation functions~$\zeta(\varpi,\pi)$ remain well-defined, because, upon taking
$0 \leq \im (s_{1}-s_2) \leq \be$, all arguments of the function~$\hat \psi$ in~\eqref{eq:extension correlation to negative} remain in the strip $\bbH_\be$, as
one verifies by inspection.  Hence, the thermodynamic limit
 is still valid for~$s_1-s_2 \in \bbH_\be$ in the sense that the
correlation functions $\langle O^{\La}_2(s_2) O^{\La}_1(s_1)
\rangle_{\rho^\La_{\be}}  
$ are bounded 
uniformly in $\La$ and converge uniformly on compact sets in $s_1-s_2 \in
\bbH_\be $. Therefore, the limit of the finite-volume correlation function is
analytic in the interior of the strip and continuous on the boundary. Thus, the
KMS condition in infinite-volume follows from the one in finite
volume.
 \end{proof}

\subsection{Bounds on the effective dynamics} \label{sec: bounds on effective dynamics}
Up to now, we have established bounds on the free correlation functions
$\caG^{0}_I$, from which we could derive crude bounds on the interacting
correlation functions. This is sufficient to prove the existence of the
thermodynamic limit. In what follows, we prove sharper bounds on the
interacting correlation functions in infinite volume, using the decay
properties of the reservoir correlation function.   From now on, all quantities refer to infinite volume, unless mentioned otherwise.
\subsubsection{Bounds on diagrams}\label{section: bounds on diagrams}
In analogy to the operator $\caV_I$ defined in  \eqref{eq:def v of s},  we set,
for $n\in\N$,
\begin{equation}  \label{def: cavn}
\caV_I^{(n)}(\caS_1^{s_1}, \ldots,  \caS_m^{s_m})  \deq   \int_{\str \varpi \str
\geq 2 n} \d \varpi  \,  \delta(\partial I -\partial \varpi) \sum_{\substack{\pi \in
\mathrm{Pair} (\varpi)\\
\pi\,\mathrm{irreducible}}} \zeta(\varpi,\pi) \,   \mathcal{G}_I^0\big(\caS_1^{s_1}, \ldots, 
\caS_m^{s_m}\,|\,\varpi\big)  \,,
\end{equation}
with $\caV_I^{(2)}=\caV_I$, because the smallest irreducible diagrams have
$\str \varpi \str=2$. 
In Lemma~\ref{lem: bounds on diagrams}, we provide a bound on the right side. To save writing, we  introduce 
\begin{equation} \label{def: tildeh}
\tilde h(t) \deq     
 \sup_{\substack{ \mathrm{Re}(s'-s)=t  \\
s-s'  \in \bbH_\be, \vs_1,\vs_2 \in \{
\links, \rechts\} }}
   \str h(s,s', \vs_1,\vs_2) \str\,.
\end{equation}
 From  Assumption~\ref{ass: exponential decay}, we get 
\begin{equation*}
 \str \tilde h(t) \str \leq C \e{-g_\res \str t\str}\,.
\end{equation*}
\begin{lemma}\label{lem: bounds on diagrams}
For  sufficiently small $\la, \nu >0$, for any $n \geq 1$,  any interval  $I\subset
[-\beta/2,\infty)$ and an arbitrary collection (possibly empty), $ \caS_1, \ldots,\caS_m $, of observables with
associated times $s_1,\ldots, s_m \in I$, 
\begin{align} \label{eq: bound wpaths}
   \left\str \caV_{I}^{(n)}(\caS_1^{s_1}, \ldots,\caS_m^{s_m} ) (w,w')
\right\str& \leq C^{m+n+1}   \la^{2n}  \max(1,\str I \str^{n-1})   \e{- \tfrac{g_\res}{3} 
\str I \str }\nonumber \\[2mm]
  & \qquad \times 
 \sum_{\footnotesize{\left.\begin{array}{c}   \uw: w \to w' \\  l(\uw)=2m+2
\end{array}\right.}}      \e{-\nu \norm \uw \norm }  \prod_{j=1}^m     
\left\str \caS_{ {j}}(w_{2j-1},w_{2j}) \right\str     \e{\nu \str  w_{2j-1} -w_{2j}
\str} \,,
\end{align}
where the constant $C$ depends only on $\tilde h$, the particle dispersion relation $\ve$, and the spatial dimension $d$.
\end{lemma}
\begin{proof}
Recalling the definition of the integration measure $\dd\varpi$ on $\caP_{I}$ in~\eqref{path space measure}, we perform the integration on the right hand of~\eqref{def: cavn}, by first fixing the number of triples in $\varpi$, ($|\varpi|=2p$, $p\ge n$), and the time coordinates in~$\varpi$, ($t_1,\ldots,t_{2p})$, while summing over the spatial- and $\uvs$-coordinates.
The times $s_1,\ldots, s_m$ induce a partition of the time interval $I$ into $m+1$ intervals $I^{(j)}$, $j=0,\ldots,m$,
and, almost surely with respect to the fixed times $t_1,\ldots,t_{2p}$, also a
partition of the path $\varpi \in \caP_I$ into subpaths $\varpi^{(j)}$ (they can be
empty, i.e., $\varpi^{(j)}=\emptyset$).  We denote by $\mathop{\sum}\limits^{\varpi^{(j)}}_{\underline{x},\uvs}$ the
sum over the $\ux$- and $\uvs$-coordinates of $\varpi^{(j)}$. 
For any one of those $I^{(j)},\varpi^{(j)}$,
we use~\eqref{eq: basic estimate cagzero}  with $\nu$ replaced by $\tilde \nu$, to obtain, for any  $0 \leq \nu \leq \tilde \nu$,

\begin{align}
\sum_{\underline{x},\uvs}^{\varpi}
\str&\mathcal{G}_I^0\big(\caS_1^{s_1},\ldots,\caS_{m}^{s_m}\,|\,\varpi\big) (w,w') \str\nonumber \\
&\leq   \sum_{\footnotesize{\left.\begin{array}{c}   \uw: w \to w' \\  l(w)=2m+2
\end{array}\right.}}  
\sum^{\varpi^{(0)}}_{\underline{x},\uvs
}\str\mathcal{G}_{I^{(0)}}^0\big(\,\varpi^{(0)} \big) (w_0,w_1) \str 
   \bigg(   \prod_{j=1}^m
\str\caS^{j}(w_{2j-1},w_{2j}) \str  \sum^{\varpi^{(j)}}_{\underline{x},\uvs}
\str\mathcal{G}_{I^{(j)}}^0\big(\,\varpi^{(j)} \big) (w_{2j},w_{2j+1})\str\bigg) 
\nonumber \\[2mm]
&\leq   \sum_{\footnotesize{\left.\begin{array}{c}   \uw: w \to w' \\  l(w)=2m+2
\end{array}\right.}}   \bigg( \prod_{j=0}^m   C(C \tilde\nu^{-2d})^{\str \varpi^{(j)} \str }  \e{-\tilde\nu \str w_{2j+1}- w_{2j} \str } \e{C \str \tilde\nu \str \str I^{(j)} \str} \bigg) \bigg(
 \prod_{j=1}^m   \str\caS^{j}(w_{2j-1},w_{2j}) \str      \bigg)
\nonumber \\[2mm]
&\leq    C^{m+1} (C \tilde\nu^{-2d})^{\str \varpi \str }  \e{C \tilde\nu \str I \str}   
 \sum_{\footnotesize{\left.\begin{array}{c}   \uw: w \to w' \\  l(w)=2m+2
\end{array}\right.}}   \e{-\nu \norm \uw \norm}  \bigg(
 \prod_{j=1}^m   \str\caS^{j}(w_{2j-1},w_{2j}) \str   \e{\nu \str w_{2j}- w_{2j-1} \str }    \bigg)\,.
\label{eq: bound xensv}
\end{align}
This way, we have
bounded the $x$-~and $\vs$-sums in \eqref{def: cavn}, so that we are left with
the $\ut$-integral, $\pi$-sum and the sum over $p\ge n$.   More precisely, to pass from  \eqref{eq: bound xensv} to \eqref{eq: bound wpaths}, it suffices to show
\begin{equation}
 \e{C \tilde\nu \str I \str}    \sum_{p\geq n}  (C \tilde\nu^{-2d})^{2p }   \int_{\Delta^{2p}_I} \d \ut   \sum_{\substack{\pi \in
\mathrm{Pair} (\varpi)\\
\pi\,\mathrm{irreducible}}} \,  \sup_{\ux, \underline{\varsigma}} \str \zeta(\varpi,\pi)\str         \leq  (C \la)^{2n}  \max(1,\str I \str^{n-1})    \e{- \tfrac{g_\res}{3} 
\str I \str }\,, \label{eq: sup sum}
\end{equation}
for some sufficiently small $\tilde \nu$, and all sufficiently small $\la$, depending on $\tilde \nu$.  The supremum in the left hand expression is over all $x$-,$\varsigma$-coordinates of~$\varpi$.   Consider any pairing $\pi$ in the sum  on the left hand side; by irreducibility,
we know that  $\sum_{(r,s) \in \pi} \str t_{r}-t_s \str  \geq \str I \str$ and that $t_1= \infi, t_{2p}=\supi$, so we can bound the left hand side by
\begin{equation}
\e{-\tfrac{2g_\res}{3} \str I \str}      \sum_{p\geq n}   \sum_{{\pi\in\mathrm{Pair}{ (p)}}} \int_{\Delta^{2p-2}_I} \d \ut  
\prod_{(r,s) \in \pi}  \str k(t_r-t_s) \str
\Big\str^{}_{\footnotesize{\left.\begin{array}{l}   t_1= \infi \\ t_{2p}=\supi
\end{array}\right.}} \,,      \label{eq: sup sum two}  \end{equation}
where we set
\begin{equation}\label{def: k}
k(t)\deq  \la^2 (C\tilde\nu^{-4d})
\e{ t(C  \tilde \nu  - \tfrac{g_\res}{3}) } \tilde h(t)\,.
\end{equation}

To deal with \eqref{eq: sup sum two}, we first develop a 
combinatorial estimate:
\begin{lemma}   Let $\e{x}_n\deq \sum_{j \geq n} x^j/j!$ and 
$\e{x}_{-n}\deq\e{x}_{0} =\e{x}$, for $n \in \bbN$.
Then, for  $ k \in \mathrm{L}^1(\bbR_+)\cap \mathrm{L}^{\infty}(\bbR_+)$,  
\begin{align}
&  \sum_{p\geq n} \sum_{\pi\in\mathrm{Pair}{ (p)}} \int_{\Delta^{2p}_I} \d \ut  
\prod_{(r,s) \in \pi}  \str k(t_r-t_s) \str  \leq \e{\str I  \str \norm k
\norm_1}_n  \label{eq: combi1}\,,
 \\[2mm] 
&  \sum_{p\geq n} \sum_{{\pi\in\mathrm{Pair}{ (p)}}} \int_{\Delta^{2p-2}_I} \d \ut  
\prod_{(r,s) \in \pi}  \str k(t_r-t_s) \str
\Big\str^{}_{\footnotesize{\left.\begin{array}{l}   t_1= \infi \\ t_{2n}=\supi
\end{array}\right.}} \leq  \norm k \norm^{}_{\infty} \e{\str I  \str \norm k 
\norm_1}_{n-1}  + \norm k \norm^2_{1} \e{\str I  \str \norm k  \norm_1}_{n-2}\,,
\label{eq: combi2}
\end{align}
 where the $p=0$ term on  the left hand side of \eqref{eq: combi1} is understood to equal $1$.  The bounds \eqref{eq: combi1},\eqref{eq: combi2} hold for $n\geq 0$,  $n\geq 1$, respectively.
\end{lemma}
\begin{proof}
For any pairing $(r,s)\in\pi$ and the corresponding time
coordinates $t_r$, $t_s$ we set $u_i\deq t_r$ and $v_i\deq t_s$, where the
indices $i=1,\ldots,n $ are chosen such that $\infi\le
u_1<u_2<\ldots<u_n\le\supi$. Note
that, by our definition of a pairing,
$u_i<v_i$. By using the change of variables $\ut \to (\uu,\uv)$, we rewrite the
left side of~\eqref{eq: combi1} as 
\begin{equation}
Z_I(n)\deq \sum_{p\geq n}  \int_{\Delta_I^{2p}} \d \uu    \int_{v_j >u_j} \d \uv  
\prod_{j=1}^p  k(v_j-u_j)\,,
\end{equation}
with the $p=0$ term being $1$, and we set $Z_I(n <0)\deq Z_I(0)$.
 Each $v_j$-integral is bounded by $\norm k \norm_1$,  the  integral over
$\Delta_I^{n}$ gives $\str I\str^n/n!$ and~\eqref{eq: combi1} follows.
To derive~\eqref{eq: combi2},  we split the expression according to whether the pairing $\pi$ contains
the pair $(1,2n)$ or not.  
In the first case,  we note that the time-coordinates of all pairs other than $(1,2n)$ are not constrained and hence we get the estimate 
\begin{equation}\label{eq: combi3}
\str k(\supi-\infi)\str\,\str Z_I(n-1)\str \,,
\end{equation}
where the first factor originates from the pair $(1,2n)$.
In the second case, there are pairs $(1,j),  (j', 2n)$ with $j \neq 2n, j' \neq 1$. The time coordinates of all other pairs are again unconstrained, so we get the estimate
\begin{equation}\label{eq: combi4}
 \int \dd v\, \str k(v-\infi)\str  \int
\dd u \,\str k(\supi-u)\str\,  \str Z_I(n-2)\str\,,
\end{equation}
where the first and second factor originate from the pairs  $(1,j),  (j', 2n)$, respectively. 
 Equation~\eqref{eq:
combi2} then follows from adding~\eqref{eq: combi3} and~\eqref{eq: combi4}, and using \eqref{eq: combi1} to evaluate $Z_I(\cdot)$.
\end{proof}

To bound  \eqref{eq: sup sum two}, we first choose  $\tilde\nu$ small enough such $\tfrac{g_\res}{3} - C \tilde \nu >0$ and hence $k$ as defined in \eqref{def: k} belongs to $\mathrm{L}^1(\bbR_+)\cap \mathrm{L}^{\infty}(\bbR_+)$.  Then, by~\eqref{eq: combi2}, we bound \eqref{eq: sup sum two} by
\begin{equation}  \label{eq: sup sum three}
\e{-\tfrac{2g_\res}{3} \str I \str} \left(\norm k \norm^{}_{\infty} \e{\str I  \str \norm k 
\norm_1}_{n-1}  + \norm k \norm^2_{1} \e{\str I  \str \norm k  \norm_1}_{n-2} \right) \,.
\end{equation}
Finally, we use $\e{x}_n \leq  x^n\e{x}/n! $ for $x\geq 0$, and 
we choose $\str\la\str$ small enough compared to $(\tilde\nu)^{-2d}$, so that $\norm k
\norm_1 \leq  g_\res/3$ and \eqref{eq: sup sum three} is bounded by the right side of \eqref{eq: sup sum}.

\end{proof}

\subsubsection{Bounds on correlations functions with unbounded observables}
\label{sec: bounds unbounded correlations}
As argued previously, we cannot bound the correlation functions
$\caG_I(\cdot),\caV_I(\cdot)$ in norm when $\caS_j$ is unbound, for some $j$, but we can bound
their kernels. 
For future use in Section~\ref{sec: analysis of resolvent around zero},  let us
bound $\partial_\field  \caV^{(n)}_I \rho_\sys$ for an exponentially localized
$\rho_\sys$ and interval $I \subset \bbR_+$.  This quantity is given by (see~\eqref{eq: examplebound1})
\begin{equation} \label{eq: derivative as integral}
\partial_\field  \caV^{(n)}_I \rho_\sys =   \lambda^2\ii \sum_{\vs= \links,
\rechts}  (\delta_{\vs, \links}-  \delta_{\vs, \rechts})   \int_{I} \d s \,   
\caV^{(n)}_I ((X_l)^{s})\rho_\sys\,,
\end{equation}
and hence, by Lemma~\ref{lem: bounds on diagrams}, for $\lambda,\nu>0$
sufficiently small,
\begin{equation} \label{eq: bounds for derivative}
\left\str\left(\partial_\field  \caV^{(n)}_I \rho_\sys\right) (w)\right\str \leq 
\max(1,\str I \str^{n}) (C\str \la \str)^{2n+2}   \e{- \tfrac{g_\res}{3}  \str I \str } 
\sum_{w_0,w_1}   \str w_1\str \e{-\nu (\str w-w_1 \str+ \str w_1-w_0 \str+ \str
w_0 \str)} \,,
\end{equation}
where $\nu$ is chosen so small that $\str\rho_\sys(w_0)\str \leq C
\e{-\nu \str w_0 \str}$. 
 Higher derivatives lead to an obvious generalization of~\eqref{eq: bounds for
derivative}; the $k$'th derivative will produce the factor $\max(1,\str I \str^{n+k-1}) (C\la^2)^{n+k}$  on the right side (because the $k^{\textrm{th}}$ derivative
corresponds to $k$ time-integrations over $s_1, \ldots,s_k$ in the
generalization of~\eqref{eq: derivative as integral})  and $k$ factors $\str
w_{i} \str, i=1, \ldots,k$, but these can still be controlled by the exponential
decay in~$\str w_0 \str$ and $\str w_{i+1}-w_i \str$. 
An obvious consequence of \eqref{eq: bounds for derivative} is that, for exponentially localized $\rho_\sys$, the function 
\begin{equation} \label{eq: smoothness of cav field}
\field \mapsto  \caV^{(n)}_I \rho_\sys
\end{equation}
is $C^\infty$ and that all derivatives are exponentially localized operators, too.

\subsubsection{Bounds on the effective dynamics $\caZ_I$}\label{section: bounds
on effective dynamics}
Next, we show how the bounds in Lemma~\ref{lem: bounds on diagrams} help to
control the reduced evolution $\caZ_I$ and the correlations functions
$\caG_I(\caS_1^{s_1},\ldots, \caS_m^{s_m}) $. 
If we demand that $\caS_1, \ldots, \caS_m$ are quasi-diagonal, then Lemma~\ref{lem:
bounds on diagrams} immediately yields
\begin{equation} \label{eq: bound on cavcas}
\left\str \caV^{(n)}_I(\caS^{s_1}_1, \ldots ,\caS^{s_m}_m) (w',w) \right\str 
\leq  C'  (C\str\la\str)^{2n}  \max(1,\str I \str^{n-1}) \e{-\tfrac{g_\res}{3} \str I \str}
\e{-\nu \str w-w'\str}\,,
\end{equation}
for $\la, \nu>0$ sufficiently small, with $C'$ depending on the $\caS_j$'s. To get a bound on the operator norm, we
note that
\begin{equation*}
\norm \caS \norm \leq   \sup_{w} \sum_{w'} \str \caS(w',w)\str\,,
\end{equation*}
where the supremum and the sum are over $\Z^d\times\Z^d$.

 Next, we define operators $\mathcal{J}_{\theta}$, with $\theta=(\theta_{\links},
\theta_{\rechts})\in\C^d \times \C^d$, by 
\begin{align} \label{eq:def jkappa}
\mathcal{J}_{\theta}O\deq\e{-\ii(\theta_\links,X)}O\,\e{-\ii(\theta_\rechts,
X)}\,,\quad\ O\in\scrB(\mathscr{H}_{\sys})\,.
\end{align}
Note that $\mathcal{J}_{\theta}$ is unbounded  if $\theta$ has an imaginary
part. Also note that an operator $O\in\scrB_2(\scrH_\sys)$, is
exponentially localized iff $\|\caJ_{\theta}O\|_2<\infty$, for
$\theta=(\theta_\links,\theta_\rechts)$ in some complex neighborhood of $(0,0)$.

From~\eqref{eq: bound on cavcas}, we get for $0<\lambda$ and $\theta\in\C^{2d}$
sufficiently small, in particular $\str \theta \str \leq \nu$, 
\begin{equation} \label{eq: bound on cavcas norm}
\norm \caJ_{\theta}\caV_I(\caS^{s_1}_1, \ldots, \caS^{s_m}_m) \caJ_{-\theta}
\norm  \leq  C' (C\str\la\str)^{2n}  \max(1,\str I \str^{n-1}) \e{-\tfrac{g_\res}{3} \str I
\str}\,.
\end{equation} 
This implies that $\caV_I^{(n)}(\caS_1^{s_1},\ldots, \caS_m^{s_m})
$  preserves the subspace of exponentially localized density operators.

Using the above bounds on $\caV^{(n)}_I(\,\cdot\,)$, with $m=0$, and  propagation
bounds on $\caU_I$, namely \[\norm  \caJ_{\theta}\,\caU_I \caJ_{-\theta} \norm
\leq C \e{\la^2 \str I \str \caO(\str \theta \str)}\,,\quad\textrm{ for }I \subset
\bbR_+\,,\qquad\qquad \norm \caJ_{\theta}\,\caU_I \caJ_{-\theta} \norm  \leq
C\,,\quad\textrm{ for }I \subset[-\beta/2,0]\,,\] we can bound the series
in~\eqref{eq:6.28} and~\eqref{eq: relation caz cazbeta} by 
\begin{equation*}
\norm \caJ_{\theta}  \caZ_{I} \caJ_{-\theta}  \norm \leq   C  \e{\la^2 \str I
\str \caO(\str \theta \str)}  \sum_{l=0}^{\infty}  \int_{{\Delta^{2l}_I}} \d
\underline{t}   \,  (C\str\la\str)^{2l}  \leq  C  \e{\la^2 \str I \str \caO(\str
\theta \str)}\,,\qquad I\subset[-\beta/2,\infty)\,.
\end{equation*}

\subsection{Laplace transform of Green functions}
As already mentioned in the introduction, it is more convenient to conduct our analysis of the long-time
behaviour of $\caZ_{I}$ in the energy-domain, instead of the time-domain. For
sufficiently large $\re z$, we set
\begin{align}\label{11eq}
\mathcal{R}(z)\deq\int_0^{\infty}\dd t\,\e{-zt}\,\Matz\,, \qquad   
\mathcal{R}_{\be}(z)\deq\int_0^{\infty}\dd t\,\e{-zt}\,\caZ_{[-\beta/2, t]}\,,
\end{align}
\begin{align} 
\caM(z) \deq \int_0^{\infty} \d t\, \e{-zt} 
(\caV_{[0,t]}-\caV_{[0,t]}^{(2)})\,, \qquad 
\caR_{\mathrm{ex}}(z) \deq \int_0^{\infty} \d t \,\e{-zt}  \caV_{[0,t]}^{(2)}\, 
, \label{eq:6.2900}
\end{align}
and (as elucidated below), 
\begin{align*} 
\caY(z) \deq   \caZ_{[-\be/2,0]} +   \int_{0}^{\infty} \d v \, \e{- z v} 
\int_{-\be}^{0} \d u\,  \caV_{[u,v]}    \caZ_{[-\beta/2,u]} \,.
\end{align*}
Note that $\caM$ is the sum/integral of the lowest order diagrams. When writing $\norm \caA \norm$, where $\caA$ is an operator acting on (a subspace of) 
$\scrB_2(\scrH_\sys)$, we understand $\norm \cdot \norm$ to be the standard
operator norm on $\scrB(\scrB_2(\scrH_\sys))$. 

Recall the definition of the operators $\caJ_\theta$, with $\theta \in
\bbC^{2d}$, in~\eqref{eq:def jkappa}.
\begin{lemma}\label{lem: pseudoresolvent}
The operator-valued function  $(z,\theta) \mapsto \caJ_{\theta} \caA(z)
\caJ_{-\theta}$, with $\caA= \caM, \caR_{\mathrm{ex}}, \caY$, is
analytic  in the region  $|\theta|<k_\theta, \mathrm{Re}\,z>-k_z$, for some
$k_{z},k_\theta >0$,  and satisfies the bounds (as $\la\to 0$)
\begin{equation}  \label{eq: bounds on carex}
\sup_{|\theta|<\theta_0\,,\,\mathrm{Re}\,z>-g_0}  \left\{  \begin{array}{rr} 
\|\mathcal{J}_{\theta}\mathcal{M}(z)\mathcal{J}_{-\theta}\|
&= \,\,\caO(\lambda^2) \,,   \\[2mm]
\|\mathcal{J}_{\theta}\mathcal{R}_{ex}(z)\mathcal{J}_{-\theta}\|
&=\,\,\caO(\lambda^4)  \,,  \\[2mm]
\|\mathcal{J}_{\theta}\mathcal{Y}(z)\mathcal{J}_{-\theta}\|
&=\,\,\caO(\lambda^0)\,.
  \end{array} \right.
\end{equation}
Moreover, for $\re z>0$, 
\begin{equation} \caR_\be(z)=  \caR(z) \caY(z)  \label{eq: relation car and carbe} \end{equation}
and
\begin{align}\label{eq:6.4}
\mathcal{R}(z) &=(z-\caL_\sys-\mathcal{M}(z)-\mathcal{R}_{ex}(z))^{-1}\,,
\end{align} 
where $\caL_\sys=\ad (H_\sys)$ is the Liouvillian of the particle system.
\end{lemma}
\begin{proof}
The bounds on $\mathcal{R}_{ex}(z), \mathcal{M}(z)$ are (the Laplace transform
of) the bound in \eqref{eq: bound on cavcas norm} for $m=0$ and $n=2$, $n=1$,
respectively, with $k_\theta < \nu$, and $k_z = g_\res/4$. For the bound
on $\caY(z)$, we also use that $\norm \caZ_I \norm  \leq C $, for  intervals
$I \subset  [-\beta/2, 0]$, as established above.
To get \eqref{eq:6.4}, let us abbreviate
 \begin{align*}
 \mathcal{R}\irr(z)\deq\mathcal{M}(z)+\mathcal{R}_{ex}(z)\,,\quad
\mathcal{R}_{\sys}(z)\deq(z-\caL_{\sys})^{-1}\,.
 \end{align*}
 Since $\caL_{\sys}$ is selfadjoint (as an operator on $\scrB_2(\scrH_{\mathrm{S}})$), we
have $\|\mathcal{R}_{\sys}(z)\|\le|\re z|^{-1}$. We
choose $\lambda$ sufficiently small and $\re z$ sufficiently large such that
$\|\mathcal{R}\irr(z)\mathcal{R}_{\sys}(z)\|\le|\re
z|^{-1}\|\mathcal{R}\irr(z)\|<1$.
Starting from the `polymer expansion' of $\Matz$ in Equation~\eqref{eq:6.28} and
taking the Laplace transform, we
find that
 \begin{align*}
\int_0^{\infty}\dd t\,\e{-zt}\,\Matz
 =\sum_{n=0}^{\infty}\mathcal{R}_{\sys}(z)\big(\mathcal{R}\irr(z)\mathcal{R}
_{\sys}(z))^n 
=\mathcal{R}_{\sys}(z)(1-\mathcal{R}\irr(z)\mathcal{R}_{\sys}(z))^{-1}\,,
 \end{align*}
which is \eqref{eq:6.4}, for $\re z$ large enough. Since $\|\caZ_I\|\le C$, the
left side of~\eqref{eq:6.4} is an analytic function in the region
$\{z\in\C\,:\,\re z>0\}$ and we can extend~\eqref{eq:6.4} to that region by
analytic continuation.  Finally, the relation $\caR_\be(z)= \caR(z) \caY(z)$
follows (first for $\re z$ large enough) by taking the thermodynamic limit and
the Laplace transform of \eqref{eq: relation caz cazbeta} (and then by
continuing analytically in $z$).

\end{proof}

\subsection{Fiber decomposition}\label{sec: fiber
decomposition}
We interrupt our analysis of Green functions in order to recall the fiber decomposition. 

To start with, we note that $\scrB_{1}(\scrH_\sys)\subset\scrB_2(\scrH_\sys)$,
$\scrH_{\sys}=\ell^2(\Z^d)$. Hence, we may view density matrices on
$\scrH_\sys$ as elements of the space of Hilbert-Schmidt operators,
$\scrB_2(\scrH_\sys)\simeq
\mathrm{L}^2(\bbT^d \times \bbT^d, \d k_{\links}\d k_{\rechts})$. 
We define 
\begin{equation*}
    \widehat O     (k_{\links},k_{\rechts}) \deq\frac{1}{(2\pi)^{d}}
\sum_{x_{\links},x_{\rechts} \in \lat}   
O(x_{\links},x_{\rechts})  \e{- \ii k_{\links} \cdot x_{\links}+\ii k_{\rechts}
\cdot x_{\rechts} }\,, \qquad  O \in    
\scrB_2(\ell^2(\lat))\,.
\end{equation*}
In what follows, we will write $O$ for $\widehat O$.
To conveniently cope with the translation invariance of our model, we make the
following
change of variables 
\begin{equation*}
k\deq \frac{k_{\links}+k_{\rechts}}{2}\,, \qquad  p\deq k_{\links}-k_{\rechts}\,,
\end{equation*}
and, for a.a.\ $p \in \tor$, we obtain a well-defined function $ O_p
\in\mathrm{L}^2(\bbT^d)$ by putting
\begin{equation}\label{eq:2.68}
(O_p)(k) \deq O \left(k+\frac{p}{2},k-\frac{p}{2}\right)\,.
\end{equation}
This follows from the fact that the Hilbert space  $\scrB_2(\scrH_\sys)
\simeq\mathrm{L}^2(\bbT^d \times \bbT^d, \d k_{\links}\d k_{\rechts})$ can be
represented as a
direct integral
\begin{equation} \label{def: fiber decomposition}
\scrB_2(\scrH_\sys) \simeq \int^\oplus_{ \bbT^d} \d p \,    \scrH_p \,, \qquad  
  O =
 \int^\oplus_{ \bbT^d} \d p \, O_p\,,
\end{equation}
where each `fiber space' $\scrH_p$ can be identified with
$\mathrm{L}^2(\bbT^d)$.

Recall the definition of the operators $\caJ_{\theta}$, with $\theta\in\C^{2d}$, in~\eqref{eq:def jkappa}. The following lemma captures some identities used later on.
\begin{lemma} \label{lemma: fibers}
Let $O \in \scrB_1(\scrH_\sys)$, then 
\begin{equation} \label{eq: trace as integral}
\Tr_{\sys} [O \e{\ii p\cdot X }]   =  \langle 1,  O_p   
\rangle_{\mathrm{L}^2(\tor)}=  \int_\tor \d k\, 
O_p(k)\,, \qquad{  p\in \tor}\,. 
\end{equation}
If there is a $\delta>0$ such that  $
 \norm \caJ_{\theta/2} O \norm_2  < \infty $, for  $ \str \theta \str  \leq
\delta$, then $p \mapsto  O_p $ is analytic in the interior of the strip
$\bbV_{\delta}  $.
\end{lemma}
(In the discussion above, the fiber operator $O_p$ is defined for a.a.\ $p$,  for $O \in \scrB_2(\scrH_{\mathrm{S}})$, but in the context of Lemma~\ref{lemma: fibers}, $O_p$ can be defined for arbitrary $p$.
The first statement of the lemma follows from the singular-value decomposition for trace-class operators and standard properties of the Fourier transform. The second
statement of Lemma~\ref{lemma: fibers} is the Paley-Wiener theorem, i.e., the relation between exponential
decay of functions and analyticity of their Fourier transforms; see~\cite{reedsimon2}.)

The fiber decomposition in Equation~\eqref{def: fiber decomposition} is useful
when one deals with operators $\caA$ acting on $\scrB_2(\scrH_\sys)$ that are
translation-invariant (TI), i.e., 
$\caT_z \caA \caT_{-z}= \caA$, with $\caT_z$ defined  as in Section~\ref{sec:
thermo}.   An important example of a TI operator $\caA$ is the
reduced time-evolution $\caZ_{[0,t]}$, see Lemma \ref{lem: thermodynamic dyn}. For TI operators $\caA$, we find that $(\caA O)_p$ depends on $O_p$ only, and
hence it makes sense to write
\begin{equation}   \label{def: fiber decomposition caa}
(\caA O)_p= \caA_p O_p\,, \qquad  \caA = \int_{\tor}^\oplus \dd p\,\caA_p\,.
\end{equation}

Similarly to Lemma~\ref{lemma: fibers} above, we find that, if $\caJ_{\theta/2}
\caA \caJ_{-\theta/2}$ is bounded for all
$\theta=(\theta_\links,\theta_\rechts)$, with $|\theta|\le \delta$, then the map
 $p \mapsto \caA_p$ is analytic in a strip $\bbV_\delta$.

\subsection{Identifying the lowest order contributions: $\caL_S+\mathcal{M}(z)$}\label{sec: calculation ladders}
We return to our the analysis of Green functions. 
Identifying the fiber spaces $\scrH_p$ with
$\mathrm{L}^2(\tor)$, we interpret~$(\caL_S)_p+\caM(z)_p$ as an operator acting on~$\mathrm{L}^2(\tor)$. 

First, we observe that
\begin{align}\label{lowest order free liouvillian}
 \lambda^{-2}(\caL_S)_{\lambda^2\kappa}=\ii\kappa\cdot\nabla\epsilon-\field\cdot\nabla+\caO(\lambda^2\kappa)\,,
\end{align}
in the limit $\ka\to0$, $\lambda\to 0$. 

Second, displaying the $\field$-dependence in $\caM(z)$ explicitly, a straightforward calculation yields, (see Section~6.2 in~\cite{paper1}),
\begin{align}\label{trick with ladders}
 \lambda^{-2}(\caM(z=0,\field=0)_{\lambda^2\kappa}f)(k)=(Gf)(k)+(Lf)(k)\,,\quad\quad f\in\mathrm{L}^2(\tor)\,,
\end{align}
where
\begin{align}
 (Gf)(k)\deq\int_{\tor}\dd k'\, r(k',k)f(k')\,,\quad\quad (Lf)(k)\deq-\int_{\tor}\dd k'\,r(k,k')f(k)\,,
\end{align}
with $r(\cdot,\cdot):\tor\times\tor\to \C$, the `rate function'
\begin{align}
 r(k,k')=\psi[\epsilon(k')-\epsilon(k)]\,,
\end{align}
where $\psi$ is the `spectral reservoir density' defined in~\eqref{definition of psi ohne hut}. Starting from Assumptions A and B, it is straightforward to verify that $r(\,\cdot\,,\,\cdot\,)$ is a real-analytic function in both variables, which is strictly positive almost everywhere for real arguments.

Hence, taking into account the contributions in~\eqref{lowest order free liouvillian}, we are led to consider the operator
\begin{align}
 M^{\kappa,\field}\deq\ii\kappa\cdot\nabla\epsilon-\field\cdot\nabla+G+L\,,
\end{align}
which is densely defined on $\mathrm{L}^2(\tor)$, with core~$C^{\infty}(\tor)$.

 The operator $M^{\kappa,\field}$ has the physical interpretation of a generator of a one-parameter (strongly continuous) semigroup on $\mathrm{L}^{2}(\tor)$, often referred to as linear Boltzmann evolution. A detailed analysis of the spectrum of~$M^{\kappa,\field}$ and of the associated evolution equation has been carried out in~\cite{paper1}. The next lemma captures the main results of Section 6 of~\cite{paper1}.
\begin{lemma}\label{kinetic lemma}
 There exist constants $k_{\kappa},k_{\field}>0$, such that for $|\kappa|\le k_{\kappa}$ and $|\field|\le k_{\field}$ the following holds.
\begin{itemize}
 \item[$i.$] 
The spectrum of the operator $M^{\kappa=0,\field}$, $|\field|\le k_{\field}$, satisfies
\begin{align}
 \sigma(M^{\kappa=0,\field})\subset\{0\}\cup\{ z\in\C\,:\, \re z\le -g_M(\field)\}\,,
\end{align}
for some $g_M(\field)>0$. Moreover, $0$ is a simple eigenvalue. The spectral projection associated with the eigenvalue $0$ is of the form $\ket{\zeta_M^{\field}}\bra 1$, where $\zeta_M^{\field}\in\mathrm{L}^2(\tor)$ is a strictly positive, smooth function on $\tor$, normalized such that $\langle 1,\zeta_M^{\field}\rangle_{\mathrm{L^2(\tor)}}=1$.
\item[$ii.$] The spectrum of the operator $M^{\kappa,\field}$, $|\kappa|\le k_{\kappa}$, $|\field|\le k_\field$, satisfies
\begin{align}
 \sigma(M^{\kappa,\field})\subset\{u_M(\kappa,\field)\}\cup\{ z\in\C\,:\, \re z\le -g_M(\kappa,\field)\}\,,
\end{align}
where $u_M(\kappa,\field)=\caO(\kappa)$ and $g_M(\kappa,\field)=g_M(\field)+\caO(\kappa)>0$. Moreover, $u_M(\kappa,\field)$ is a simple (isolated) eigenvalue, whose associated spectral projection, $P_M^{\ka,\field}$, can be written as $\ket{\zeta_M^{\kappa,\field}}\bra{\widetilde\zeta_M^{\kappa,\field}}$, with ${\zeta_M^{\kappa,\field}},{\widetilde\zeta_M^{\kappa,\field}}\in\mathrm{L}^2(\tor)$ two smooth functions on $\tor$, normalized such that $\langle \widetilde\zeta_M^{\kappa,\field},\zeta_M^{\kappa,\field}\rangle=1$.
\end{itemize}
\end{lemma}
In the following we will refer to $\zeta_M^{\kappa,\field}$ as the invariant state of $M^{\kappa,\field}$.

We refer to Section 6 of~\cite{paper1}, for a detailed proof of this lemma. One key ingredient is that the spectrum of the multiplication operator $L$ has a gap, as follows from the strict positivity of the rate function $r(\,\cdot,\,\cdot\,)$. Since $\field\cdot\nabla$ is anti-self adjoint, the operator $-\field\cdot\nabla+L$ has the same gap. Next, since the rate function is analytic in both variables, $G$ is a compact operator. Then Weyl's theorem on the stability of the essential spectrum and a Perron-Frobenius type argument, using again the positivity of the rate function, yield part~$i$. Part~$ii$ follows from analytic perturbation theory since $\ii\kappa\cdot(\nabla\epsilon)$ is a bounded operator.\\

Next, recalling the definition of $\mathcal{M}$, we obtain from~\eqref{lowest order free liouvillian} and~\eqref{trick with ladders}, for $\field=0$,
\begin{align}  
(\caL_\sys+\mathcal{M}(z=0,\field=0) )_{\lambda^2\kappa} =\lambda^2  M^{\kappa,\field=0} + \caO(\lambda^4\str \ka \str) \,,
\end{align}
as $\lambda\to 0$, $\kappa\to 0$. For $\field\not=0$, the situations is more subtle: In~\eqref{trick with ladders}, we have set $\field=0$. However, as it turns out, $\caM(z,\field)_p$ is {\it not} an analytic perturbation of $\caM(z,\field=0)_p$. To overcome this technical difficulty, we define an operator $\widetilde{M}^{\lambda,\kappa,\field}$ on $\mathrm{L}^2(\tor)$ by
\begin{align}\label{definition of M tilde}
\widetilde M^{\lambda,\kappa,\field}\deq \ii\kappa\cdot(\nabla\epsilon)-\field\cdot\nabla+\lambda^{-2}(\caM(z=0,\field))_{\lambda^2\kappa}\,,
\end{align}
with core $\caD=C^{\infty}(\tor)$, such that
\begin{align}  \label{eq: deltam rules}
(\caL_\sys+\mathcal{M}(z=0,\field) )_{\lambda^2\kappa} =
\lambda^2\widetilde M^{\lambda,\kappa,\field} + \caO(\lambda^4\str \ka \str) \,,
\end{align}
as $\lambda\to 0$, $\kappa\to 0$, holds.

In Section 6.2 of~\cite{paper1}, we have proven the following lemma, relating the spectrum of $\widetilde M\equiv\widetilde M^{\lambda,\kappa,\field}$ to the spectrum of $M\equiv M^{\kappa,\field}$ in a small neighborhood of zero. Recall the definition of the gap $g\kin(\field)$ in Lemma~\ref{kinetic lemma}, and define $B_r$ to be the disk $B_r \deq\left\{ z \in \bbC\,:\,\str z \str \leq r
\right\}$. 

\begin{lemma}\label{lem: spectrum of tilde m}
 There is a constant $r>0$, $r \propto g\kin(0)$, such that, inside the ball $B_r$,
the operators $M$ and $\widetilde M$  have unique simple
eigenvalues $\eig\kin\equiv\eig\kin(\ka,\field)$ and $\eig\kini\equiv
\eig\kini(\la,\ka,\field) $, respectively, with $\str
\eig\kini-\eig_{M^{\phantom{'}}}\str =\caO(\la^2)$.  Moreover, for $z \in B_r$, 
\begin{equation*}
\frac{1}{z-\widetilde M } = \frac{1}{z-u\kini} P\kini + \caO(z^0)\,.
\end{equation*}
\end{lemma}
For a proof we refer to Lemma~6.3 in~\cite{paper1}.

 This concludes our discussion on the lowest order contributions. In the following section, we explain how the higher order contribution can be controlled.

\section{Analysis of $\caR(z)$ around $z=0$} \label{sec: analysis of resolvent
around zero}\label{sec: perturbation
kinetic limit}

In this section, we show that the map $z\mapsto\caR(z)$, a priori defined for $\re
z>0$, can be analytically extended into the region $\{z\in\C\,:\, |z|< \lambda^2 r\}$, for some
$r>0$ and $\la>0$ sufficiently small. This is accomplished by applying
perturbation theory to the (fibers of the) operators  $\caR(z)$. The guiding
idea is that  $(\caR(z))_{\la^2 \ka}$ is a small perturbation of $(z-\lambda^2
\widetilde{M}^{\lambda,\ka,\field})^{-1}$. The small parameters will be the coupling constant $\la$, the (rescaled) fiber
momentum $\ka$ and the field $\field$.
 
All of these three parameters are assumed
to be {\it sufficiently small} throughout, and we do not repeat this at every step.

In Lemma~7.4 of~\cite{paper1} we have shown that the map $z\mapsto
(\caR(z))_{\la^2\ka}$ has a unique simple pole in a neighborhood of $z=0$, whose
residue, $P\equiv P^{\lambda,\ka,\field}$ is a rank-one operator with the property that, in the fiber
indexed by $\kappa=0$,
\begin{align}
 P^{\la,\ka=0,\field}=\ket\zeta\bra1\,,\qquad\textrm{ with
}\qquad\|\zeta-\zeta_M\|_{\mathrm{L}^2(\tor)}=\caO(\lambda^2)\,,
\end{align}
where $\zeta_M\equiv\zeta_M^{\kappa,\field}$ is the invariant state of the generator $M\equiv M^{\kappa,\field} $; see Lemma~\ref{kinetic lemma}. In Lemma~\ref{lem: prop for asymptotic}, we establish that $P$ and $\zeta$ are regular function of $\field$.

To start with, we define an operator acting on $\mathrm{L}^2(\tor)$:
\begin{align}\label{eq: definition of S}
 S \equiv &   S(z,\field,\la, \ka) \deq  (\caL_\sys+ \caM(z) +
\caR_{\mathrm{ex}}(z))_{\la^2 \ka}\,,
\end{align}
such that $(\mathcal{R}(z))_{\la^2 \ka}= (z-S)^{-1}$ whenever the left side
is well-defined. Note that $S$ is a closed operator on~$\mathrm{L}^2(\tor)$:
It is bounded except for the term $ \field\cdot \nabla$  that comes from
$\caL_\sys$. 

For simplicity, we often abbreviate $S(z,\la,\ka,\field)$ by, e.g., $S(z)$, when
considering the function \mbox{$z\mapsto S(z)$}, with the other variables kept
fixed. We use similar shorthand notations for $u_M\equiv
u_M(\kappa,\field)$, $P\equiv P^{\lambda,\kappa,\field}$, etc.\ in this
and the remaining sections.

Recalling~\eqref{definition of M tilde}, we write
\begin{align}
 S(z)=\lambda^2 \widetilde{M}+(\caM(z,\field)-\caM(0,\field))_{\lambda^2\kappa}+(\caR(z))_{\lambda^2\kappa}\,.
\end{align}
From the definition of $\caM(z,\field)$ in~\eqref{eq:6.2900}, we infer, using~\eqref{eq: bound on cavcas norm}, that
\begin{align}
 \|(\caM(z,\field)-\caM(0,\field))_{\lambda^2\kappa}\|\le C\lambda^2 |z|\,,
\end{align}
as an operator on $\mathrm{L}^2(\tor)$. Moreover, from Lemma~\ref{lem: pseudoresolvent}, $\|(R_{\mathrm{ex}})_{\lambda^2\kappa}\|\le C\lambda^4$. Thus for $z\in B_{\lambda^2 r}$, $r>0$, sufficiently small, we have
\begin{align}
 S(z)=\lambda^2 \widetilde{M}+\caO(\lambda^4(1+|\kappa|)\,.
\end{align}

Recall the constants
$k_z$ and $k_\theta$ from Lemma~\ref{lem: pseudoresolvent}. 
\begin{lemma}\label{lem: prop for asymptotic}
Let $\caD \subset \mathrm{L}^2(\bbT^d)$ be the dense subspace of real-analytic
functions on $\bbT^d$.
\begin{itemize}
 \item[$i.$] $\caD$ is a core for $S$ and $S \caD \subset \caD$. For all $ z 
\in \bbC$ satisfying $\re z \geq -k_z $ and such that $ (z-S(z))^{-1}$ exists (i.e., as a bounded operator), we have $(z-S(z))^{-1} \caD \subset
\caD$. Moreover, the differences $S(z)-S(z=0)$ and $S(\ka)-S(\ka=0)$ are bounded
operators, and they are analytic in the variables $\ka,z$
 in the region
\mbox{$\mathrm{Re}z>-k_z$} and $|\kappa|<k_\theta$.
\item[$ii.$]  For $ f \in \caD$, the function $\field \mapsto
S(\field)f $ is $C^{\infty}$, and  all derivatives  are bounded in the variables $z,\la,\ka,\field$,  uniformly on
compacts. Moreover, all derivatives of $\field\mapsto S(\field)f$, $f\in\caD$, are in $\caD$.
\item[$iii.$] Fix $r>0$ sufficiently small, e.g., $r= g\kin(0)/4$. Then there is a unique
$z=z^{*}(\la,\ka,\field)$ in $B_{\la^2 r}$ such that $z-S(z)$ is not invertible,
i.e., such that $z \in \sigma(S(z))$. Denoting this unique $z^*$ by $u^{\lambda,\kappa,\field}$, we have that $u^{\lambda,\kappa,\field}$ is an isolated simple
eigenvalue of $S(u^{\lambda,\kappa,\field})$, and when considered as pole of the map $z\to S(z)$, the residue at $u^{\lambda,\kappa,\field}$, $P^{\lambda,\kappa,\field}$, is a rank-one operator.
\item[$iv.$] For $r>0$ sufficiently small,
\begin{align}
 \frac{1}{z-S(z)}=\frac{1}{z-u^{\lambda,\kappa,\field}}P^{\lambda,\kappa,\field}+R^{\lambda,\kappa,\field}(z)\,,
\end{align}
for $z\in B_{\lambda^2r}$, where $z\mapsto R^{\lambda,\kappa,\field}(z)$ is bounded analytic in $B_{\lambda^2r}$.
\item[$v.$] The pole $u\equiv u^{\lambda,\ka,\field}$ and the operators $P^{\ka}\equiv P^{\lambda,\kappa,\field},R^{\ka}(z)\equiv R^{\lambda,\kappa,\field}(z)$ are analytic in $\ka$ and
$\ga$. The pole $u$ is a $C^{\infty}$-function in $\field$, and, for $B=P^\kappa,R^\kappa(z)$, the function $\field
\mapsto B(\field)f$ is $C^{\infty}$, for any $f \in \caD$.
\end{itemize}
\end{lemma}
\begin{proof}
Statements $i$, $iii$, $iv$ and the statements about analyticity in $\kappa$ in item $iv$ have been proven in Section~7 of~\cite{paper1}. Here, we only give the proofs of statement $ii$ and of the claim on regularity in $\field$ in~$v$.

Item $ii$ is a consequence of the discussion in
Subsection~\ref{sec: bounds unbounded correlations}. In particular, the first
derivative is constructed by restricting \eqref{eq: bounds for derivative} to a
fiber (take $n=2$ to get $\caR_{\mathrm{ex}}$ and $n=1$ to get $\caM$,  for
$\caL_\sys$, the differentiability is explicit).  Higher derivatives are dealt
with analogously; see the remark after \eqref{eq: bounds for derivative}. 

To prove the regularity claim in $iv$, we first note that residue and pole of the map $z\mapsto (z-S(z))^{-1}$, can be expressed as contour integrals, 
\begin{equation} \label{eq: contour rep}
P =  \frac{1}{2 \pi \ii}    \int_\caC \d z \,\frac{1}{z-S(z)}, \qquad  \eig   P=
\frac{1}{2 \pi \ii}     \int_\caC \d z  \,\frac{z}{z-S(z)}\,,
\end{equation}
with the (positively oriented) contour $\caC\deq\{z\in\C\,:\,|z|= \lambda^2r/2\}$. 
 We claim that
\begin{equation}\label{asymptotic 1}
 \left(\frac{1}{z-S(z,\field')}-  \frac{1}{z-S(z,\field)}\right)f =   
\frac{1}{z-S(z,\field')} (S(z,\field')- S(z,\field))  
\frac{1}{z-S(z,\field)}f\,,\quad\quad f\in\caD\,.
\end{equation}
By statement $i$, both sides are well defined for $z\in \caC$ and the equality is checked by multiplying both sides with the invertible operator $z-S(z,\field')$. Since $(z-S(z,\field'))^{-1}$ is uniformly bounded for all $z\in\caC$, and since $(z-S(z,\field))^{-1}f$ is in $\caD$, it follows from $i$ that the right side of~\eqref{asymptotic 1} converges to~$0$, as $\field'\to\field$, for $f\in\caD$ (and by density this holds for all $f\in\mathrm{L}^2(\tor)$). Thus, forming the difference quotient, we obtain
\begin{align}\label{asymptotic 2}
\frac{\partial}{\partial\field}\frac{1}{z-S(z,\field)}f=\frac{1}{z-S(z,\field)}\left(\frac{\partial}{\partial\field}S(z,\field)\right)\frac{1}{z-S(z,\field)}f\,,\quad\quad f\in\caD\,,
\end{align}
 where we used that $S(z,\field)f$, $f\in \caD$, is differentiable in $\field$, with derivative uniformly bounded for $z\in\caC$; see item~$ii$. We also claim that the right hand side defines a function in $\caD$. This follows from $(z-S(z,\field))^{-1}\caD\subset \caD$, see~$i$, and $\frac{\partial}{\partial\field}S\caD\subset \caD$, see~$ii$. Thus the right side of~\eqref{asymptotic 2} is a function in $\caD$, which is uniformly bounded on compacts in the variables $\lambda,\kappa,\field$ and uniformly bounded for all $z\in\caC$. Thus the above procedure can be iterated, and we infer that $(z-S(z,\field))^{-1}f$, $f\in\caD$, is a $C^{\infty}$-function in the variable $\field$, whose derivatives are uniformly bounded on compacts in the variables $\lambda,\kappa,\field$ and for all $z\in\caC$. 

The identities in~\eqref{eq: contour rep} then immediately lead to the conclusion that $\eig$, $Pf$, $R(z)f$, for $f\in\caD$, are $C^{\infty}$-functions of $\field$. We refer to~\cite{katoperturbation} for a more detailed treatment of asymptotic perturbation theory.

\end{proof}
 Replacing $S(z)$ by the Boltzmann generator $M^{\kappa,\field}$, the same proof also shows that $u_M^{\kappa,\field}$, $\zeta_M^{\kappa,\field}$, in Lemma~\ref{kinetic lemma}, are $C^{\infty}$-functions of $\field$.
\newline

Next, we recall Lemma~\ref{lem: spectrum of tilde m}, to compare the eigenvalue $u^{\lambda,\kappa,\field}$, the rank-one operator $P^{\lambda,\kappa,\field}$, with the corresponding quantities of the operator $M^{\kappa,\field}$, from Section~\eqref{sec: calculation ladders}. Combination of  Lemma~\ref{lem: spectrum of tilde m} with Lemma~\ref{lem: prop for asymptotic} yields:

\begin{lemma}\label{cor:laurent} For $\kappa=0$, the residue at $z=0$, $P^{\kappa=0}$, can be written as $P^{\kappa=0} = \str \zeta^{\lambda,\kappa=0,\field} \rangle \langle 1
\str$, with $\zeta\equiv \zeta^{\la,\ka, \field}$ a
real-analytic function on $\tor$ satisfying
$$  \norm \zeta -\zeta_M \norm =\caO(\lambda^2)\,,$$
where $\zeta_M\equiv \zeta_M^{\ka,\field}$ is the invariant state of $M\equiv M^{\kappa,\field}$. 
For $\ka=0$, $\zeta$ is a probability density on $\tor$.
The function $\eig=\eig(\la,\ka,\field) \in \bbC $ satisfies $$\overline{
\eig(\ka)}=\eig(-\bar\ka)\,, \qquad |\eig-\la^2 \eig\kin |=\caO(\lambda^2)\,. $$
Moreover, we have that
\begin{align}\label{eq: eigenvalue is zero}
 \eig(\kappa=0)=0\,,\qquad P^{\ka=0} R^{\ka=0}(z)=0\,.
\end{align}

\end{lemma}

\section{The equilibrium regime $\field=0$}\label{sec: equilibrium case} 
In this section, we discuss properties of the equilibrium correlation functions. We will often need restrictions of operators acting on $\scrB_2(\scrH_\sys)$ to the
fiber space indexed by $\kappa=0$. We indicate these restrictions by writing
$(\caA)_0$, for $\caA\in\scrB(\scrB_2(\scrH_{\sys}))$, below; in particular,
$(\caA)_0$ acts on $\mathrm{L}^2(\tor)\simeq\scrH_{\kappa=0}$. Whenever we use
such fiber restrictions of operators pointwisely, i.e., for a given fiber
indexed by $\ka=0$, this can be justified by Lemma~\ref{lemma: fibers} because
all operators are quasi-diagonal on $\scrB_2(\scrH_\sys)$, and we will omit these
justifications.

Recall the results of Theorem~\ref{thm: stationary}: Statement $i$, for $\field\not=0$, is proven in~\cite{paper1}; statement $ii$, for $\field=0$, has been proven in~\cite{deroeckfrohlichpizzo}. Note that the statement for $\field=0$ is stronger. In the notation of Section~\ref{sec: calculation ladders}, this follows from the observation that $\widetilde{M}=M$, for $\field=0$. As argued in~\cite{paper1,deroeckfrohlichpizzo}, the function $z \mapsto (\caR(z))_{\la^2 \ka}$ consequently has only one
pole, namely $\eig(\la,\ka,\field=0)$, in the region  $ \re z > -\la^2g\kin(\ka,
\field=0) +\caO(\la^4)$.  Then, the pole $\eig(\la,\ka,\field=0)$ determines the long-time
properties.  By the inverse Laplace transform, one then proves the following theorem.
\begin{theorem}\emph{[Equilibrium asymptotics]} \label{thm: equilibrium rte}
Take $\field =0$. 
Then, for $0<\la$ and $\ka$ sufficiently small, there is  $g >0$ such that
\begin{align*}
\left\|(\mathcal{Z}_{[0,t]})_{\lambda^2\kappa}-\e{ \eig(\ka) t}P^\ka\right\|=\caO\big(\e{-
g\la^2  t} \big)\,, \qquad t \to \infty\,,
\end{align*}
as operators on $\mathrm{L}^2(\tor)$.
\end{theorem}
Note that $g$ can be chosen to be given by $g= g\kin(0)/5$. Also recall that
$u(\kappa=0)=0$, by Lemma~\ref{cor:laurent}. For the proof, we refer to
Theorem~4.5 of~\cite{deroeckfrohlichpizzo}.

\subsection{Correlation functions}
In this subsection, we discuss properties of the {\it equilibrium correlation
functions} and we prove Theorem~\ref{thm:correlationfunctions}. Below we choose
$\caS_i=(O)_{\vs}$, $\vs\in\{\links,\rechts\}$, with
$O\in{\mathop{\mathfrak{A}}\limits^{\circ}}_{\mathrm{ti}}$, i.e., $O$ is quasi-diagonal and translation-invariant. The extension of our results to
$\mathfrak{A}_{\mathrm{ti}}$, the norm closure of
${\mathop{\mathfrak{A}}\limits^{\circ}}_{\mathrm{ti}}$, follows by density. Recalling our discussion on Green functions in Section~\ref{negative1}, we may represent correlation functions as
\begin{align}
 \langle O_2(s_2)O_1(s_1)\rangle_{\rho_{\sys}\otimes\rho_{\referres}}&=\Tr\left[\caS_2 \caG_{[0, s_2]}(\caS^{s_1}_1) \rho_{\sys} \right]\,,\nonumber\\
 \langle O_2(s_2)O_1(s_1)\rangle_{\rho_\beta}&=\Tr\left[\caS_2 \caG_{[-\beta/2, s_2]}(\caS^{s_1}_1) \eta_\beta \right]\,,\label{representation of correlations}
\end{align}
for $s_2> s_1\ge 0$, $\caS_i=(O_i)_\links$ and $\rho_{\sys}$ an exponentially localized density matrix. We note, however, that the representation of correlation functions in~\eqref{representation of correlations} is not unique, since the value of
\begin{align}
\Tr_{\sys}[ \caG_{[0,t]}(\mathcal{S}^{s_1}_1,
\caS_2^{s_2})\rho_{\sys}]\qquad\textrm{ and }\qquad
\Tr_{\sys}[\mathcal{G}_{[-\beta/2,t]}          
(\mathcal{S}^{s_1}_1, \caS_2^{s_2})  \label{eq: to be considered}
\eta_{\beta}]\,,
\end{align}
with $t\ge s_2>s_1\ge 0$, is independent of $t$, as long as $t \geq s_2$, and we may equally well write
\begin{align}
 \langle O_2(s_2)O_1(s_1)\rangle_{\rho_{\sys}\otimes\rho_{\referres}}&=\Tr\left[\caG_{[0, t]}(\caS^{s_1}_1,\caS^{s_2}_2) \rho_{\sys} \right]\,,\nonumber\\
 \langle O_2(s_2)O_1(s_1)\rangle_{\rho_\beta}&=\Tr\left[ \caG_{[-\beta/2, t]}(\caS^{s_1}_1,\caS^{s_2}_2)) \eta_\beta \right]\,.
\end{align}
This freedom is not apparent in our expansions for
$\caG_{[0,t]}(\mathcal{S}^{s_1}_1, \caS_2^{s_2})$ (or
$\caG_{[-\be/2,t]}(\mathcal{S}^{s_1}_1, \caS_2^{s_2})$) because these
expansions contain diagrams `crossing' the time $s_2$. This suggests that 
there is a `sum rule' in our expansions, and this is the topic of the next
Lemma.
\begin{lemma}\label{lem: sum rule}
Let $P^0=\ket{{\zeta}^{0}}\bra 1$ be the spectral projection as given in
Lemma~\ref{thm: equilibrium rte} (for $\kappa=0)$. Then, for all $s\ge s_1 \geq
0$,
\begin{align}\label{eq: ward identity}
\int_{s_1}^{s}\dd t_2\int_{u}^{s_1}\dd
t_1\,P^0\,\left(\mathcal{V}_{[t_1,t_2]}(\caS_1)\,\mathcal{Z}_{[u,t_1]}
\right)_0=0\,,
\end{align}
where $u=0$ or $u=-\be/2$. 
\end{lemma}
\begin{proof}
In finite volume, cyclicity of the trace implies
\begin{align}\label{eq:7.67}
\Tr[\caS_1\e{-\ii s_1\caL}
\rho_\sys\otimes\rho_{\referres}]=\Tr\left[\e{\ii(s-s_1)\caL}\left(\caS_1\e{-\ii
s_1\caL}\rho_\sys\otimes\rho_{\referres}\right)\right]\,.
\end{align}
 Upon inserting the expansions as above and passing to the thermodynamic limit,
we find that
\begin{align}\label{eq:7.68}
P^0\left(\caS_1\mathcal{Z}_{[0,s_1]}\right)_0=P^0\left(\mathcal{Z}_{[s_1,s]}
\caS_1\mathcal{Z}_{[0,s_1]}\right)_0+
P^0\int_{s_1}^{s}\dd t_2\int_{0}^{s_1}\dd
t_1\,\left(\mathcal{Z}_{[t_2,s]}\,\mathcal{V}_{[t_1,t_2]}(\caS_1)\mathcal{Z}_{[0
,t_1]}\right)_0\,.
\end{align}
The operator $P^0$ in front of all terms in~\eqref{eq:7.68} corresponds to the
traces in Equation~\eqref{eq:7.67}. Observe that $P^0\mathcal{Z}_{I}=P^0$ for
$I\subset[0,\infty)$, since $\mathcal{Z}_I$ is trace preserving.  Using this in
the two terms of the  right side of \eqref{eq:7.68} yields \eqref{eq: ward
identity} for $u=0$. The proof for $u=-\be/2$ is similar. 
\end{proof}
We are now prepared to prove Theorem~\ref{thm:correlationfunctions}. Since the
technical input --the exponential decay in Theorem~\ref{thm: equilibrium rte}-- has been prepared, all that remains are some straightforward algebraic
manipulations. 
\begin{proof}[Proof of Theorem~\ref{thm:correlationfunctions}]
\emph{The case $m=1$:} By Theorem~\ref{thm: equilibrium rte}, we know that, for
$f$ a continuous function on $\tor$, (hence {
$M_f\in{\mathop{\mathfrak{A}}_{\mathrm{ti}}}$)}, \[
 \Tr_{\sys}[ M_f \caZ_{[0,t]}{\rho_\sys}] =  \langle f, \zeta^{0}
\rangle+\caO(\e{-g \lambda^2 t})\,,\] so the only thing in need of a proof is
that $\langle f, \zeta^{0} \rangle= \langle M_f \rangle_{\rho_\be}$. 
Recall that, by the construction of the equilibrium dynamics, we have the
stationarity property $\langle O^{\field=0}(t)\rangle_{\rho_\beta}=\langle
O\rangle_{\rho_\beta}$,
$O\in{\mathop{\mathfrak{A}}\limits^{\circ}}_{\mathrm{ti}}$, i.e.,
$(\caZ_{[-\beta/2,t]}\eta_{\beta})_0=(\caZ_{[-\beta/2,0]}\eta_{\beta})_0$,
$t\ge0$, where $\eta_\beta=\tdl\frac{Z_{\be,\res} \str \La \str }{Z_{\be}} \lone_{0}$; see Lemma~\ref{lemma: use of pinned}. Thus
\begin{align}
\langle M_f \rangle_{\rho_\be} & = z \int_0^{\infty} \d t \,\e{-t z} \langle f,
\left(\caZ_{[-\be/2,t]}\eta_\be\right)_0\rangle =    z \langle f, (\caR^{\be}(z)
\eta_{\be})_0 \rangle =  z \langle f, (\caR(z)\caY(z) \eta_{\be})_0
\rangle\nonumber \\
& =     \langle f, P^0\left(\caY(z) \eta_{\be}\right)_0 \rangle+ z  \langle f,
R^0(z)\left(\caY(z)\,\eta_{\be}\right)_0 \rangle\label{eq: mini r2e} \,.
\end{align}
The second term on the right side of the second line vanishes at $z=0$ by
the analyticity of $R(z), \caY(z)$. Choosing $f\equiv 1$, we obtain
\begin{align}\label{eq: normalization of Y}
 1=\langle \lone\rangle_{\rho_\beta}=\langle 1, P^0\left(\caY(z=0)
\eta_{\be}\right)_0 \rangle=\langle 1,\left(\caY(z=0)
\eta_{\be}\right)_0\rangle\,,
\end{align}
and hence, setting $z=0$ in~\eqref{eq: mini r2e},  $\langle M_f
\rangle_{\rho_\be} =\langle 1,\zeta^0\rangle$, in particular,
$(\caZ_{[-\beta/2,t]}\eta_{\beta})_0=(\caZ_{[-\beta/2,0]}\eta_{\beta}
)_0=\zeta^0$.

\emph{The case $m=2$:} First, we prove~\eqref{eq:2.12}. We start by considering
the correlation functions
$\Tr_\sys \left[\caS_2\, \caG_{[0,s_2]}( \caS_1^{s_1}) \rho_\sys \right]  $ and
\mbox{$\Tr_\sys \left[\caS_2\, \caG_{[-\beta/2,s_2]}( \caS_1^{s_1}) \eta_\be
\right]  $}.
Because we take the trace, and $\caS_1$, $\caS_2$ are translation-invariant
(hence fiber preserving), it is sufficient to consider 
 $(\mathcal{G}_{[\,\cdot\,,s_2]}(1))_0$ with $(1)$ standing for~$(\caS_1^{s_1})$. We obtain that
\begin{align}
\mathcal{G}_{[0,s_2]}(1)&=\mathcal{Z}_{[s_1,s_2]}\caS_1\mathcal{Z}_{[0,s_1]}
+\int_{s_1}^{s_2}\dd t_2\,\int_{0}^{s_1}\dd
t_1\,\mathcal{Z}_{[t_2,s_2]}\mathcal{V}_{[t_1,t_2]}(1)\mathcal{Z}_{[0,t_1]}
\label{eq:07.62}
\end{align}
and
\begin{align}
\mathcal{G}_{[-\be/2,s_2]}(1)&=\mathcal{Z}_{[s_1,s_2]}\caS_1\mathcal{Z}_{[
-\beta/2,s_1]}
+\int_{s_1}^{s_2}\dd t_2\,\int_{-\beta/2}^{s_1}\dd
t_1\,\mathcal{Z}_{[t_2,s_2]}\mathcal{V}_{[t_1,t_2]}(1)\mathcal{Z}_{[-\beta/2,t_1
]}\,.\label{eq: the other equation}
\end{align}
It then follows from Theorem~\ref{thm: equilibrium rte}, the bounds  \eqref{eq: bound on cavcas}
and $\norm (\caZ_I)_{0} \norm \leq C $ that
\begin{align}\label{eq:7.70}
(\mathcal{G}_{[0,s_2]}(1))_0=\left( \mathcal{Z}_{[s_1,s_2]}\caS_1\right)_0  P^0
+ \int_{s_1}^{s_2}\dd t_2\,\int_{0}^{s_1}\dd
t_1\,\left(\mathcal{Z}_{[t_2,s_2]}\mathcal{V}_{[t_1,t_2]}(1)\right)_0  P^0   +
\caO(\e{-g\lambda^2 s_1})\,.
\end{align}
If we replace $\caG_{[0,s_2]}(1)$ by $\caG_{[-\beta/2,s_2]}(1)$, we get a similar identity, except that $ P^0$ 
is replaced by $ P^0 \caY(z=0)$. However, once~\eqref{eq:7.70} is applied to
$(\eta_{\beta})_0$, we can use~\eqref{eq: normalization of Y}, to conclude
\begin{equation*}
\left( \caG_{[0,s_2]}(\caS_1^{s_1})
\rho_\sys\right)_0=\left(\caG_{[-\beta/2,s_2]}(\caS_1^{s_1}) \eta_{\be}
\right)_0 + \caO(\e{- \la^2 g s_1})\,.
\end{equation*}
 This proves~\eqref{eq:2.12} for $m=2$.

In order to prove the `cluster property' of the correlation function, i.e.,
~\eqref{eq:2.13}, we consider the limit $s_2-s_1\to\infty$ in~\eqref{eq: the
other equation}:
\begin{align*}
\left(\mathcal{G}_{[-\beta/2,s_2]}(1)\right)_0&=P^0\left(\caS_1\mathcal{Z}_{[
-\beta/2,s_1]}\right)_0+\int_{s_1}^{s_2}\dd t_2\,\int_0^{s_1}\dd t_1\, 
P^0\left(\mathcal{V}_{[t_1,t_2]}(1)\,\mathcal{Z}_{[-\beta/2,t_1]}
\right)_0+\caO(\e{-\lambda^2g(s_1-s_2)})\nonumber\\
&=P^0\left(
\caS_1^{s_1}\mathcal{Z}_{[-\beta/2,s_1]}\right)_0+\caO(\e{-\lambda^2g
(s_1-s_2)})\,,\quad\ s_2-s_1\to\infty\,,
\end{align*}
where we have used the `sum rule'~\eqref{eq: ward identity} in the second line.
Applying the above equation to $(\eta_\be)_0$ yields
\begin{align*}
\Tr_\sys \left[\caS_2^{s_2} \caG_{[-\beta/2,s_2]}( \caS_1^{s_1})\eta_\be \right]
& =   \langle1, ( \caS_2)_0 \zeta^{0} \rangle    \langle 1,
(\caS_1\caZ_{[-\beta/2,s_1]} \eta_\be)_0 \rangle +\caO(\e{-\lambda^2g
(s_1-s_2)})  \\
 & =   \langle1,  (\caS_2)_0 \zeta^{0} \rangle    \langle 1, \caS_1\zeta^0
\rangle +\caO(\e{-\lambda^2g (s_1-s_2)}) \,,
\end{align*}
where we used the stationarity to get the second line. This proves the desired
cluster property for $m=2$.

\emph{The cases $m>2$:} Straightforward generalization of the above
arguments.

\end{proof}

\section{Proof of Theorem~\ref{thm: einstein}}

We proceed to proving the Einstein relation. Recall the definition of the velocity operator
\begin{equation} \label{eq: repeat finite volume velocity}
V^j \deq\ii [ H,X^j]   = \ii [T,X^j]\,.
\end{equation}
In a finite volume $\Lambda$, the derivative with respect to $\field$ of $\langle V^{\Lambda}(t)\rangle_{\rho_\beta^{\Lambda}}$ can be computed using Duhamel's principle: 
\begin{align}\label{duhamel}
\frac{\partial}{\partial \field^i}\bigg|_{\field=0}\langle
V^{\Lambda,j}(t)\rangle_{\rho_\be^{\Lambda}}=&
-\ii\lambda^2\int_0^t\,\dd s\,\langle
[X^{\Lambda,i}(t-s),V^{\Lambda,j}(t)]\rangle_{\rho_\beta}\,.
\end{align}
Here, it is well-understood that time-evolution on the right side is at $\field=0$, whereas on the left side the force field is set to 0 only after the differentiation. For simplicity, we drop the spatial indices~$i,j$ in the following.

 Using the time-translation invariance of the state $\rho_\beta^{\Lambda}$, the KMS condition~\eqref{eq: finite volume kms} and the time-reversal invariance of the model at vanishing $\field$, one finds that
\begin{equation}\label{finite volume einstein}
\frac{\partial}{\partial \field}\bigg|_{\field=0}\langle
V^{\Lambda}(t)\rangle_{\rho_\be^{\Lambda}}=
-\ii   \frac{ \lambda^2 \beta}{2}\int_{-t}^{t}\,\dd s\,\langle
V^{\Lambda}V^{\Lambda}(s)\rangle_{\rho_\beta^{\Lambda}} +Q^{\Lambda}(t)\,,
\end{equation}
where
\begin{align*}
Q^{\Lambda}(t)\deq\ii\frac{1}{2}\int_{0}^{\beta}\,\dd u\,\int_0^u\,\dd s\,\langle V^{\Lambda}V^{\Lambda}(\ii
s+t)\rangle_{\rho_\beta^{\Lambda}}-\ii\frac{1}{2}\int_{0}^{\beta}\,\dd
u\,\int_0^u\,\dd s\,\langle V^{\Lambda}V^{\Lambda}(\ii s-t)\rangle_{\rho_\beta^{\Lambda}}\,.
\end{align*}
Again, it is understood that time-evolution on the right side of~\eqref{finite volume einstein} is taken for a vanishing external field. For a detailed derivation of Equation~\eqref{finite volume einstein}, we refer to Section 4.2 in~\cite{paper1}. Note that the correlation functions on the right side involve imaginary times. Finally, we observe that, by the discussion following Lemma~\ref{lem: thermo on kernels},~\eqref{finite volume einstein} holds in infinite volume as well.

To complete the proof of Theorem~\ref{thm: einstein}, we show that, in the thermodynamic limit, $Q(t)\to0$, as $t\to\infty$.

We first establish a lemma concerning the behaviour of
correlation functions continued to imaginary times. 
\begin{lemma}\label{lem: decay interior}  The equilibrium correlation functions
satisfy the bound
\begin{equation}
\str \langle  O_2(t_2)  O_1(t_1)  \rangle_{\rho_\beta} \str \leq  C \norm O_2
\norm \norm O_1 \norm   \e{-\la^2 g \str \re (t_2-t_1)\str}\,,
\end{equation}
for $t_1- t_2 \in \bbH_\be$ and $O_1,O_2 \in \mathfrak{A}_{\mathrm{ti}}$.
\end{lemma}
\begin{proof}
Recall the finite volume approximations  $\langle  O^{\La}_2(t_2) 
O^{\La}_1(t_1)  \rangle_{\rho^\La_\beta}$. By Proposition~5.3.7. in
\cite{bratellirobinson},
\begin{equation} \label{eq: bound strip by exterior}
\str \langle  O^{\La}_2(t_2)  O^{\La}_1(t_1)  \rangle_{\rho^\La_\beta} \str 
\leq   \norm O^{\La}_2 \norm \norm O^{\La}_1 \norm \,,\quad\quad t_1-t_2\in\mathbb{H}_{\beta}\,.
\end{equation}
Since $\norm O^{\La} \norm \leq \norm O \norm$, we conclude that the
correlation functions in infinite volume satisfy~\eqref{eq: bound strip by
exterior}, too. We set $t_2=0$ and define 
\begin{equation*}
f_a(t) \deq    \e{\la^2 g t-at^2}     \langle  O_2(0)  O_1(t) 
\rangle_{\rho_\beta} \,,\quad\quad t\in\mathbb{H}_{\beta}\,,
\end{equation*}
with the decay rate $\la^2 g>0$ as in Theorem~\ref{thm: equilibrium rte}, and
$a>0$. From Theorem~\ref{thm: equilibrium rte} and the KMS condition we infer that
\begin{equation*}
 \sup_{a>0} \sup_{t \in \partial H_\be}\str f_a(t) \str   <\infty\,.
\end{equation*}
  Furthermore,  by the infinite-volume version of \eqref{eq: bound strip by
exterior}, $f_a$ is bounded on the whole strip $\bbH_\be$, and the KMS condition
implies that it is continuous on $\bbH_\be$ and  analytic in the interior. 
    Therefore, the maximum principle for infinite domains (the
Phragmen-Lindel{\"o}f theorem) yields
    \begin{equation*}
     \sup_{a>0} \sup_{t \in H_\be}\str f_a(t) \str \leq       \sup_{a>0} \sup_{t
\in  \partial H_\be}\str f_a(t) \str   <\infty\,.
    \end{equation*}
By varying $a$, we deduce that  $ \str\langle  O_2(0)  O_1(t) 
\rangle_{\rho_\beta} \str  \leq C \e{-\la^2 g \str\re t \str}$.  By
time-translation invariance, the claim of the lemma follows.
\end{proof}
\begin{proof}[Proof of Theorem~\ref{thm: einstein}]
Recall the definition of the operator $S$ in~\eqref{eq: definition of S}.
Starting from the bounds in Subsection~\ref{sec: bounds unbounded correlations}, Lemma~\ref{lem: prop for asymptotic} says that $\field \mapsto S(z,\field)f$ is
$C^{\infty}$, for $f \in \caD$, ($\caD$ the set of real-analytic
functions on $\bbT^d$).  As pointed out in Lemma~\ref{lem: prop for asymptotic},  this implies that the function
\begin{equation*}
(\caR(z))_{\la^2 \ka}f  = \frac{1}{z-S(z,\la,\ka,\field)}f\,,\quad\quad f\in\caD\,,
\end{equation*}
is also $C^{\infty}$ in $\field$. 
Inspecting the reasoning in Subsection~\ref{sec: bounds unbounded correlations} leading to \eqref {eq: smoothness of cav field} and restricting to fibers, it is clear that 
 $\field \mapsto (\caY(z))_{\la^2 \ka}f$, for $(z,\field) $ in a neighborhood of $(0,0)$,  is $C^\infty$, with all derivatives uniformly bounded on compacts.  Hence, by \eqref{eq: relation car and carbe}, we get smoothness of 
$\field\mapsto (\caR_\be(z))_{\la^2 \ka}f$.  Therefore, for exponentially
localized $\rho_\sys$ (in particular $\eta_\be$), the function
\begin{equation*}
(z, \field) \mapsto  z  \langle \nabla\varepsilon, (\caR_\be(z) \rho_\sys)_0 
\rangle\,,
\end{equation*}
is analytic in $z$ and $C^{\infty}$ in $\field$ for $(z,\field)$ in a
neighborhood of $(0,0)$, with the corresponding derivatives uniformly bounded on
compacts. 

Next, starting from the identity $v(\field)=\langle
\nabla\epsilon,\zeta^0\rangle=\lim_{z\to
0}z\langle\nabla\epsilon,(\caR_{\beta}(z))_0(\eta_\beta)_0\rangle$, we obtain,
using the above regularity properties, that
\begin{align*}
 \frac{\partial}{\partial \field}\bigg|_{\field=0}v(\field)=\lim_{z\to
0}z\frac{\partial}{\partial\field}\bigg|_{\field=0}\langle\nabla\epsilon,(\caR_{
\beta}(z))_0(\eta_\beta)_0\rangle\,.
\end{align*}
Hence, setting $z=1/T$, using the definition of $\caR_{\beta}$ and the
bounds $|\langle V^{\field}(t)\rangle_{\rho_\beta}|\le C$ and
$|\partial_\field\langle V^{\field}(t)\rangle_{\rho_\beta}|\le \lambda^2 Ct$,
for $\field$ sufficiently small, as follows from~\eqref{eq: examplebound1}, we
get
\begin{align*}
 \frac{\partial}{\partial
\field}\bigg|_{\field=0}v(\field)&=\lim_{T\to\infty}\frac{1}{T}\int_{0}^{\infty}
\dd t\,\e{-\frac{t}{T}}\frac{\partial}{\partial \field}\bigg|_{\field=0}\langle
V^\field(t)\rangle_{\rho_\beta}\\[2mm]
&=\frac{\beta\lambda^2}{2}\lim_{T\to\infty}\frac{1}{T}\int_{0}^{\infty}\dd
t\,\e{-\frac{t}{T}}\,\left(\int_{-t}^{t}\dd s\,\langle
VV(s)\rangle_{\rho_\beta}+Q(t)\right) \\[2mm]
&= \frac{\beta\lambda^2}{2} \int_{\R}\dd s\,\langle VV(s)\rangle_{\rho_\beta}\,.
\end{align*}
The second equality is \eqref{finite volume einstein} (in the thermodynamic limit), and the third equality follows because,
by Lemma~\ref{lem: decay interior}, $Q(t)\to0$, as $t\to \infty$, and $\langle
VV(s)\rangle_{\rho_\beta}=\caO(\e{-\lambda^2 g s})$, by
Theorem~\ref{thm:correlationfunctions}.

To complete the derivation of the Einstein relation we have to argue that the
equilibrium auto-correlation function $\int_{\R}\dd s\,\langle
VV(s)\rangle_{\rho_\beta}$ is indeed equal to the diffusion constant $D$. This
is checked by using the exponential decay of the correlation
function and the identity
\begin{align*}
\langle X^i(t)X^j(t)\rangle_{\rho_{\beta}}&=\int_0^t\dd s_1\,\int_0^t\dd
s_2\,\langle V^i(s_1)V^j(s_2)\rangle_{\rho_{\beta}}\,,\nonumber
\end{align*}
 cf., \eqref{eq: manipulations allowed}.
\end{proof}

\bibliographystyle{plain}
\bibliography{biblio}

\end{document}